\documentclass[letterpaper,11pt]{article}

\usepackage[letterpaper,hmargin=3cm,vmargin=2.8cm]{geometry}

\usepackage{amsmath,amssymb,amsthm,dsfont,mathrsfs,mathtools}

\usepackage{tikz}

\usepackage[colorlinks=true,linkcolor=black,citecolor=black]{hyperref}

\newtheorem{thm}{Theorem}
\newtheorem{lem}{Lemma}
\newtheorem{prop}{Proposition}
{\theoremstyle{definition} \newtheorem{hyp}{Hypothesis}}

\newcommand{\Z}{\mathds Z}

\newcommand{\R}{\mathds R}
\newcommand{\C}{\mathds C}

\newcommand{\la}{\langle}
\newcommand{\ra}{\rangle}

\newcommand{\F}{\mathcal{F}}
\newcommand{\G}{\mathcal{G}}
\newcommand{\Fhat}{\widehat{\mathcal{F}}}
\newcommand{\Gdual}{\Gamma^\#}
\newcommand{\eps}{\varepsilon}

\title{Asymptotics for Fermi curves: small magnetic potential}

\author{Gustavo de Oliveira \\
{\small Department of Mathematics,} \\
{\small University of British Columbia, Canada} \\
{\small goliveira5d@gmail.com}}

\date{March 1, 2010}

\begin{document}

\maketitle

\begin{abstract}
We consider complex Fermi curves of electric and magnetic periodic
fields. These are analytic curves in~$\C^2$ that arise from the study of
the eigenvalue problem for periodic Schr\"odinger operators. We
characterize a certain class of these curves in the region of~$\C^2$
where at least one of the coordinates has ``large'' imaginary part. The
new results in this work extend previous results in the absence of
magnetic field to the case of ``small'' magnetic field. Our theorems can
be used to show that generically these Fermi curves belong to a class of
Riemann surfaces of infinite genus.
\end{abstract}


\section{Introduction}

In~\cite{FKT2}, the authors introduced a class of Riemann surfaces of
infinite genus that are ``asymptotic to'' a finite number of complex
lines joined by infinite many handles. These surfaces are constructed by
pasting together a compact submanifold of finite genus, plane domains,
and handles. All these components satisfy a number of geometric/analytic
hypotheses stated in~\cite{FKT2} that specify the asymptotic holomorphic
structure of the surface. The class of surfaces obtained in this way
yields an extension of the classical theory of compact Riemann surfaces
that has analogues of many theorems of the classical theory. It was
proven in~\cite{FKT2} that this new class includes quite general
hyperelliptic surfaces, heat curves (which are spectral curves
associated to a certain ``heat-equation''), and Fermi curves with zero
magnetic potential. In order to verify the geometric/analytic hypotheses
for the latter the authors proved two ``asymptotic'' theorems similar to
the ones we prove below. This is the main step needed to verify these
hypotheses. In this work we extend their results to Fermi curves with
``small'' magnetic potential.

There are two immediate applications of our results. First, as we have
already mentioned, one can use our theorems for verifying the
geometric/analytic hypotheses of~\cite{FKT2} for Fermi curves with small
magnetic potential. This would show that these curves belong to the
class of Riemann surfaces mentioned above. Secondly, one can prove that
a class of these curves are irreducible (in the usual
algebraic-geometrical sense). Both these applications were done in
\cite{FKT2}~for Fermi curves with zero magnetic potential.

Complex Fermi curves (and other similar spectral curves) have been
studied, in different perspectives, in the absence of magnetic field
\cite{FKT2, GKT, KT, Kr, Mc}, and in the presence of magnetic
field~\cite{FKT1}. Some results on the real Fermi curve in the
high-energy region were obtained in~\cite{Ka}. There one also finds a
short description of the existing results on periodic magnetic
Schr\"odinger operators. An even more general review is presented
in~\cite{E}. To our knowledge our work provides new results on complex
Fermi curves with magnetic field. At this moment we are only able to
handle the case of ``small'' magnetic potential. The asymptotic
characterization of Fermi curves with arbitrarily large magnetic
potential remains as an open problem. In order to prove our theorems we
follow the same strategy as~\cite{FKT2}. The presence of magnetic field
makes the analysis considerably harder and requires new estimates. As it
was pointed out in~\cite{Ka, E}, the study of an operator with magnetic
potential is essentially more complicated than the study of the operator
with just an electric potential. This seems to be the case in this
problem as well.

Before we outline our results let us introduce some definitions.
Let~$\Gamma$ be a lattice in~$\R^2$ and let~$A_1$, $A_2$ and~$V$ be
real-valued functions in~$L^2(\R^2)$ that are periodic with respect to
$\Gamma$. Set~$A \coloneqq (A_1,A_2)$ and define the operator
$$ H(A,V) \coloneqq (i \nabla + A)^2 + V $$
acting on $L^2(\R^2)$, where~$\nabla$ is the gradient operator
in~$\R^2$. For $k \in \R^2$ consider the following
eigenvalue-eigenvector problem in~$L^2(\R^2)$ with boundary conditions,
\begin{align*}
H(A,V) \varphi & = \lambda \varphi, \\
\varphi(x+\gamma) & = e^{i k \cdot \gamma} \varphi(x)
\end{align*}
for all $x \in \R^2$ and all $\gamma \in \Gamma$. Under suitable
hypotheses on the potentials $A$ and $V$ this problem is self-adjoint
and its spectrum is discrete. It consists of a sequence of real
eigenvalues
$$ E_1(k,A,V) \leq E_2(k,A,V) \leq \cdots \leq E_n(k,A,V) \leq \cdots $$
For each integer $n \geq 1$ the eigenvalue $E_n(k,A,V)$ defines a
continuous function of $k$. From the above boundary condition it is easy
to see that this function is periodic with respect to the dual lattice
$$ \Gdual \coloneqq \{ b \in \R^2 \; | \; b \cdot \gamma \in 2\pi \Z \,
\text{ for all } \gamma \in \Gamma \}, $$
where $b \cdot \gamma$ is the usual scalar product on~$\R^2$. It is
customary to refer to $k$ as the crystal momentum and to $E_n(k,A,V)$ as
the $n$-th band function. The corresponding normalized eigenfunctions
$\varphi_{n,k}$ are called Bloch eigenfunctions.

The operator $H(A,V)$ (and its three-dimensional counterpart) is
important in solid state physics. It is the Hamiltonian of a single
electron under the influence of magnetic field with vector
potential~$A$, and electric field with scalar potential~$V$, in the
independent electron model of a two-dimensional solid~\cite{RS4}. The
classical framework for studying the spectrum of a differential operator
with periodic coefficients is the Floquet (or Bloch) theory \cite{RS4,
K, MW}. Roughly speaking, the main idea of this theory is to
``decompose'' the original eigenvalue problem, which usually has
continuous spectrum, into a family of boundary value problems, each one
having discrete spectrum. In our context this leads to decomposing the
problem $H(A,V) \varphi = \lambda \varphi$ (without boundary conditions)
into the above $k$-family of boundary value problems.

Let $U_k$ be the unitary transformation on~$L^2(\R^2)$ that acts as
$$ U_k \, : \, \varphi(x) \longmapsto e^{ik \cdot x} \varphi(x). $$
By applying this transformation we can rewrite the above problem and put
the boundary conditions into the operator. Indeed, if we define
$$ H_k(A,V) \coloneqq U_k^{-1} \, H(A,V) \, U_k \qquad \text{and} \qquad
\psi \coloneqq U_k^{-1} \varphi, $$
then the above problem is unitarily equivalent to
$$ H_k(A,V) \psi = \lambda \psi \qquad \text{for} \qquad \psi \in
L^2(\R^2 / \Gamma). $$
Furthermore, a simple (formal) calculation shows that
$$ H_k(A,V) = (i \nabla + A - k)^2 + V. $$

The real ``lifted'' Fermi curve of $(A,V)$ with energy $\lambda \in \R$
is defined as
$$ \Fhat_{\lambda, \R}(A,V) \coloneqq \{ k \in \R^2 \; | \; (H_k(A,V) -
\lambda) \varphi = 0 \, \text{ for some } \varphi \in
\mathcal{D}_{H_k(A,V)} \setminus \{0\} \}, $$
where $\mathcal{D}_{H_k(A,V)} \subset L^2(\R^2/\Gamma)$ denotes the
(dense) domain of $H_k(A,V)$. The adjective ``lifted'' indicates that
$\widehat{\mathcal{F}}_{\lambda, \R}(A,V)$ is a subset of $\R^2$ rather
than $\R^2 / \Gdual$. As we may replace $V$ by $V-\lambda$, we only
discuss the case $\lambda=0$ and write $\Fhat_\R(A,V)$ in place of
$\Fhat_{0,\R}(A,V)$ to simplify the notation. Let $|\Gamma| \coloneqq
\int_{\R^2 / \Gamma} dx$ and $\hat{A}(0) \coloneqq |\Gamma|^{-1}
\int_{\R^2 / \Gamma} A(x) \, dx$. Since $H_k(A,V)$ is equal to
$H_{k-\hat{A}(0)}(A-\hat{A}(0),V)$, if we perform the change of
coordinates $k \to k + \hat{A}(0)$ and redefine $A-\hat{A}(0) \to A$ we
may assume, without loss of generality, that $\hat{A}(0)=0$. The dual
lattice $\Gdual$ acts on $\R^2$ by translating $k \mapsto k+b$ for $b
\in \Gdual$. This action maps $\Fhat_\R(A,V)$ to itself because for each
$n \geq 1$ the function $k \mapsto E_n(k,A,V)$ is periodic with respect
to $\Gdual$. In other words, the real lifted Fermi curve ``is periodic''
with respect to $\Gdual$. Define
$$ \F_\R(A,V) \coloneqq \Fhat_\R(A,V) / \Gdual. $$
We call $\F_\R(A,V)$ the real Fermi curve of $(A,V)$. It is a curve in
the torus $\R^2/\Gdual$.

The above definitions and the real Fermi curve have physical meaning. It
is useful and interesting, however, to study the ``complexification'' of
these curves. Knowledge about the complexified curves may provide
information about the real counterparts. For complex-valued functions
$A_1$, $A_2$ and $V$ in $L^2(\R^2)$ and for $k \in \C^2$ the above
problem is no longer self-adjoint. Its spectrum, however, remains
discrete. It is a sequence of eigenvalues in the complex plane. From the
boundary condition in the original problem it is easy to see that the
family of functions $k \mapsto E_n(k,A,V)$ remains periodic with respect
to $\Gdual$. Moreover, the transformation $U_k$ is no longer unitary but
it is still bounded and invertible and it still preserves the spectrum,
that is, we can still rewrite the original problem in the form $H_k(A,V)
\psi = \lambda \psi$ for $\psi \in L^2(\R^2/\Gamma)$ without modifying
the eigenvalues. Thus, it makes sense to define
\begin{align*} \Fhat(A,V) & \coloneqq \{ k \in \C^2 \; | \; H_k(A,V)
\varphi = 0 \, \text{ for some } \varphi \in \mathcal{D}_{H_k(A,V)}
\setminus \{0\} \}, \\ \F(A,V) & \coloneqq \Fhat(A,V) / \Gdual.
\end{align*}
We call $\Fhat(A,V)$ and $\F(A,V)$ the complex ``lifted'' Fermi curve
and the complex Fermi curve, respectively. When there is no risk of
confusion we refer to either simply as Fermi curve.

We are now ready to outline our results. When $A$ and $V$ are zero the
(free) Fermi curve can be found explicitly. It consists of two copies of
$\mathds{C}$ with the points $-b_2 + i b_1$ (in the first copy) and $b_2
+ i b_1$ (in the second copy) identified for all $(b_1,b_2) \in \Gdual$
with $b_2 \neq 0$. In this work we prove that in the region of $\C^2$
where $k \in \C^2$ has ``large'' imaginary part the Fermi curve (for
nonzero $A$ and $V$) is ``close to'' the free Fermi curve. In a compact
form, our main result (that will be stated precisely in Theorems
\ref{t:reg} and \ref{t:hand}) is essentially the following.

\vspace{0.3cm}

\noindent {\bf Main result.} \emph{ Suppose that $A$ and $V$ have some
regularity and assume that (in a suitable norm) $A$ is smaller than a
constant given by the parameters of the problem. Write $k$ in
$\mathds{C}^2$ as $k = u + iv$ with $u$ and $v$ in $\mathds{R}^2$ and
suppose that $|v|$ is larger than a constant given by the parameters of
the problem. (Recall that the free Fermi curve is two copies of
$\mathds{C}$ with certain points in one copy identified with points in
the other one.) Then, in this region of $\mathds{C}^2$, the Fermi curve
of $A$ and $V$ is very close to the free Fermi curve, except that
instead of two planes we may have two deformed planes, and
identifications between points can open up to handles that look like $\{
(z_1,z_2) \in \mathds{C}^2 \; | \; z_1 z_2 = \text{constant} \}$ in
suitable local coordinates.}

\vspace{0.3cm}

The proof of our results has basically three steps:

\begin{itemize}

\item We first derive very detailed information about the free Fermi
curve (which is explicitly known). Then, to compute the interacting
Fermi curve we have to find the kernel of $H$ in $L^2(\R^2)$ with the
above boundary conditions.

\item In the second step of the proof we derive a number of estimates
for showing that this kernel has finite dimension for small $A$ and $k
\in \C^2$ with large imaginary part. Our strategy here is similar to the
Feshbach method in perturbation theory \cite{GS}.  Indeed, we prove that
in the complement of the kernel of $H$ in $L^2(\R^2)$, after a suitable
invertible change of variables in $L^2(\R^2)$, the operator $H$
multiplied by the inverse of the operator that implements this change of
variables is a compact perturbation of the identity that is invertible
for such $A$ and $k$. This reduces the problem of finding the kernel to
finite dimension and thus we can write local defining equations for the
Fermi curve.

\item In the third step of the proof we use these equations to study the
Fermi curve. A few more estimates and the implicit function theorem
gives us the deformed planes. The handles are obtained using a
quantitative Morse lemma from \cite{deO} that is available in the
Appendix \ref{s:morse}.

\end{itemize}

Steps two and three contain most of the novelties in this work. The
critical part of the proof is the second step. The main difficulty
arises due to the presence of the term $A \cdot i \nabla$ in the
Hamiltonian $H(A,V)$. When $A$ is large, taking the imaginary part of $k
\in \C^2$ arbitrarily large is not enough to control this term---it is
not enough to make its contribution small and hence have the interacting
Fermi curve as a perturbation of the free Fermi curve. (The term $V$ in
$H(A,V)$ is easily controlled by this method.) However, the proof can be
implemented by assuming that $A$ is small.

This work is organized as follows. In \S\ref{s:free} we collect some
properties of the free Fermi curve and in \S\ref{s:tubes} we define
$\eps$-tubes about it. In \S\ref{s:main} we state our main results and
in \S\ref{s:idea} we describe the general strategy of analysis used to
prove them.  Subsequently, we implement this strategy by proving a
number of lemmas and propositions in \S\ref{s:inv} to \S\ref{s:der},
which we put together later in \S\ref{s:reg} and \S\ref{s:hand} to prove
our main theorems. The proof of the estimates of \S\ref{s:coeff} and
\S\ref{s:der} are left to the Appendices \ref{s:app2} and \ref{s:app3}.

\vspace{0.3cm}

{\bf Acknowledgments.} I would like to thank Professor Joel Feldman for
suggesting this problem and for the many discussions I have had with him. I
am also grateful to Alessandro Michelangeli for useful comments about the
manuscript. This work is part of the author's Ph.D. thesis \cite{deO}
defended at the University of British Columbia in Vancouver, Canada.


\section{The free Fermi curve} \label{s:free}

When the potentials $A$ and $V$ are zero the curve $\Fhat(A,V)$ can be
found explicitly. In this section we collect some properties of this
curve. For $\nu \in \{1,2\}$ and $b \in \Gdual$ set
\begin{align*}
N_{b,\nu}(k) & \coloneqq (k_1 + b_1) + i(-1)^\nu(k_2 + b_2), \\
\mathcal{N}_\nu(b) & \coloneqq \{ k \in \C^2 \; | \; N_{b,\nu}(k)=0\},\\
N_b(k) & \coloneqq N_{b,1}(k) N_{b,2}(k), \\
\mathcal{N}_b & \coloneqq \mathcal{N}_1(b) \cup \mathcal{N}_2(b), \\
\theta_\nu(b) & \coloneqq \tfrac{1}{2} ((-1)^\nu b_2 + i b_1).
\end{align*}
Observe that $\mathcal{N}_\nu(b)$ is a line in $\C^2$. The free lifted
Fermi curve is an union of these lines. Here is the precise statement.

\begin{prop}[The free Fermi curve] \label{p:free}
The curve $\Fhat(0,0)$ is the locally finite union
$$ \bigcup_{b\in\Gdual} \bigcup_{\nu\in\{1,2\}} \mathcal{N}_\nu(b).$$
In particular, the curve $\F(0,0)$ is a complex analytic curve in $\C^2
/ \Gdual$.
\end{prop}

The proof of this proposition is straightforward. It can be found in
\cite{deO}. Here we only give its first part.

\begin{proof}[Proof of Proposition \ref{p:free} (first part)]
For all $k \in \C^2$ the functions $\{ e^{ib \cdot x} \; | \; b \in
\Gdual \}$ form a complete set of eigenfunctions for $H_k(0,0)$ in
$L^2(\R^2/\Gamma)$ satisfying
$$ H_k(0,0) e^{ib \cdot x} = (i\nabla-k)^2 e^{ib \cdot x} = (b+k)^2
e^{ib \cdot x} = N_b(k) e^{ib \cdot x}. $$
Hence, 
$$ \Fhat(0,0) = \{ k \in \C^2 \; | \; N_b(k)=0 \, \text{ for some } b
\in \Gdual \} = \bigcup_{b \in \Gdual} \mathcal{N}_b = \bigcup_{b \in
\Gdual} \bigcup_{\nu \in \{1,2\}} \mathcal{N}_\nu(b). $$
This is the desired expression for $\Fhat(0,0)$.
\end{proof}
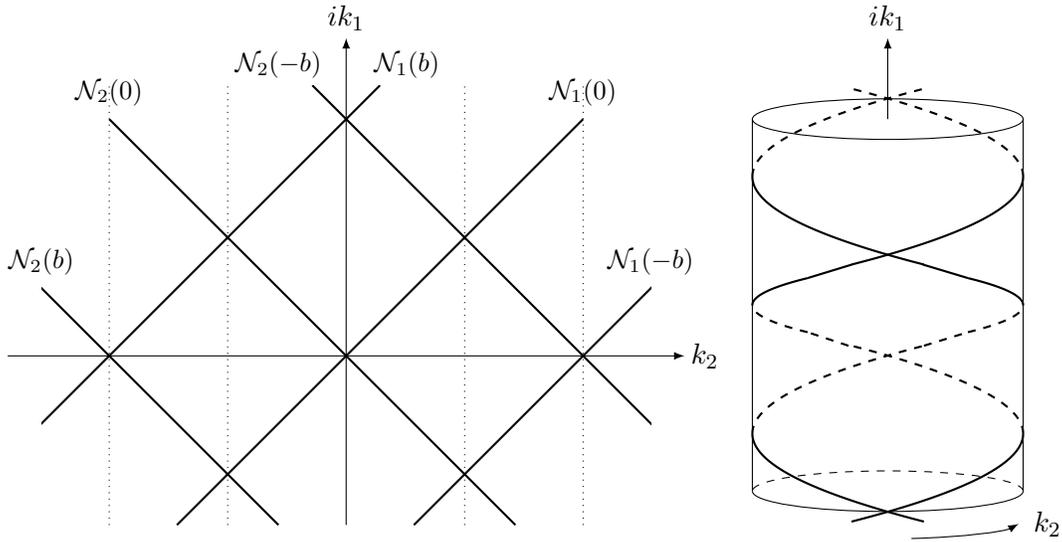
\begin{figure}[htb]
\begin{center}
\begin{tikzpicture}[scale=0.9]
\draw (1.5,6.4) node {{\small $\mathcal{N}_2(0)$}};
\draw (0.5,3.9) node {{\small $\mathcal{N}_2(b)$}};
\draw (9.5,3.9) node {{\small $\mathcal{N}_1(-b)$}};
\draw (8.5,6.4) node {{\small $\mathcal{N}_1(0)$}};
\draw (5.9,6.8) node {{\small $\mathcal{N}_1(b)$}};
\draw (4,6.8) node {{\small $\mathcal{N}_2(-b)$}};
\draw[-,dashed] (15,0.5) arc (0:180:2cm and 0.3cm);
\draw (15,0.5) arc (0:-180:2cm and 0.3cm);
\draw[-latex] (13.35,-0.2) arc (-80:-15: 2cm and 0.3cm) node {$\qquad k_2$};
\draw (13,6) ellipse (2 and 0.3);
\draw[-latex] (13,6)--(13,7.2);
\draw (13,7.5) node {$i k_1$};
\draw (15,0.5)--(15,6) (11,0.5)--(11,6);
\draw[-latex] (0,2.5)--(10,2.5);
\draw (10.3,2.5) node {$k_2$};
\draw[-latex] (5,0)--(5,7.2); 
\draw (5,7.5) node {$i k_1$};
\draw[thick,dashed,domain=0:1.3,smooth,variable=\t] 
plot(15-1.5*\t*\t,1.35+\t);
\draw[thick,domain=-1.3:0,smooth,variable=\t] 
plot(15-1.5*\t*\t,1.35+\t);
\draw[thick,dashed,domain=0:1.3,smooth,variable=\t] 
plot(11+1.5*\t*\t,1.35+\t);
\draw[thick,domain=-1.3:0,smooth,variable=\t] 
plot(11+1.5*\t*\t,1.35+\t);
\draw[thick,dashed,domain=0:1.3,smooth,variable=\t] 
plot(15-1.5*\t*\t,5.15+\t);
\draw[thick,domain=-1.3:0,smooth,variable=\t] 
plot(15-1.5*\t*\t,5.15+\t);
\draw[thick,dashed,domain=0:1.3,smooth,variable=\t] 
plot(11+1.5*\t*\t,5.15+\t);
\draw[thick,domain=-1.3:0,smooth,variable=\t] 
plot(11+1.5*\t*\t,5.15+\t);
\draw[thick,domain=0:0.4,smooth,variable=\t] 
plot(15-5.1*\t*\t,3.25+\t);
\draw[thick,dashed,domain=-0.4:0,smooth,variable=\t] 
plot(15-5.1*\t*\t,3.25+\t);
\draw[thick,domain=0:0.4,smooth,variable=\t] 
plot(11+5.1*\t*\t,3.25+\t);
\draw[thick,dashed,domain=-0.4:0,smooth,variable=\t] 
plot(11+5.1*\t*\t,3.25+\t);
\draw[thick,dashed] (12.42,2.65)--(11.816,2.85);
\draw[thick,dashed] (13.535,2.65)--(14.184,2.85);
\draw[thick] (14.184,3.65)--(13.52,3.85);
\draw[thick] (11.816,3.65)--(12.47,3.85);
\clip (0.5,0) rectangle (9.5,6.5);
\draw[-,dotted]
(1.5,0)--(1.5,6.5) (3.25,0)--(3.25,6.5)  (6.75,0)--(6.75,6.5)
(8.5,0)--(8.5,6.5);
\draw[-,thick]
(6,0)--(10.5,4.5) (4,0)--(0,4)  (2.5,0)--(8.5,6)
(7.5,0)--(1.5,6) (10.5,0.5)--(4.5,6.5)  (0,1)--(5.5,6.5);
\end{tikzpicture}
\end{center}
\caption{Sketch of $\Fhat(0,0)$ and $\F(0,0)$ when both $ik_1$ and $k_2$
are real.}
\label{fig:free}
\end{figure}

The lines $\mathcal{N}_\nu(b)$ have the following properties (see
\cite{deO} for a proof).

\begin{prop}[Properties of $\mathcal{N}_\nu(b)$] \label{p:lines}
Let $\nu \in \{1,2\}$ and let $b,c,d \in \Gdual$. Then:
\begin{itemize}
\item[\rm (a)] $\mathcal{N}_\nu(b) \cap \mathcal{N}_\nu(c) = \varnothing
\quad \text{if} \quad b \neq c \,$;

\item[\rm (b)] $\text{\rm dist}(\mathcal{N}_\nu(b), \mathcal{N}_\nu(c))
= \tfrac{1}{\sqrt{2}} |b-c|$;

\item[\rm (c)] $\mathcal{N}_1(b) \cap \mathcal{N}_2(c) = \{ ( i
\theta_1(c) + i \theta_2(b), \, \theta_1(c) - \theta_2(b) ) \}$; 

\item[\rm (d)] the map $k \mapsto k+d$  maps $\mathcal{N}_\nu(b)$ to
$\mathcal{N}_\nu(b-d)$;

\item[\rm (e)] the map $k \mapsto k+d$ maps $\mathcal{N}_1(b) \cap
\mathcal{N}_2(c)$ to $\mathcal{N}_1(b-d) \cap \mathcal{N}_2(c-d)$.
\end{itemize}
\end{prop}

Let us briefly describe what the free Fermi curve looks like. In the
Figure \ref{fig:free} there is a sketch of the set of $(k_1, k_2) \in
\Fhat(0,0)$ for which both $ik_1$ and $k_2$ are real, for the case where
the lattice $\Gdual$ has points over the coordinate axes, that is, it
has points of the form $(b_1,0)$ and $(0,b_2)$.  Observe that, in
particular, Proposition \ref{p:lines} yields
\begin{align*}
\mathcal{N}_1(0) \cap \mathcal{N}_2(b) & = \{ (i\theta_1(b),
\theta_1(b)) \},\\
\mathcal{N}_1(-b) \cap \mathcal{N}_2(0) & = \{ (i \theta_2(-b),
\theta_2(b))\},\\
\text{the map } k \mapsto k+b \text{ maps } & \mathcal{N}_1(0) \cap
\mathcal{N}_2(b) \text{ to } \mathcal{N}_1(-b) \cap \mathcal{N}_2(0).
\end{align*}
Recall that points in $\Fhat(0,0)$ that differ by elements of $\Gdual$
correspond to the same point in $\F(0,0)$. Thus, in the sketch on the
left, we should identify the lines $k_2 = -b_2/2$ and $k_2 = b_2/2$ for
all $b \in \Gdual$ with $b_2 \neq 0$, to get a pair of helices climbing
up the outside of a cylinder, as illustrated by the figure on the right.
The helices intersect each other twice on each cycle of the
cylinder---once on the front half of the cylinder and once on the back
half. Hence, viewed as a ``manifold'' (with singularities), the pair of
helices are just two copies of $\R$ with points that corresponds to
intersections identified. We can use $k_2$ as a coordinate in each copy
of $\R$ and then the pairs of identified points are $k_2 = b_2/2$ and
$k_2 = -b_2/2$ for all $b \in \Gdual$ with $b_2 \neq 0$. So far we have
only considered $k_2$ real. The full $\Fhat(0,0)$ is just two copies of
$\C$ with $k_2$ as a coordinate in each copy, provided we identify the
points $\theta_1(b) = \tfrac{1}{2} (-b_2 + i b_1)$ (in the first copy)
and $\theta_2(b) = \tfrac{1}{2} (b_2 + i b_1)$ (in the second copy) for
all $b \in \Gdual$ with $b_2 \neq 0$.

\section{The $\eps$-tubes about the free Fermi curve} \label{s:tubes}

We now introduce real and imaginary coordinates in $\C^2$ and define
$\eps$-tubes about the free Fermi curve. We derive some properties of
the $\eps$-tubes as well. For $k \in \C^2$ write
$$ k_1 = u_1 + i v_1 \qquad \text{and} \qquad k_2 = u_2 + iv_2, $$
where $u_1$, $u_2$, $v_1$ and $v_2$ are real numbers. Then,
\begin{align*}
N_{b,\nu}(k) & = (k_1 + b_1) + i(-1)^\nu(k_2 + b_2) \\ & = i(v_1 +
(-1)^\nu(u_2 + b_2)) -(-1)^\nu (v_2 - (-1)^\nu (u_1+b_1)),
\end{align*}
so that
$$ |N_{b,\nu}(k)| = |v + (-1)^\nu (u+b)^\perp |, $$
where $(y_1, y_2)^\perp \coloneqq (y_2, -y_1)$. Since $N_b(k) =
N_{b,1}(k) N_{b,2}(k)$, we have $N_b(k)=0$ if and only if
$$ v - (u+b)^\perp = 0 \qquad \text{or} \qquad v + (u+b)^\perp = 0. $$
Let $2 \Lambda$ be the length of the shortest nonzero ``vector'' in
$\Gdual$. Then there is at most one $b \in \Gdual$ with $|v+(u+b)^\perp|
< \Lambda$ and at most one $b \in \Gdual$ with $|v-(u+b)^\perp| <
\Lambda$ (see \cite{deO} for a proof).

Let $\eps$ be a constant satisfying $0 < \eps < \Lambda/6$. For $\nu \in
\{1,2\}$ and $b \in \Gdual$ define the $\eps$-tube about
$\mathcal{N}_{\nu}(b)$ as
$$ T_\nu(b) \coloneqq \{ k \in \C^2 \;\; | \;\; |N_{b,\nu}(k)| = |v +
(-1)^\nu (u+b)^\perp | < \eps \}, $$
and the $\eps$-tube about $\mathcal{N}_b = \mathcal{N}_1(b) \cup
\mathcal{N}_2(b)$ as
$$ T_b \coloneqq T_1(b) \cup T_2(b). $$
Since $(v+(u+b)^\perp) + (v-(u+b)^\perp) = 2v$, at least one of the
factors $|v+(u+b)^\perp|$ or $|v-(u+b)^\perp|$ in $|N_b(k)|$ must always
be greater or equal to $|v|$. If $k \not\in T_b$ both factors are also
greater or equal to $\eps$. If $k \in T_b$ one factor is bounded by
$\eps$ and the other must lie within $\eps$ of $|2v|$. Thus,
\begin{align}
k \not\in T_b \quad & \Longrightarrow \quad |N_b(k)| \geq \eps |v|,
\label{bd1} \\
k \in T_b \quad & \Longrightarrow \quad |N_b(k)| \leq \eps(2|v|+\eps).
\label{bd2}
\end{align}
Finally, denote by $\overline{T}_b$ the closure of $T_b$. The
intersection $\overline{T}_b \cap \overline{T}_{b'}$ is compact whenever
$b \neq b'$, and $\overline{T}_b \cap \overline{T}_{b'} \cap
\overline{T}_{b''}$ is empty for all distinct elements $b, b', b'' \in
\Gdual$ (see \cite{deO} for details).

If a point $k$ belongs to the free Fermi curve the function $N_b(k)$
vanishes for some $b \in \Gdual$. We now give a lower bound for this
function when $(b,k)$ is not in the zero set.

\begin{prop}[Lower bound for $|N_b(k)|$] \label{p:Nb}
$ \, $
\begin{itemize}
\item[\rm (a)] If $|b+u+v^\perp| \geq \Lambda$ and $|b+u-v^\perp| \geq
\Lambda$, then $|N_b(k)| \geq \frac{\Lambda}{2}(|v|+|u+b|)$.

\item[\rm (b)] If $|v| > 2\Lambda$ and $k \in T_0$, then $|N_b(k)| \geq
\frac{\Lambda}{2}(|v|+|u+b|)$ for all $b \neq 0$ but at most one $b \neq
0$. This exceptional $\tilde{b}$ obeys $|\tilde{b}| > |v|$ and $| \,
|u+\tilde{b}|-|v| \, | < \Lambda$.

\item[\rm (c)] If $|v| > 2\Lambda$ and $k \in T_0 \cap T_d$ with $d \neq
0$, then $|N_b(k)| \geq \frac{\Lambda}{2}(|v|+|u+b|)$ for all $b \not\in
\{0, d\}$. Furthermore we have $|d| > |v|$ and $|\, |u+d|-|v| \, | <
\Lambda$.
\end{itemize}
\end{prop}

\begin{proof}
(a) By hypothesis, both factors in $|N_b(k)| = |v + (u+b)^\perp | \, |v
- (u+b)^\perp |$ are greater or equal to $\Lambda$. We now prove that at
  least one of the factors must also be greater or equal to
$\frac{1}{2}(|v|+|u+b|)$. Suppose that $|v| \geq |u+b|$. Then, since $(v
+ (u+b)^\perp) + (v - (u+b)^\perp) = 2v$, at least one of the factors
must also be greater or equal to $|v| = \frac{1}{2}(|v|+|v|) \geq
\frac{1}{2}(|v|+|u+b|)$. Now suppose that $|v| < |u+b|$. Then similarly
we prove that $|u+b| > \frac{1}{2}(|v|+|u+b|)$.
All this together implies that $|N_b(k)| \ge
\frac{\Lambda}{2}(|v|+|u+b|)$, which proves part (a).

(b) By hypothesis $\eps < \Lambda / 6 < |v|$. Let $k \in T_{0}$. Then,
by \eqref{bd2},
\begin{equation} \label{partb}
|N_{0}(k)| \leq \eps(2|v|+\eps) < 3 \eps |v| < \frac{\Lambda}{2} |v|.
\end{equation}
Thus we have either $|u + v^\perp| < \Lambda$ or $|u - v^\perp| <
\Lambda$ (otherwise apply part (a) to get a contradiction). Suppose that
$|u + v^\perp| < \Lambda$. Then there is no $b \in \Gdual \setminus
\{0\}$ with $|b+u+v^\perp| < \Lambda$ and there is at most one
$\tilde{b} \in \Gdual \setminus \{0\}$ satisfying $|\tilde{b}+u-v^\perp|
< \Lambda$. This inequality implies $| \, |u + \tilde{b}| - |v| \, | <
\Lambda$. Furthermore, for this $\tilde{b}$,
$$ |\tilde{b}| = |2v^\perp - (u+v^\perp) + (\tilde{b} + u - v^\perp)| >
2|v| - 2\Lambda > |v|, $$
since $-2\Lambda > -|v|$. Now suppose that $|u - v^\perp| < \Lambda$.
Then similarly we prove that $|\tilde{b}| > |v|$.
%
%
%
%
Finally observe that, if $b \not\in \{0,\tilde{b}\}$ then $|b+u+v^\perp|
\geq \Lambda$ and $|b+u-v^\perp| \geq \Lambda$. Hence, applying part (a)
it follows that $|N_b(k)| \geq \frac{\Lambda}{2}(|v|+|u+b|)$. This
proves part (b).

(c) As in the proof of part (b), if $k \in T_0 \cap T_d$ then in
addition to \eqref{partb} we have $|N_{d}(k)| < \frac{\Lambda}{2} |v|$.
Thus, applying part (b) we conclude that $d$ must be the exceptional
$\tilde{b}$ of part (b). The statement of part (c) follows then from
part (b). This completes the proof.
\end{proof}


\section{Main results} \label{s:main}

The Riemann surfaces introduced in \cite{FKT2} can be decomposed into
$$ X^{\rm com} \cup X^{\rm reg} \cup X^{\rm han}, $$
where $X^{\rm com}$ is a compact submanifold with smooth boundary and
finite genus, $X^{\rm reg}$ is a finite union of open ``regular
pieces'', and $X^{\rm han}$ is an infinite union of closed ``handles''.
All these components satisfy a number of geometric/analytic hypotheses
stated in \cite{FKT2} that specify the asymptotic holomorphic structure
of the surface. Below we state two ``asymptotic'' theorems that
essentially characterize the $X^{\rm reg}$ and $X^{\rm han}$ components
of Fermi curves with small magnetic potential. Before we move to the
theorems let us introduce some definitions.

For any $\varphi \in L^2(\R^2/\Gamma)$ define $\hat{\varphi} : \Gdual
\to \C$ as
$$ \hat{\varphi}(b) \coloneqq (\F \varphi)(b) \coloneqq
\frac{1}{|\Gamma|} \int_{\R^2/\Gamma} \varphi(x) \, e^{-i b \cdot x} \,
d x, $$
where $|\Gamma| \coloneqq \int_{\R^2/\Gamma} dx$. Then,
$$ \varphi(x) = (\F^{-1} \hat{\varphi})(x) = \sum_{b \in \Gamma^\#}
\hat{\varphi}(b) \, e^{i b \cdot x}, $$
$$ \| \varphi \|_{L^2(\R^2 / \Gamma)} = |\Gamma|^{1/2} \| \hat{\varphi}
\|_{l^2(\Gdual)}. $$
Recall that $k = u + iv$ with $u,v \in \R^2$, let $\rho$ be a positive
constant, and set
$$ \mathcal{K}_\rho \coloneqq \{k \in \C^2 \;|\;\; |v|\le \rho \}.$$
Finally, consider the projection
\begin{align*}
pr \, : \, \C^2 \, & \longrightarrow \, \C, \\
(k_1,k_2) \, & \longmapsto \, k_2,
\end{align*}
and define
$$ q \coloneqq (i\nabla \cdot A) + A^2 + V. $$

It is easy to construct a holomorphic map $E : \Fhat(A,V) \to \F(A,V)$
\cite{deO}. The precise form of this map is irrelevant here. For our
purposes it is enough to think of it simply as a ``projection'' (or
``exponential map'').

We are ready to state our results. Clearly, the set $\mathcal{K}_\rho$
is invariant under the action of $\Gdual$ and $\mathcal{K}_\rho /
\Gdual$ is compact. Hence, the image of $\Fhat(A,V) \cap
\mathcal{K}_\rho$ under the holomorphic map $E$ is compact in $\F(A,V)$.
This image set will essentially play the role of $X^{\rm com}$ in the
decomposition of $\F(A,V)$. Our first theorem characterizes the regular
piece $X^{\rm reg}$ of $\F(A,V)$.

\begin{thm}[The regular piece] \label{t:reg}
Let $0 < \eps < \Lambda/6$ and suppose that $A_1$, $A_2$ and $V$ are
functions in $L^2(\R^2/\Gamma)$ with $\| b^2 \hat{q}(b) \|_{l^1(\Gdual)}
< \infty$ and $\|(1+b^2) \hat{A}(b)\|_{l^1(\Gdual\setminus \{0\})} <
2\eps/63$. Then there is a constant $\rho = \rho_{\Lambda,\eps,q,A}$
such that, for $\nu \in \{1,2\}$, the projection $pr$ induces a
biholomorphic map between
$$ \left( \Fhat(A,V) \cap T_\nu(0) \right) \setminus \left(
\mathcal{K}_\rho \cup \bigcup_{b \in \Gdual \setminus \{0\}} T_b \right)
$$
and its image in $\C$. This image component contains
$$ \Big \{ z \in \C \,\, \Big| \,\, 8 |z| > \rho \, \text{ and }
|z+(-1)^\nu \theta_\nu(b)| > \eps \, \text{ for all } b \in \Gdual
\setminus \{0\} \Big \} $$
and is contained in
$$ \Bigg \{ z \in \C \,\, \Bigg | \,\, |z+(-1)^\nu\theta_\nu(b)| >
\frac{1}{2} \left( \eps - \frac{\eps^2}{40 \Lambda} \right) \, \text{
for all } b \in \Gdual \setminus \{0\} \Bigg \}, $$
where $\theta_{\nu}(b) = \frac{1}{2} ( (-1)^\nu b_2 + i b_1)$.
Furthermore,
\begin{align*}
pr^{-1} \, : \, \text{{\rm Image(}pr{\rm)}} & \, \longrightarrow \,
T_\nu(0),\\
y & \, \longmapsto \, (-\beta_2^{(1,0)} -i(-1)^\nu y - r(y),y),
\end{align*}
where $\beta_2^{(1,0)}$ is a constant given by \eqref{cte1} that depends
only on $\rho$ and $A$,
$$ |\beta^{(1,0)}_2| < \frac{\eps^2}{100 \Lambda} \qquad \text{and}
\qquad |r(y)| \leq \frac{\eps^3}{50 \Lambda^2} + \frac{C}{\rho}, $$
where $C = C_{\Lambda,\eps,q,A}$ is a constant.
\end{thm}

Now observe that, since $T_b + c = T_{b+c}$ for all $b, c \in \Gdual$,
the complement of $E \big(\Fhat(A,V) \cap \mathcal{K}_\rho\big)$ in
$\F(A,V)$ is the disjoint union of
$$ E \Bigg( \big(\Fhat(A,V) \cap T_0 \big) \setminus
\Bigg(\mathcal{K}_\rho \cup \bigcup_{\substack{b \in \Gdual \\ b_2 \neq
0}} T_b \Bigg) \Bigg) $$
and
$$ \bigcup_{\substack{b \in \Gdual \\ b_2 \neq 0}} E \big( \Fhat(A,V)
\cap T_0 \cap T_b \big). $$
Basically, the first of the two sets will be the regular piece of
$\F(A,V)$, while the second set will be the handles. The map $\Phi$
parametrizing the regular part will be the composition of the map $E$
with the inverse of the map discussed in the above theorem. The detailed
information about the handles $X^{\rm han}$ in $\F(A,V)$ comes from our
second main theorem.

\begin{thm}[The handles] \label{t:hand}
Let $0 < \eps < \Lambda/6$ and suppose that $A_1$, $A_2$ and $V$ are
functions in $L^2(\R^2/\Gamma)$ with $\| b^2 \hat{q}(b) \|_{l^1(\Gdual)}
< \infty$ and $\|(1+b^2) \hat{A}(b) \|_{l^1(\Gdual\setminus \{0\})} <
2\eps/63$. Then, for every sufficiently large constant $\rho$ and for
every $d \in \Gdual \setminus \{0\}$ with $2|d| > \rho$, there are maps
\begin{align*}
\phi_{d,1} \, & : \, \Big\{ (z_1,z_2) \in \C^2 \; \Big| \;\; |z_1| \le
\frac{\eps}{2} \text{ and } |z_2| \le \frac{\eps}{2} \Big\}
\longrightarrow T_1(0) \cap T_2(d), \\
\phi_{d,2} \, & : \, \Big\{ (z_1,z_2) \in \C^2 \; \Big| \;\; |z_1| \le
\frac{\eps}{2} \text{ and } |z_2| \le \eps \Big\} \longrightarrow
T_1(-d) \cap T_2(0),
\end{align*}
and a complex number $t_d$ with $|t_d| \leq \frac{C}{|d|^4}$ such that:
\begin{itemize}
\item[\rm (i)] For $\nu \in \{1,2\}$ the domain of the map
$\phi_{d,\nu}$ is biholomorphic to its image, and the image contains
$$ \Big\{ k \in \C^2 \; \Big| \;\; |k_1+i(-1)^\nu k_2| \le
\frac{\eps}{8} \; \text{ and } \; |k_1+(-1)^{\nu+1} d_1 - i(-1)^\nu( k_2
+ (-1)^{\nu+1} d_2)| \le \frac{\eps}{8} \Big \}. $$
Furthermore,
$$ D \hat{\phi}_{d,\nu} = \frac{1}{2} \begin{pmatrix} 1 & 1 \\
-i(-1)^\nu & i(-1)^\nu \end{pmatrix} \left(I + O \left( \frac{1}{|d|^2}
\right) \right) $$
and
$$ \phi_{d,\nu}(0) = (i \theta_\nu(d), (-1)^{\nu+1} \theta_\nu(d)) + O
\left( \frac{\eps}{900} \right) + O \left( \frac{1}{\rho} \right). $$

\item[\rm (ii)]
\begin{align*}
\phi_{d,1}^{-1} (T_1(0) \cap T_2(d) \cap \Fhat(A,V)) & = \Big\{
(z_1,z_2) \in \C^2 \; \Big| \;\; z_1z_2 = t_d, \;\; |z_1| \le
\frac{\eps}{2} \; \text{ and } \; |z_2| \le \frac{\eps}{2} \Big\}, \\
\phi_{d,2}^{-1} (T_1(-d) \cap T_2(0) \cap \Fhat(A,V)) & = \Big\{
(z_1,z_2) \in \C^2 \; \Big| \;\; z_1z_2 = t_d, \;\; |z_1| \le
\frac{\eps}{2} \; \text{ and } \; |z_2| \le \frac{\eps}{2} \Big\}.
\end{align*}

\item[\rm (iii)]
$$ \phi_{d,1}(z_1,z_2) = \phi_{d,2}(z_2,z_1) - d. $$
\end{itemize}
\end{thm}

These are the main results in this paper. In the next section we outline
the strategy for proving them. The proofs are presented in the
subsequent sections divided in many steps.


\section{Strategy outline} \label{s:idea}

Below we briefly describe the general strategy of analysis used to prove
our results. We first introduce some notation and definitions. Observe
that
\begin{align*}
H_k(A,V) \varphi & = ( (i \nabla + A - k)^2 + V ) \varphi \\ & = ( (i
\nabla-k)^2 + 2A \cdot (i \nabla -k) + (i \nabla \cdot A) + A^2 + V )
\varphi,
\end{align*}
and write
$$ H_k(A,V) = \Delta_k + h(k,A) + q(A,V) $$
with
$$ \Delta_k \coloneqq (i \nabla - k)^2, \qquad h(k,A) \coloneqq 2 A
\cdot (i \nabla - k) \qquad \text{and} \qquad q(A,V) \coloneqq (i\nabla
\cdot A) + A^2 + V. $$
For each finite subset $G$ of $\Gdual$ set
\begin{align*}
G' \coloneqq \Gdual \setminus G \qquad & \text{and} \qquad \C^2_G
\coloneqq \C^2 \setminus \bigcup_{b \in G'} \mathcal{N}_b, \\ L^2_G
\coloneqq \overline{\text{span} \{ e^{ib \cdot x} \; | \; b \in G \}}
\qquad & \text{and} \qquad L^2_{G'} \coloneqq \overline{\text{span} \{
e^{ib\cdot x} \; | \; b \in G' \}}.
\end{align*}
To simplify the notation write $L^2$ in place of $L^2(\R^2/\Gamma)$. Let
$I$ be the identity operator on $L^2$, and let $\pi_G$ and $\pi_{G'}$ be
the orthogonal projections from $L^2$ onto $L^2_G$ and $L^2_{G'}$,
respectively. Then,
$$ L^2 = L^2_G \oplus L^2_{G'} \qquad \text{and} \qquad I = \pi_G +
\pi_{G'}. $$
For $k \in \C^2_G$ define the partial inverse $(\Delta_k)^{-1}_G$ on
$L^2$ as
$$ (\Delta_k)^{-1}_G \coloneqq \pi_G + \Delta_k^{-1} \pi_{G'}. $$
Its matrix elements are
$$ \big(  (\Delta_k)^{-1}_G \big)_{b,c} \coloneqq \left \la \frac{e^{ib
\cdot x}}{|\Gamma|^{1/2}}, (\Delta_k)^{-1}_G \frac{e^{ic \cdot
x}}{|\Gamma|^{1/2}} \right\ra_{\!\!L^2} = \begin{cases} \delta_{b,c} &
\text{if} \quad c \in G, \\ \delta_{b,c} \frac{1}{N_c(k)} & \text{if}
\quad c \not\in G, \end{cases} $$
where $b,c \in \Gdual$.

Here is the main idea. By definition, a point $k$ is in $\Fhat(A,V)$ if
$H_k(A,V)$ has a nontrivial kernel in $L^2$. Hence, to study the part of
the curve in the intersection of $\cup_{d' \in G} T_{d'}$ with $\C^2
\setminus \cup_{b \in G'} T_b$ for some finite subset $G$ of $\Gdual$,
it is natural to look for a nontrivial solution of
$$ (\Delta_k + h + q) (\psi_G + \psi_{G'}) = 0, $$
where $\psi_G \in L^2_G$ and $\psi_{G'} \in L^2_{G'}$. Equivalently, if
we make the following (invertible) change of variables in $L^2$,
$$ (\psi_G + \psi_{G'}) = (\Delta_k)^{-1}_G ( \varphi_G + \varphi_{G'}),
$$
where $\varphi_G \in L^2_G$ and $\varphi_{G'} \in L^2_{G'}$, we may
consider the equation
\begin{equation} \label{eq} (\Delta_k + h + q) \varphi_G + (I + (h + q)
\Delta_k^{-1}) \varphi_{G'} = 0.
\end{equation}
The projections of this equation onto $L^2_{G'}$ and $L^2_G$ are,
respectively,
\begin{align} \pi_{G'} (h+q) \varphi_G + \pi_{G'}(I + (h + q)
\Delta_k^{-1}) \varphi_{G'} & = 0, \label{projG'} \\ \pi_G (\Delta_k + h
+ q) \varphi_G + \pi_G (h + q) \Delta_k^{-1} \varphi_{G'} & = 0.
\label{projG} 
\end{align}
Now define $R_{G'G'}$ on $L^2$ as
$$ R_{G'G'} \coloneqq \pi_{G'}(I + (h + q) \Delta_k^{-1}) \pi_{G'}. $$
Observe that $R_{G'G'}$ is the zero operator on $L^2_G$. Then, if
$R_{G'G'}$ has a bounded inverse on $L^2_{G'}$, the equation
\eqref{projG'} is equivalent to
$$ \varphi_{G'} = - R_{G'G'}^{-1} \pi_{G'} (h+q) \varphi_G. $$
Substituting this into \eqref{projG} yields
$$ \pi_G ( \Delta_k + h + q - (h + q) \Delta_k^{-1} R_{G'G'}^{-1}
\pi_{G'} (h+q) ) \varphi_G = 0. $$
This equation has a nontrivial solution if and only if the (finite) $|G|
\times |G|$ determinant
$$ \det \, [ \, \pi_G ( \Delta_k + h + q - (h + q) \Delta_k^{-1}
R_{G'G'}^{-1} \pi_{G'} (h+q) ) \pi_G \, ] = 0 $$
or, equivalently, expressing all operators as matrices in the basis $\{
|\Gamma|^{-1/2} \, e^{ib \cdot x} \; | \; b \in \Gdual \}$,
\begin{equation} \label{detGxG} \det \left[ N_{d'}(k) \delta_{d',d''} +
w_{d',d''} - \sum_{b,c \in G'} \frac{w_{d',b}}{N_b(k)}
(R_{G'G'}^{-1})_{b,c} w_{c,d''} \right]_{d',d'' \in G} = 0,
\end{equation}
where
$$ w_{b,c} \coloneqq h_{b,c} + \hat{q}(b-c) = -2(c+k) \cdot \hat{A}(b-c)
+ \hat{q}(b-c). $$
Therefore, if $R_{G'G'}$ has a bounded inverse on $L^2_{G'}$---which is
in fact the case under suitable conditions---in the region under
consideration we can study the Fermi curve in detail using the (local)
defining equation \eqref{detGxG}.


\section{Invertibility of $R_{G'G'}$ } \label{s:inv}

The following notation will be used whenever we consider vector-valued
quantities. Let $\mathcal{X}$ be a Banach space and let $A, B \in
\mathcal{X}^2$, where $A = (A_1, A_2)$ and $B = (B_1, B_2)$. Then,
$$ \| A \|_\mathcal{X} \coloneqq (\| A_1 \|_\mathcal{X}^2 + \| A_2
\|_\mathcal{X}^2)^{1/2} \qquad \text{and} \qquad A \cdot B \coloneqq A_1
B_1 + A_2 B_2. $$
Furthermore, we will denote by $\| \, \cdot \, \|$ the operator norm on
$L^2(\R^2 / \Gamma)$.

In general, for any $B,C \subset \Gdual$ ($C$ such that $\Delta_k^{-1}
\pi_C$ exists) define the operator $R_{BC}$ as
\begin{align}
R_{BC} \coloneqq & \, \pi_B (I + (h + q) \Delta_k^{-1}) \pi_C \notag \\
= & \, \pi_B \pi_C + \pi_B \, q \, \Delta_k^{-1} \pi_C + \pi_B (2A \cdot
i \nabla) \Delta_k^{-1} \pi_C - \pi_B (2 k \cdot A) \Delta_k^{-1} \pi_C.
\label{RBC} 
\end{align}
Its matrix elements are
\begin{equation} \label{RBCm}
(R_{BC})_{b,c} = \delta_{b,c} + \frac{\hat{q}(b-c)}{N_c(k)} - \frac{2c
\cdot \hat{A}(b-c)}{N_c(k)} - \frac{2 k \cdot \hat{A}(b-c)}{N_c(k)},
\end{equation}
where $b \in B$ and $c \in C$. We first estimate the norm of the last
three terms on the right hand side of \eqref{RBC}. We begin with the
following proposition.

\begin{prop} \label{p:bd1}
Let $k \in \C^2$ and let $B,C \subset \Gdual$ with $C \subset \{ b \in
\Gdual \; | \; N_b(k) \neq 0 \}$. Then,
\begin{align*}
\| \pi_B \, q \, \Delta_k^{-1} \pi_C \| & \leq \| \hat{q} \|_{l^1}
\sup_{c \in C} \frac{1}{|N_c(k)|}, \\
\| \pi_B (A \cdot i \nabla) \Delta_k^{-1} \pi_C \| & \leq \| \hat{A}
\|_{l^1} \sup_{c \in C} \frac{|c|}{|N_c(k)|}, \\
\| \pi_B (k \cdot A) \Delta_k^{-1} \pi_C \| & \leq \| \hat{A} \|_{l^1}
\, |k| \, \sup_{c \in C} \frac{1}{|N_c(k)|}.
\end{align*}
\end{prop}

To prove this proposition we apply the following well-known inequality
(see \cite{deO}).

\begin{prop} \label{p:ineq}
Consider a linear operator $T : L^2_C \to L^2_B$ with matrix elements
$T_{b,c}$. Then,
$$ \| T \| \leq \max \left\{ \sup_{c \in C} \, \sum_{b \in B} |T_{b,c}|,
\; \sup_{b \in B} \, \sum_{c \in C} |T_{b,c}| \right\}. $$
\end{prop}

\begin{proof}[Proof of Proposition \ref{p:bd1}]
We only prove the first inequality. The proof of the other ones is
similar.  Write $T \coloneqq \pi_B \, q \, \Delta_k^{-1} \pi_C$. Then,
in view of \eqref{RBC} and \eqref{RBCm},
\begin{align*}
\sup_{c \in C} \sum_{b \in B} |T_{b,c}| & \leq \sup_{c \in C} \sum_{b
\in B} \frac{|\hat{q}(b-c)|}{|N_c(k)|} \leq \sup_{c \in C}
\frac{1}{|N_c(k)|} \| \hat{q} \|_{l^1}, \\
\sup_{b \in B} \sum_{c \in C} |T_{b,c}| & \leq \sup_{b \in B} \sum_{c
\in C} \frac{|\hat{q}(b-c)|}{|N_c(k)|} \leq \sup_{c \in C}
\frac{1}{|N_c(k)|} \| \hat{q} \|_{l^1}.
\end{align*}
By Proposition \ref{p:ineq}, these estimates yield the desired
inequality.  \end{proof}

The key estimate for the existence of $R_{G'G'}^{-1}$ is given below.

\begin{prop}[Estimate of $\| R_{SS} - \pi_{S} \|$] \label{p:bdR}
Let $k \in \C^2$ with $|u| \leq 2|v|$ and $|v| > 2 \Lambda$. Suppose
that $S \subset \{ b \in \Gdual \; | \;\; |N_b(k)| \geq \eps |v| \}$.
Then,
\begin{equation} \label{key}
\| R_{SS} - \pi_S \| \leq \| \hat{q} \|_{l^1} \frac{1}{\eps |v|} +
\frac{14}{\eps} \| \hat{A} \|_{l^1}.
\end{equation}
\end{prop}

If $A=0$, the right hand side of \eqref{key} can be made arbitrarily
small for any $V$ by taking $|v|$ sufficiently large (recall that
$q(0,V)=V$). If $A \neq 0$, however, we need to take $\| \hat{A}
\|_{l^1}$ small to make that quantity less than 1. The term $
\tfrac{14}{\eps} \| \hat{A}\|_{l^1}$ in \eqref{key} comes from the
estimate we have for $ \| \pi_{G'} \, h \, \Delta_k^{-1} \pi_{G'} \|$.

\begin{proof}[Proof of Proposition \ref{p:bdR}]
By hypothesis, for all $b \in S$,
\begin{equation} \label{b5}
\frac{1}{|N_b(k)|} \leq \frac{1}{\eps |v|}.
\end{equation}
We now show that, for all $b \in S$,
\begin{equation} \label{b6}
\frac{|b|}{|N_b(k)|} \leq \frac{4}{\eps}.
\end{equation}
First suppose that $|b| \leq 4|v|$. Then,
$$ \frac{|b|}{|N_b(k)|} \leq \frac{4|v|}{\eps |v|} = \frac{4}{\eps}. $$
Now suppose that $|b| \geq 4|v|$. Again, by hypothesis we have $|u| \leq
2|v|$ and $|v| > 2\Lambda > \eps$. Hence,
$$ |v \pm (u+b)^\perp| \geq |b|-|u|-|v| \geq |b|-3|v| \geq |b|-
\frac{3}{4} |b| = \frac{|b|}{4}. $$
Consequently,
$$ \frac{|b|}{|N_b(k)|} = \frac{|b|}{|v+(u+b)^\perp| \, |v-(u+b)^\perp|}
\leq |b| \frac{4}{|b|} \frac{4}{|b|} = \frac{16}{|b|} \leq \frac{4}{|v|}
\leq \frac{4}{\eps}. $$
This proves \eqref{b6}.

The expression for $R_{SS} - \pi_S$ is given by \eqref{RBC}. Observe
that $|k| \leq |u|+|v| \leq 3|v|$. Then, applying Proposition
\ref{p:bd1} and using \eqref{b5} and \eqref{b6} we obtain
\begin{align*}
\| R_{SS} - \pi_S \| & \leq ( 6 |v| \, \| \hat{A} \|_{l^1} + \| \hat{q}
\|_{l^1} ) \sup_{b \in S} \frac{1}{|N_c(k)|} + 2 \| \hat{A} \|_{l^1}
\sup_{b \in S} \frac{|c|}{|N_c(k)|} \\
& \leq ( 6 |v| \, \| \hat{A} \|_{l^1} + \| \hat{q} \|_{l^1} )
\frac{1}{\eps |v|} + \frac{8}{\eps} \| \hat{A} \|_{l^1} = \| \hat{q}
\|_{l^1} \frac{1}{\eps |v|} + \frac{14}{\eps} \| \hat{A} \|_{l^1}.
\end{align*}
This is the desired inequality.
\end{proof}

From the last proposition it follows easily that $R_{SS}$ has a bounded
inverse for large $|v|$ and weak magnetic potential.

\begin{lem}[Invertibility of $R_{SS}$] \label{l:invR}
Let $k \in \C^2$,
$$ |u| \leq 2|v|, \qquad |v| > \max \left\{ 2 \Lambda, \, \| \hat{q}
\|_{l^1} \frac{2}{\eps} \right\}, \qquad \| \hat{q} \|_{l^1} < \infty
\qquad \text{and} \qquad \| \hat{A} \|_{l^1} < \frac{2}{63} \eps. $$
Suppose that $S \subset \{ b \in \Gdual \; | \; |N_b(k)| \geq \eps |v|
\}$. Then the operator $R_{SS}$ has a bounded inverse with
\begin{align*}
\| R_{SS} - \pi_S \| & < \| \hat{q} \|_{l^1} \frac{1}{\eps |v|} + \|
\hat{A} \|_{l^1} \frac{14}{\eps} < \frac{17}{18}, \\
\| R_{SS}^{-1} - \pi_S \| & < 18 \| R_{SS} - \pi_S \|.
\end{align*}
\end{lem}

\begin{proof}
Write $R_{SS} = \pi_S + T$ with $T = R_{SS} - \pi_S$. Then, by
Proposition \ref{p:bdR},
$$ \| T \| = \| R_{SS} - \pi_S \| \leq \| \hat{q} \|_{l^1} \frac{1}{\eps
|v|} + \| \hat{A} \|_{l^1} \frac{14}{\eps} < \frac{1}{2} + \frac{4}{9} =
\frac{17}{18} < 1. $$
Hence, the Neumann series for $R_{SS}^{-1} = (\pi_S + T)^{-1}$ converges
(and is a bounded operator). Furthermore,
\begin{align*}
\| R_{SS}^{-1} - \pi_S \| & = \| (\pi_S + T)^{-1} - \pi_S \| = \| (\pi_S
+ T)^{-1} - (\pi_S + T)^{-1} (\pi_S + T) \| \\
& = \| (\pi_S + T)^{-1} T \| \leq (1-\|T\|)^{-1} \| T \| < 18 \| R_{SS}
- \pi_S \|,
\end{align*}
as was to be shown.
\end{proof}

Lemma \ref{l:invR} says that if $G$ is such that $G' \subset \{ b \in
\Gdual \; | \; |N_b(k)| \geq \eps |v| \}$ the operator $R_{G'G'}$ has a
bounded inverse on $L^2_{G'}$ for $|u| \leq 2|v|$, large $|v|$, and weak
magnetic potential. We are now able to write local defining equations
for $\Fhat(A,V)$ under such conditions.


\section{Local defining equations}

In this section we derive local defining equations for the Fermi curve.
We begin with a simple proposition.

\begin{prop} \label{p:Gprime}
Suppose either {\rm (i)} or {\rm (ii)} or {\rm (iii)} where:

\begin{itemize}
\item[] \begin{itemize}
\item[{\rm (i)}] $G = \{ 0 \}$ and $k \in T_0 \setminus \cup_{b \in
\Gdual \setminus \{0\} }T_b$;

\item[{\rm (ii)}] $G = \{ 0, d \}$ and $k \in T_0 \cap T_d$;

\item[{\rm (iii)}] $G = \varnothing$ and $k \in \C^2 \setminus \cup_{b
\in \Gdual} T_b$.
\end{itemize}
\end{itemize}
Then $G' = \Gdual \setminus G = \{ b \in \Gdual \; | \;\; |N_b(k)| \geq
\eps |v| \}$.
\end{prop}

\begin{proof}
The proposition follows easily if we observe that $G' = \Gdual \setminus
G$ and recall from \eqref{bd1} that
\[ k \not\in T_b \quad \Longrightarrow \quad |N_b(k)| \geq \eps |v|.
\qedhere \]
\end{proof}

We now introduce some notation. Let $\mathcal{B}$ be a fundamental cell
for $\Gdual \subset \R^2$ (see \cite[p~310]{RS4}). Then any vector $u
\in \R^2$ can be written as $u = \xi + \mathfrak{u}$ for some $\xi \in
\Gdual$ and $\mathfrak{u} \in \mathcal{B}$. Define
\[ \alpha \coloneqq \sup \{ |\mathfrak{u}| \; | \; \mathfrak{u} \in
\mathcal{B} \}, \qquad R \coloneqq \max \left\{ \alpha, \, 2\Lambda, \,
\| \hat{q} \|_{l^1} \frac{2}{\eps} \right \}, \qquad \mathcal{K}_R
\coloneqq \{ k \in \C^2 \; | \; |v| \leq R \}. \]
We first show that in $\C^2 \setminus \mathcal{K}_R$ the Fermi curve is
contained in the union of $\eps$-tubes about the free Fermi curve.

\begin{prop}[$\Fhat(A,V) \setminus \mathcal{K}_R$ is contained in the
union of $\eps$-tubes]
$$ \Fhat(A,V) \setminus \mathcal{K}_R \subset \bigcup_{b \in \Gdual}
T_b. $$ 
\end{prop}

\begin{proof}
Without loss of generality we may consider $k \in \C^2$ with real part
in $\mathcal{B}$. We now prove that any point outside the region
$\mathcal{K}_R$ and outside the union of $\eps$-tubes does not belong to
$\Fhat(A,V)$. Suppose that $k \in \C^2 \setminus (\mathcal{K}_R \cup
\bigcup_{b \in \Gdual} T_b)$ and recall that $k$ is in $\Fhat(A,V)$ if
and only if \eqref{eq} has a nontrivial solution. If we choose $G =
\varnothing$ then $G' = \Gdual$ and this equation reads
$$ R_{G'G'} \varphi_{G'} = 0. $$
By Proposition \ref{p:Gprime}(iii) we have $G' = \Gdual = \{ b \in
\Gdual \; | \; |N_b(k)| \geq \eps |v| \}$. Furthermore, since $u \in
\mathcal{B}$ and $|v| > R \geq \alpha$, it follows that $|u| \leq 
\alpha < |v| < 2|v|$. Consequently, the operator $R_{G'G'}$ has a 
bounded inverse by Lemma \ref{l:invR}. Thus, the only solution of the
above equation is $\varphi_{G'} = 0$. That is, there is no nontrivial
solution of this equation and therefore $k \not\in \Fhat(A,V)$.
\end{proof}

We are left to study the Fermi curve inside the $\eps$-tubes. There are
two types of regions to consider: intersections and non-intersections of
tubes. To study non-intersections we choose $G = \{ 0 \}$ and consider
the region $(T_0 \setminus \cup_{b \in \Gdual \setminus \{0\}} T_b)
\setminus \mathcal{K}_R$. For intersections we take $G = \{0, d\}$ for
some $d \in \Gdual \setminus \{0\}$ and consider $(T_0 \cap T_d)
\setminus \mathcal{K}_R$. Observe that, since the tubes $T_b$ have the
following translational property, $T_b + c = T_{b+c}$ for all $b,c \in
\Gdual$, and the curve $\Fhat(A,V)$ is invariant under the action of
$\Gdual$, there is no loss of generality in considering only the two
regions above. Any other part of the curve can be reached by
translation.

Recall that $G' = \Gdual \setminus G$ and for $d',d'' \in G$ and $i,j
\in \{1,2 \}$ set
\begin{equation} \label{defBC}
\begin{aligned}
B_{ij}^{d'd''}(k;G) & \coloneqq -4 \sum_{b,c \in G'}
\frac{\hat{A}_i(d'-b)}{N_b(k)} (R_{G'G'}^{-1})_{b,c} \,
\hat{A}_j(c-d''), \\
C_{i}^{d'd''}(k;G) & \coloneqq -2 \hat{A}_i(d'-d'') + 2 \sum_{b,c \in
G'} \frac{\hat{q}(d'-b) - 2b \cdot \hat{A}(d'-b)}{N_b(k)}
(R_{G'G'}^{-1})_{b,c} \hat{A}_i(c-d'') \\
& \quad + 2 \sum_{b,c \in G'} \frac{\hat{A}_i(d'-b)}{N_b(k)}
(R_{G'G'}^{-1})_{b,c} (\hat{q}(c-d'') - 2d'' \cdot \hat{A}(c-d'')), \\
C_{0}^{d'd''}(k;G) & \coloneqq \hat{q}(d'-d'') -2 d'' \cdot
\hat{A}(d'-d'') \\
& \quad - \sum_{b,c \in G'} \frac{\hat{q}(d'-b) - 2b \cdot
\hat{A}(d'-b)}{N_b(k)} (R_{G'G'}^{-1})_{b,c} (\hat{q}(c-d'') - 2d''
\cdot \hat{A}(c-d'')).  \end{aligned}
\end{equation}
Then,
\begin{align*}
D_{d',d''}(k;G) & \coloneqq w_{d',d''} - \sum_{b,c \in G'}
\frac{w_{d',b}}{N_b(k)} (R_{G'G'}^{-1})_{b,c} w_{c,d''} \\
& = B_{11}^{d'd''} k_1^2 + B_{22}^{d'd''} k_2^2 + (B_{12}^{d'd''} +
B_{21}^{d'd''}) k_1 k_2 + C_1^{d'd''} k_1 + C_2^{d'd''} k_2 +
C_0^{d'd''}.
\end{align*}
These functions have the following property.

\begin{prop} \label{p:holom}
For $d',d'' \in G$ and $i,j \in \{1,2\}$, the functions
$B_{ij}^{d'd''}$, $C_i^{d'd''}$, $C_0^{d'd''}$ {\rm(}and consequently
$D_{d',d''}${\rm)} are analytic on $(T_0 \setminus \cup_{b \in \Gdual
\setminus \{0\}} T_b) \setminus \mathcal{K}_R$ and $(T_0 \cap T_d)
\setminus \mathcal{K}_R$ for $G=\{0\}$ and $G=\{0,d\}$, respectively.
\end{prop}

\begin{proof}[Sketch of the proof]
It suffices to show that $B_{ij}^{d'd''}$, $C_i^{d'd''}$ and
$C_0^{d'd''}$ are analytic functions. This property follows from the
fact that all the series involved in the definition of these functions
are uniformly convergent sums of analytic functions. The argument is
similar for all cases. See \cite{deO} for details.
\end{proof}

Using the above functions we can write (local) defining equations for
the Fermi curve.

\begin{lem}[Local defining equations for $\Fhat(A,V)$] \label{l:eqn}
$\,$
\begin{itemize}
\item[{\rm (i)}] Let $G = \{0\}$ and $k \in (T_0 \setminus \cup_{b \in
\Gdual \setminus \{0\}} T_b) \setminus \mathcal{K}_R$. Then $k \in
\Fhat(A,V)$ if and only if
$$ N_0(k) + D_{0,0}(k) = 0. $$

\item[{\rm (ii)}] Let $G=\{0,d\}$ and $k \in (T_0 \cap T_d) \setminus
\mathcal{K}_R$. Then $k \in \Fhat(A,V)$ if and only if
$$ (N_0(k) + D_{0,0}(k))(N_d(k) + D_{d,d}(k)) - D_{0,d}(k)D_{d,0}(k) =
0. $$ 
\end{itemize}
\end{lem}

\begin{proof}
We only prove part (i). The proof of part (ii) is similar. First, by
Proposition \ref{p:Gprime}(i) we have $G' = \Gdual \setminus \{0\} = \{
b \in \Gdual \; | \; |N_b(k)| \geq \eps |v| \}$. Furthermore, since $k
\in T_0$, we have either $|v-u^\perp|<\eps$ or $|v+u^\perp|<\eps$. In
either case this implies $|u| < \eps + |v| < 2\Lambda + |v| < 2|v|$.
Hence, the operator $R_{G'G'}$ has a bounded inverse by Lemma
\ref{l:invR}. Thus, in the region under consideration $\Fhat(A,V)$ is
given by \eqref{detGxG}:
$$ 0 = N_0(k) + w_{0,0} - \sum_{b,c \in G'} \frac{w_{0,b}}{N_b(k)}
(R_{G'G'}^{-1})_{b,c} w_{c,0} = N_0(k) + D_{0,0}(k). $$
This is the desired expression.
\end{proof}

To study in detail the defining equations above we shall estimate the
asymptotic behaviour of the functions $B_{ij}^{d'd''}$, $C_i^{d'd''}$,
$C_0^{d'd''}$ and $D_{d',d''}$ for large $|v|$. (We sometimes refer to
these functions as coefficients.) Since all these functions have a
similar form it is convenient to prove these estimates in a general
setting and specialize them later. This is the contents of \S
\ref{s:coeff} and \S \ref{s:der}. We next introduce a change of
variables in $\C^2$ that will be useful for proving these bounds.


\section{Change of coordinates} \label{s:change}

Define the (complementary) index $\nu'$ as $\nu' \coloneqq \nu -
(-1)^\nu$.  Observe that $\nu' = 2$ if $\nu = 1$, $\nu' = 1$ if $\nu =
2$, and $(-1)^\nu = -(-1)^{\nu'}$.
%
%
%
The following change of coordinates in $\C^2$ will be useful for our
analysis. For $\nu \in \{1,2\}$ and $d',d'' \in G$ define the functions
$w_{\nu,d'}, \, z_{\nu,d'} : \C^2 \to \C$ as
\begin{equation} \label{change}
\begin{aligned}
w_{\nu,d'}(k) & \coloneqq k_1+d'_1 + i(-1)^\nu (k_2+d'_2), \\
z_{\nu,d'}(k) & \coloneqq k_1+d'_1 - i(-1)^\nu (k_2+d'_2).
\end{aligned}
\end{equation}
Observe that, the transformation $(k_1,k_2) \mapsto (w_{\nu,d'},
z_{\nu,d'})$ is just a translation composed with a rotation.
Furthermore, if $k \in T_\nu(d') \setminus \mathcal{K}_R$ then
$|w_{\nu,d'}(k)|$ is ``small'' and $|z_{\nu,d'}(k)|$ is ``large''.
Indeed, $|w_{\nu,d'}(k)| = |N_{d',\nu}(k)| < \eps$ and $|z_{\nu,d'}(k)|
= |N_{d',\nu'}(k)| \geq |v| > R$. Define also
\begin{align*}
J^{d'd''}_\nu & \coloneqq \tfrac{1}{4} (B_{11}^{d'd''} - B_{22}^{d'd''}
+ i(-1)^\nu (B_{12}^{d'd''} + B_{21}^{d'd''})), \\
K^{d'd''} & \coloneqq \tfrac{1}{2} (B_{11}^{d'd''} + B_{22}^{d'd''}),\\
L^{d'd''}_\nu & \coloneqq -d_1' B_{11}^{d'd''} - i(-1)^\nu d_2'
B_{22}^{d'd''} - \tfrac{1}{2}(d_2' + i(-1)^\nu d_1') (B_{12}^{d'd''} 
+ B_{21}^{d'd''}) \\
& \qquad + \tfrac{1}{2}( C_1^{d'd''} + i(-1)^\nu C_2^{d'd''}), \\
M^{d'd''} & \coloneqq d_1'^2 B_{11}^{d'd''} + d_2'^2 B_{22}^{d'd''} 
+ d_1' d_2' (B_{12}^{d'd''} + B_{21}^{d'd''}) -d_1' C_1^{d'd''} 
- d_2' C_2^{d'd''} + C_0^{d'd''},
\end{align*}
where $J^{d'd''}_\nu$, $K^{d'd''}$, $L^{d'd''}_\nu$ and $M^{d'd''}$ are
functions of $k \in \C^2$ that also depend on the choice of $G \subset
\Gdual$. Using these functions we can express $N_{d'}(k) + D_{d',d'}(k)$
and $D_{d',d''}(k)$ as follows.

\begin{prop} \label{p:newv} 
Let $\nu \in \{1,2\}$ and let $d',d'' \in G$. Then,
\begin{align*}
N_{d'} + D_{d',d'} & = J^{d'd'}_{\nu'} w_{\nu,d'}^2 + J^{d'd'}_\nu
z_{\nu,d'}^2 + (1 + K^{d'd'}) w_{\nu,d'} z_{\nu,d'} + L^{d'd'}_{\nu'}
w_{\nu,d'} + L^{d'd'}_\nu z_{\nu,d'} + M^{d'd'}, \\
D_{d',d''} & = J^{d'd''}_{\nu'} w_{\nu,d'}^2 + J^{d'd''}_\nu
z_{\nu,d'}^2 + K^{d'd''} w_{\nu,d'} z_{\nu,d'} + L^{d'd''}_{\nu'}
w_{\nu,d'} + L^{d'd''}_\nu z_{\nu,d'} + M^{d'd''}.
\end{align*}
Furthermore,
{\allowdisplaybreaks
\begin{align*}
J^{d'd''}_\nu(k) & = - \sum_{b,c \in G'} \frac{(1,-i(-1)^\nu) \cdot
\hat{A}(d'-b)}{N_b(k)} (R_{G'G'}^{-1})_{b,c} \, (1,-i(-1)^\nu) \cdot
\hat{A}(c-d''), \\
K^{d'd''}(k) & = - 2 \sum_{b,c \in G'} \frac{\hat{A}(d'-b) \cdot
\hat{A}(c-d'')}{N_b(k)} (R_{G'G'}^{-1})_{b,c}, \\ L^{d'd''}_\nu(k) & =
\sum_{b,c \in G'} \frac{\hat{q}(d'-b) + 2(d'-b) \cdot
\hat{A}(d'-b)}{N_b(k)} (R_{G'G'}^{-1})_{b,c} (1,i(-1)^\nu) \cdot
\hat{A}(c-d'')\\
& \quad + \sum_{b,c \in G'} \frac{(1,i(-1)^\nu) \cdot
\hat{A}(d'-b)}{N_b(k)} (R_{G'G'}^{-1})_{b,c} \, (\hat{q}(c-d'') + 2
(d'-d'') \cdot \hat{A}(c-d'')) \\ & \quad - (1,i(-1)^\nu) \cdot
\hat{A}(d'-d''), \\
M^{d'd''}(k) & = - \sum_{b,c \in G'} \frac{\hat{q}(d'-b) + 2(d'-b) \cdot
\hat{A}(d'-b)}{N_b(k)} (R_{G'G'}^{-1})_{b,c} \, \hat{q}(c-d'')\\ & \quad
+ \hat{q}(d'-d'') + 2 (d'-d'') \cdot \hat{A}(d'-d'').
\end{align*}
}
\end{prop}

\begin{proof}
To simplify the notation write $w = w_{\nu,d'}$, $z = z_{\nu,d'}$,
$B_{ij} = B_{ij}^{d'd''}$ and $C_i = C_i^{d'd''}$. First observe that,
in view of \eqref{change},
$$ N_{d'} = (k_1+d_1' + i(-1)^\nu(k_2+d_2')) (k_1+d_1' -
i(-1)^\nu(k_2+d_2')) = wz. $$
Furthermore,
\begin{align*}
k_1 & = \tfrac{1}{2}(w+z)-d_1', \\
k_2 & = \tfrac{(-1)^\nu}{2i}(w-z)-d_2', \\
k_1^2 & = \tfrac{1}{4} (w^2+z^2) + \tfrac{1}{2}wz -d_1'(w+z)+d_1'^2,\\
k_2^2 & = -\tfrac{1}{4} (w^2+z^2) + \tfrac{1}{2}wz + i(-1)^\nu d_2'(w-z)
+ d_2'^2, \\
k_1 k_2 & = \tfrac{i(-1)^\nu}{4} (z^2 - w^2) - \tfrac{1}{2} (d_2' -
i(-1)^\nu d_1') w - \tfrac{1}{2} (d_2' + i(-1)^\nu d_1') + d_1'd_2'.
\end{align*}
Hence,
\begin{align*}
D_{d',d''} & = B_{11} k_1^2 + B_{22} k_2^2 + (B_{12} + B_{21}) k_1 k_2 +
C_1 k_1 + C_2 k_2 + C_0 \\
& = \tfrac{1}{4}(B_{11} - B_{22} -i(-1)^\nu(B_{12} +B_{21})) w^2 +
\tfrac{1}{4} (B_{11} - B_{22} + i(-1)^\nu(B_{12} + B_{21})) z^2\\
& \quad + \big( -d_1' B_{11} + i(-1)^\nu d_2' B_{22} - \tfrac{1}{2}(d_2'
- i(-1)^\nu d_1')(B_{12} + B_{21}) + \tfrac{1}{2} (C_1 - i(-1)^\nu C_2)
\big ) w \\
& \quad + \big( -d_1' B_{11} + i(-1)^\nu d_2' B_{22} - \tfrac{1}{2}(d_2'
+ i(-1)^\nu d_1')(B_{12} + B_{21}) + \tfrac{1}{2} (C_1 + i(-1)^\nu C_2)
\big) z \\
& \quad + d_1'^2 B_{11} + d_2'^2 B_{22} + d_1' d_2' (B_{12} + B_{21})
-d_1' C_1 - d_2' C_2 + C_0 + \tfrac{1}{2}(B_{11}+B_{22}) wz \\
& = J^{d'd''}_{\nu'} w^2 + J^{d'd''}_\nu z^2 + K^{d'd''} w z +
L^{d'd''}_{\nu'} w + L^{d'd''}_\nu z + M^{d'd''}.
\end{align*}
This proves the first claim. Consequently,
$$ N_{d'} + D_{d',d'} = J^{d'd'}_{\nu'} w^2 + J^{d'd'}_\nu z^2 +
(1+K^{d'd'}) w z + L^{d'd'}_{\nu'} w + L^{d'd'}_\nu z + M^{d'd'}, $$
which proves the second claim.

Now, again to simplify the notation write
$$ f g = \sum_{b,c \in G'} \frac{\hat{f}(b,d')}{N_b(k)}
(R_{G'G'}^{-1})_{b,c} \, \hat{g}(c,d''), $$
that is, to represent sums of this form suppress the summation and the
other factors. Note that $f g \neq g f$ according to this notation.
Then, substituting \eqref{defBC} into the definition of $J^{d'd''}_\nu$
we have
\begin{align*}
J^{d',d''}_\nu & = \tfrac{1}{4} (B_{11} - B_{22} + i(-1)^\nu (B_{12} +
B_{21})) = -A_1A_1  + A_2A_2 - i(-1)^\nu (A_1A_2+A_2A_1) \\
& = (A_1 -i(-1)^\nu A_2) ( -A_1 + i(-1)^\nu A_2) = - ((1,-i(-1)^\nu)
\cdot A)\, ((1,-i(-1)^\nu) \cdot A ) \\
& = - \sum_{b,c \in G'} \frac{(1,-i(-1)^\nu) \cdot
\hat{A}(d'-b)}{N_b(k)} (R_{G'G'}^{-1})_{b,c} \, (1,-i(-1)^\nu) \cdot
\hat{A}(c-d'').
\end{align*}
Similarly, substituting \eqref{defBC} into the definitions of
$K^{d'd''}$, $L^{d'd''}_\nu$ and $M^{d'd''}$ we derive the other
expressions.
\end{proof}


\section{Asymptotics for the coefficients} \label{s:coeff}

Let $f$ and $g$ be functions on $\Gdual$ and for $k \in \C^2$ and
$d',d'' \in G$ set
\begin{equation} \label{Phi}
\Phi_{d',d''}(k;G) \coloneqq \sum_{b,c \in G'} \frac{f(d'-b)}{N_b(k)}
(R_{G'G'}^{-1})_{b,c} \, g(c-d'').
\end{equation}
In this section we study the asymptotic behaviour of the function
$\Phi_{d',d''}(k)$ for $k$ in the union of $\eps$-tubes with large
$|v|$. Here we only give the statements. See Appendix \ref{s:app2} for
the proofs. Reset the constant $R$ as
\begin{equation} \label{R1}
R \coloneqq \max \left\{1, \, \alpha, \, 2\Lambda, \, 140 \| \hat{A}
\|_{l^1}, \, \| (1+b^2)\hat{q}(b) \|_{l^1} \frac{4}{\eps} \right\},
\end{equation}
and make the following hypothesis.
\begin{hyp} \label{h:small}
$$ \| b^2 \hat{q}(b) \|_{l^1} < \infty \qquad \text{and} \qquad \|
(1+b^2) \hat{A}(b) \|_{l^1} < \frac{2}{63} \eps. $$
\end{hyp}

Our first lemma provides and expansion for $\Phi_{d',d'}(k)$ ``in powers
of $1/|z_{\nu,d'}(k)|$''.

\begin{lem}[Asymptotics for $\Phi_{d',d'}(k)$] \label{l:b1}
Under Hypothesis \ref{h:small}, let $\nu \in \{1,2\}$ and let $f$ and
$g$ be functions on $\Gdual$ with $\| b^2 f(b) \|_{l^1} < \infty$ and
$\| b^2 g(b) \|_{l^1} < \infty$. Suppose either {\rm (i)} or {\rm (ii)}
where:
\begin{itemize}
\item[]
\begin{itemize}
\item[\rm (i)] $G=\{0\}$ and $k \in (T_\nu(0) \setminus \cup_{b \in G'}
T_b) \setminus \mathcal{K}_R$;

\item[\rm (ii)] $G=\{0, d\}$ and $k \in (T_\nu(0) \cap T_{\nu'}(d))
\setminus \mathcal{K}_R$.
\end{itemize}
\end{itemize}
Then, for $(\mu,d') = (\nu,0)$ if {\rm (i)} or $(\mu,d') \in \{ (\nu,0),
(\nu',d) \}$ if {\rm (ii)},
$$ \Phi_{d',d'}(k) = \alpha_{\mu,d'}^{(1)}(k) + \alpha_{\mu,d'}^{(2)}(k)
+ \alpha_{\mu,d'}^{(3)}(k), $$
where for $1 \le j \le 2$,
$$ |\alpha_{\mu,d'}^{(j)}(k)| \leq \frac{C_j}{(2|z_{\mu,d'}(k)|-R)^j}
\qquad \text{and} \qquad |\alpha_{\mu,d'}^{(3)}(k)| \leq
\frac{C_3}{|z_{\mu,d'}(k)| R^2}, $$
where $C_j = C_{j;\Lambda,A,q,f,g}$ and $C_3 =
C_{3;\eps,\Lambda,A,q,f,g}$ are constants. Furthermore, the functions
$\alpha_{\mu,d'}^{(j)}(k)$ are given by \eqref{alpha} and \eqref{alpha3}
and are analytic in the region under consideration.
\end{lem}

Below we have more information about the function
$\alpha_{\mu,d'}^{(1)}(k)$.

\begin{lem}[Asymptotics for $\alpha_{\mu,d'}^{(1)}(k)$] \label{l:b1b} 
Consider the same hypotheses of Lemma \ref{l:b1}. Then, for $(\mu,d') =
(\nu,0)$ if {\rm (i)} or $(\mu,d') \in \{ (\nu,0), (\nu',d) \}$ if {\rm
(ii)},
$$ z_{\mu,d'}(k) \, \alpha_{\mu,d'}^{(1)}(k) = \alpha^{(1,0)}_{\mu,d'} +
\alpha^{(1,1)}_{\mu,d'}(w(k)) + \alpha^{(1,2)}_{\mu,d'}(k) +
\alpha^{(1,3)}_{\mu,d'}(k), $$
where $\alpha^{(1,0)}_{\mu,d'}$ is a constant given by \eqref{a10}, and
the remaining functions $\alpha^{(1,j)}_{\mu,d'}$ are given by
\eqref{a1j}. Furthermore, for $0 \le j \le 2$,
$$ |\alpha_{\mu,d'}^{(1,j)}| \leq C_j \qquad \text{and} \qquad
|\alpha_{\mu,d'}^{(1,3)}| \leq \frac{C_3}{2|z_{\mu,d'}(k)|-R}, $$
where $C_j = C_{j;\Lambda,A,f,g}$ and $C_3 = C_{3;\Lambda,A,f,g}$ are
constants given by \eqref{cts}.
\end{lem}

The next lemma estimates the decay of $\Phi_{d',d''}(k)$ with respect to
$z_{\nu',d}(k)$ for $d' \neq d''$.

\begin{lem}[Decay of $\Phi_{d',d''}(k)$ for $d'\neq d''$] \label{l:b2}
Under Hypothesis \ref{h:small}, let $\nu \in \{1,2\}$ and let $f$ and
$g$ be functions on $\Gdual$ with $\| b^2 f(b) \|_{l^1} < \infty$ and
$\| b^2 g(b) \|_{l^1} < \infty$. Suppose further that $G=\{0, d\}$ and
$k \in (T_\nu(0) \cap T_{\nu'}(d)) \setminus \mathcal{K}_R$. Then, for
$d',d'' \in G$ with $d' \neq d''$,
$$ |\Phi_{d',d''}(k)| \le
\frac{C_{\Gdual,\eps,f,g}}{|z_{\nu',d}(k)|^{3-10^{-1}}}, $$
where $C_{\Gdual,\eps,f,g}$ is a constant.
\end{lem}

The next proposition relates the quantities $|v|$, $|k_2|$,
$|z_{\nu,d'}(k)|$ and $|d|$ for $k$ in the $\eps$-tubes with large
$|v|$.

\begin{prop} \label{p:order}
For $\nu \in \{1,2\}$ we have:
\begin{itemize}
\item[{\rm (i)}] Let $k \in T_\nu(0) \setminus \mathcal{K}_R$. Then,
$$ \frac{1}{|z_{\nu,0}(k)|} \leq \frac{1}{|v|} \leq
\frac{3}{|z_{\nu,0}(k)|} \qquad \text{and} \qquad \frac{1}{4|v|} \leq
\frac{1}{|k_2|} \leq \frac{8}{|v|}. $$

\item[{\rm (ii)}] Let $k \in (T_\nu(0) \cap T_{\nu'}(d)) \setminus
\mathcal{K}_R$. Then,
$$ \frac{1}{|z_{\nu,0}(k)|} \leq \frac{1}{|v|} \leq
\frac{3}{|z_{\nu,0}(k)|}, \qquad \frac{1}{|z_{\nu',d}(k)|} \leq
\frac{1}{|v|} \leq \frac{3}{|z_{\nu',d}(k)|}, $$
$$ \frac{1}{2|z_{\nu',d}(k)|} \leq \frac{1}{|d|} \leq
\frac{2}{|z_{\nu',d}(k)|}. $$
\end{itemize}
\end{prop}


\section{Bounds on the derivatives} \label{s:der}

In the last section we expressed $\Phi_{d',d''}(k)$ as a sum of certain
functions $\alpha_{\mu,d'}^{(j)}(k)$ for $k$ in the $\eps$-tubes with
large $|v|$. In this section we provide bounds for the derivatives of
all these functions. Here we only give the statements. See Appendix
\ref{s:app3} for the proofs.

Our first lemma concerns the derivatives of $\Phi_{d',d''}(k)$.

\begin{lem}[Derivatives of $\Phi_{d',d''}(k)$] \label{l:der}
Under Hypothesis \ref{h:small}, let $f$ and $g$ be functions in
$l^1(\Gdual)$ and suppose either {\rm (i)} or {\rm (ii)} where:
\begin{itemize}
\item[]
\begin{itemize}
\item[\rm (i)] $G=\{0\}$ and $k \in (T_0 \setminus \cup_{b \in G'} T_b)
\setminus \mathcal{K}_R$;

\item[\rm (ii)] $G=\{0, d\}$ and $k \in (T_0 \cap T_d) \setminus
\mathcal{K}_R$.
\end{itemize}
\end{itemize}
Then, for any integers $n$ and $m$ with $n+m \geq 1$ and for any $d',d''
\in G$,
$$ \left| \frac{\partial^{n+m}}{\partial k_1^n \partial k_2^m}
\Phi_{d',d''}(k) \right| \leq \frac{C}{|v|}, $$
where $C$ is a constant with $C = C_{\eps,\Lambda, A, f,g,m,n}$ if {\rm
(i)} or $C = C_{\Lambda,A,f,g,m,n}$ if {\rm (ii)}.
\end{lem}

We now improve the estimate of Lemma \ref{l:der}(ii) for $d'\neq d''$.

\begin{lem}[Derivatives of $\Phi_{d',d''}(k)$ for $d' \neq d''$] 
\label{l:derii}
Consider a constant $\beta \ge 2$ and suppose that $\| |b|^\beta
\hat{q}(b) \|_{l^1} < \infty$ and $\| (1+|b|^\beta) \hat{A}(b) \|_{l^1}
< 2\eps/63$. Let $\nu \in \{1,2\}$ and let $f$ and $g$ be functions on
$\Gdual$ obeying $\| |b|^\beta f(b) \|_{l^1} < \infty$ and $\| |b|^\beta
g(b) \|_{l^1} < \infty$. Suppose further that $G=\{0, d\}$ and $k \in
T_0 \cap T_d$ with $|v| > \frac{2}{\eps} \| |b|^\beta \hat{q}(b)
\|_{l^1}$. Then, for any integers $n$ and $m$ with $n+m \geq 0$ and for
any $d',d'' \in G$ with $d' \neq d''$,
$$ \left| \frac{\partial^{n+m}}{\partial k_1^n \partial k_2^m}
\Phi_{d',d''}(k) \right| \le \frac{C}{|d|^{1+\beta}}, $$
where $C = C_{\eps,\Lambda, A, f,g,m,n}$ is a constant.
\end{lem}

Observe that, in particular, this lemma with $m=n=0$ generalizes Lemma
\ref{l:b2}. We next have bounds for the derivatives of
$\alpha_{\mu,d'}^{(j)}(k)$.

\begin{lem}[Derivatives of $\alpha_{\mu,d'}^{(j)}(k)$] \label{l:der2}
Under Hypothesis \ref{h:small}, let $\nu \in \{1,2\}$ and let $f$ and
$g$ be functions in $l^1(\Gdual)$. Suppose either {\rm (i)} or {\rm
(ii)} where:
\begin{itemize}
\item[]
\begin{itemize}
\item[\rm (i)] $G=\{0\}$ and $k \in (T_\nu(0) \setminus \cup_{b \in G'}
T_b) \setminus \mathcal{K}_R$;

\item[\rm (ii)] $G=\{0, d\}$ and $k \in (T_\nu(0) \cap T_{\nu'}(d))
\setminus \mathcal{K}_R$.
\end{itemize}
\end{itemize}
Then, there is a constant $\rho = \rho_{\eps,A,q,m,n}$ with $\rho \geq
R$ such that, for $|v| \geq \rho$ and for $(\mu,d') = (\nu,0)$ if {\rm
(i)} or $(\mu,d') \in \{ (\nu,0), (\nu',d) \}$ if {\rm (ii)}, for any
integers $n$ and $m$ with $n+m \geq 1$ and for $1 \le j \le 2$,
$$ \left| \frac{\partial^{n+m}}{\partial k_1^n \partial k_2^m}
\alpha_{\mu,d'}^{(j)}(k) \right| \leq \frac{C_j}{(2|z_{\mu,d'}(k)|-R)^j}
\qquad \text{and} \qquad \left| \frac{\partial^{n+m}}{\partial k_1^n
\partial k_2^m} \alpha_{\mu,d'}^{(3)}(k) \right| \leq
\frac{C_3}{|z_{\mu,d'}(k)| R^2}, $$
where $C_l = C_{l;f,g,\Lambda,A,q,n,m}$ for $1 \le l \le 3$ are
constants. Furthermore,
$$ C_{1;f,g,\Lambda, A, 1, 0}, \, C_{1;f,g,\Lambda, A, 0, 1} \leq 13
\Lambda^{-2} \| f \|_{l^1} \| g \|_{l^1} \qquad \text{and} \qquad
C_{1;f,g,\Lambda, A, 1, 1} \leq 65 \Lambda^{-3} \| f \|_{l^1} \| g
\|_{l^1}. $$
\end{lem}


\section{The regular piece} \label{s:reg}

\begin{proof}[Proof of Theorem \ref{t:reg}]
{\tt Step 1 (defining equation).} We first derive a defining equation
for the Fermi curve. Without loss of generality we may assume that
$\hat{A}(0)=0$. Let $G=\{0\}$, recall that $G'=\Gdual \setminus \{0\}$,
and consider the region $( T_\nu(0) \setminus \cup_{b \in G'} T_b )
\setminus \mathcal{K}_\rho$, where $\rho$ is a constant to be chosen
sufficiently large obeying $\rho \geq R$. By Proposition
\ref{p:Gprime}(i) we have $G' = \{ b \in \Gdual \; | \;\; |N_b(k)| \geq
\eps |v| \}$. To simplify the notation write
$$ \mathcal{M}_\nu \coloneqq \left( \Fhat(A,V) \cap T_\nu(0) \right)
\setminus \left( \mathcal{K}_\rho \cup \bigcup_{b \in \Gdual \setminus
\{0\}} T_b \right). $$
By Lemma \ref{l:eqn}(i), a point $k$ is in $\mathcal{M}_\nu$ if and only
if
$$ N_0(k) + D_{0,0}(k) = 0. $$
By Proposition \ref{p:newv}, if we set
$$ w(k) \coloneqq w_{\nu,0}(k) = k_1 + i (-1)^\nu k_2 \qquad \text{and}
\qquad z(k) \coloneqq z_{\nu,0}(k) = k_1 - i (-1)^\nu k_2, $$
this equation becomes
\begin{equation} \label{deq1}
\beta_1 w^2 + \beta_2 z^2 + (1 + \beta_3) wz + \beta_4 w + \beta_5 z +
\beta_6 + \hat{q}{(0)} = 0,
\end{equation}
where
\begin{alignat*}{3}
\beta_1 & \coloneqq J_{\nu'}^{00}, & \qquad \beta_2 & \coloneqq
J_\nu^{00}, & \qquad \beta_3 & \coloneqq K^{00}, \\
\beta_4 & \coloneqq L_{\nu'}^{00}, & \qquad \beta_5 & \coloneqq
L_\nu^{00}, & \qquad \beta_6 & \coloneqq M^{00} - \hat{q}(0),
\end{alignat*}
with $J_\nu^{00}$, $K^{00}$, $L_\nu^{00}$ and $M^{00}$ given by
Proposition \ref{p:newv}. Observe that all the coefficients $\beta_1,
\dots, \beta_6$ have exactly the same form as the function
$\Phi_{0,0}(k)$ of Lemma \ref{l:b1}(i) (see \eqref{Phi}). Thus, by this
lemma, for $1 \leq i \leq 6$ we have
\begin{equation} \label{coef}
\beta_i = \beta_i^{(1)} + \beta_i^{(2)} + \beta_i^{(3)},
\end{equation}
where the function $\beta_i^{(j)}$ is analytic in the region under
consideration with
$$ |\beta_i^{(j)}(k)| \leq \frac{C}{(2|z(k)|-\rho)^j} \leq
\frac{C}{|z(k)|^j} \quad \text{for} \quad 1 \le j \le 2 \qquad
\text{and} \qquad |\beta_i^{(3)}(k)| \leq \frac{C}{|z(k)| \rho^2}, $$
where $C = C_{\eps, \Lambda, q, A}$ is a constant. The exact expression
for $\beta_i^{(j)}$ can be easily obtained from the definitions and from
Lemma \ref{l:b1}(i). Substituting \eqref{coef} into \eqref{deq1} and
dividing both sides of the equation by $z$ yields
\begin{equation} \label{deq1b}
w + \beta_2^{(1)} \, z + g = 0,
\end{equation}
where
\begin{equation} \label{defg}
g \coloneqq \frac{\beta_1 w^2}{z} + (\beta_2^{(2)} + \beta_2^{(3)}) z +
\beta_3 w + \frac{\beta_4 w}{z} + \beta_5 + \frac{\beta_6}{z} +
\frac{\hat{q}(0)}{z}
\end{equation}
obeys
\begin{equation} \label{decg}
|g(k)| \leq \frac{C}{\rho},
\end{equation}
with a constant $C = C_{\eps, \Lambda, q, A}$. Therefore, a point $k$ is
in $\mathcal{M}_\nu$ if and only if
$$ F(k) = 0, $$
where
$$ F(k) \coloneqq w(k) + \beta_2^{(1)}(k) \, z(k) + g(k) $$
is an analytic function (in the region under consideration).

{\tt Step 2 (candidates for a solution).} Let us now identify which
points are candidates to solve the equation $F(k)=0$. First observe
that, by Proposition \ref{p:lines}(c) the lines $\mathcal{N}_\nu(0)$ and
$\mathcal{N}_{\nu'}(d)$ intersect at $\mathcal{N}_\nu(0) \cap
\mathcal{N}_{\nu'}(d) = \{ ( i \theta_\nu(d), (-1)^{\nu'} \theta_\nu(d))
\}$. Hence, the second coordinate of this point and the second
coordinate of a point $k$ differ by
$$ pr(k) - pr( \mathcal{N}_\nu(0) \cap \mathcal{N}_{\nu'}(d) ) = k_2 -
(-1)^{\nu'} \theta_\nu(d) = k_2 + (-1)^\nu \theta_\nu(d). $$
Now observe that, if $k \in T_\nu(0) \cap T_{\nu'}(d)$ then $|k_1 +
i(-1)^\nu k_2| < \eps$ and
\begin{align*}
|k_2 + (-1)^\nu \theta_\nu(d)| & = \big| \tfrac{1}{2} (k_1+i(-1)^\nu
k_2) - \tfrac{1}{2} ( k_1 + d_1 - i(-1)^\nu (k_2+d_2) \big| \\ & \leq
\tfrac{1}{2} \big| N_{0,\nu}(k) - N_{d,\nu'}(k) \big | < \tfrac{\eps}{2}
+ \tfrac{\eps}{2} = \eps.
\end{align*}
That is, the second coordinate of $k$ and the second coordinate of
$\mathcal{N}_\nu(0) \cap \mathcal{N}_{\nu'}(d)$ must be apart from each
other by at most $\eps$. This gives a necessary condition on the second
coordinate of a point $k$ for being in $\mathcal{M}_\nu$. Conversely, if
a point $k$ is in the $(\eps/4)$-tube inside $T_\nu(0)$, that is, $|k_1
+ i(-1)^\nu k_2| < \frac{\eps}{4}$, and its second coordinate differ
from the second coordinate of $\mathcal{N}_\nu(0) \cap
\mathcal{N}_{\nu'}(d)$ by at most $\eps/4$, that is, $|k_2 + (-1)^\nu
\theta_\nu(d)| < \frac{\eps}{4}$, then
$$ |N_{d,\nu'}(k)| = \big| N_{0,\nu}(k) - 2 (k_2 + (-1)^\nu
\theta_\nu(d))| \leq \frac{\eps}{4} + 2 \, \frac{\eps}{4} < \eps, $$
that is, the point $k$ is also in $T_{\nu'}(d)$ and hence lie in the
intersection $T_\nu(0) \cap T_{\nu'}(d)$. This gives a sufficient
condition on the first and second coordinates of a point $k$ for being
in $T_\nu(0) \cap T_{\nu'}(d)$.

For $y \in \C$ define the set of candidates for a solution of $F(k)=0$
as
$$ M_\nu(y) \coloneqq pr^{-1}(y) \cap \left( T_\nu(0) \setminus
\bigcup_{b \in \Gdual \setminus \{0\}} T_b \right) = pr^{-1}(y) \cap
\left( T_\nu(0) \setminus \bigcup_{b \in \Gdual \setminus \{0\}}
T_{\nu'}(b) \right). $$
Observe that, if $|y + (-1)^\nu \theta_\nu(b)| \geq \eps$ for all $b \in
\Gdual \setminus \{0\}$ then
\begin{equation} \label{inc1}
M_\nu(y) = pr^{-1}(y) \cap T_\nu(0) = \{ (k_1, y) \in \C^2 \; | \;\;
|k_1 + i(-1)^\nu y | < \eps \}.
\end{equation}
On the other hand, if $|y + (-1)^\nu \theta_\nu(d)| < \eps$ for some $d
\in \Gdual \setminus \{0\}$, then there is at most one such $d$ and
consequently
\begin{equation} \label{inc2}
\begin{aligned}
M_\nu(y) & = pr^{-1}(y) \cap (T_\nu(0) \setminus T_{\nu'}(d)) \\ & = \{
(k_1, y) \in \C^2 \; | \;\; |k_1 + i(-1)^\nu y | < \eps \, \text{ and }
\, |k_1+d_1 + i(-1)^{\nu'} (y+d_2)| \geq \eps \}.
\end{aligned}
\end{equation}
Indeed, suppose there is another $d' \neq 0$ such that $|y + (-1)^\nu
\theta_\nu(d')| < \eps$. Then,
$$ |d-d'| = |2(-1)^\nu \theta_\nu(d-d')| = |y + (-1)^\nu \theta_\nu(d) -
(y + (-1)^\nu \theta_\nu(d'))| \leq 2\eps < 2 \Lambda, $$
which contradicts the definition of $\Lambda$. Thus, there is no such
$d' \neq 0$.

{\tt Step 3 (uniqueness).} We now prove that, given $k_2$, if there
exists a solution $k_1(k_2)$ of $F(k_1,k_2)=0$, then this solution is
unique and it depends analytically on $k_2$. This follows easily using
the implicit function theorem and the estimates below, which we prove
later.

\begin{prop} \label{p:imp}
Under the hypotheses of Theorem \ref{t:reg} we have
\begin{align}
|F(k)-w(k)| & \leq \frac{\eps}{900} + \frac{C_1}{\rho}, \tag{a} \\
\left| \frac{\partial F}{\partial k_1}(k) - 1 \right| & \leq \frac{1}{7
\cdot 3^4} + \frac{C_2}{\rho}, \tag{b}
\end{align}
where the constants $C_1$ and $C_2$ depend only on $\eps$, $\Lambda$,
$q$ and $A$.
\end{prop}

Now suppose that $(k_1,y) \in M_\nu(y)$. Then,
$$ \left| \frac{\partial F}{\partial k_1}(k_1,y) - 1 \right| \leq
\frac{1}{7 \cdot 3^4} + \frac{C_2}{\rho}. $$
Hence, by the implicit function theorem, by choosing the constant $\rho
\geq R$ sufficiently large, if $F(k_1^*,y)=0$ for some $(k_1^*,y) \in
M_{\nu}(y)$, then there is a neighbourhood $U \times V \subset \C^2$
which contains $(k_1^*,y)$, and an analytic function $\eta : V \to U$
such that $F(k_1,k_2) = 0$ for all $(k_1,k_2) \in U \times V$ if and
only if $k_1 = \eta(k_2)$. In particular this implies that the equation
$F(k_1,k_2)=0$ has at most one solution $(\eta(y),y)$ in $M_\nu(y)$ for
each $y \in \C$. We next look for conditions on $y$ to have a solution
or have no solution in $M_\nu(y)$.

{\tt Step 4 (existence).} We first state an improved version of
Proposition \ref{p:imp}(a).

\begin{prop} \label{p:imp2}
Under the hypotheses of Theorem \ref{t:reg} we have
$$ F(k)-w(k) = \beta_2^{(1,0)} + \beta_2^{(1,1)}(w(k)) +
\beta_2^{(1,2)}(k) + h(k), $$
where
\begin{equation} \label{cte1}
\beta^{(1,0)}_2 = -2i \sum_{b,c \in G'_1} \frac{
\theta_{\nu'}(\hat{A}(-b))}{\theta_{\nu'}(b)} \left[ \delta_{b,c} +
\frac{\theta_{\nu'}(\hat{A}(b-c))}{ \theta_{\nu'}(c)} \right]
\theta_\nu(\hat{A}(c))
\end{equation}
is a constant that depends only on $\rho$ and $A$ and
$$ h \coloneqq \beta_2^{(1,3)} + g. $$
Furthermore,
\begin{align*}
|\beta^{(1,0)}_2| < \frac{1}{100 \Lambda} \, \eps^2, & \qquad \quad
|\beta^{(1,1)}_2(k)| < \frac{1}{40 \Lambda^2} \, \eps^3, \\
|\beta^{(1,2)}_2(k)| < \frac{1}{7^4 \Lambda^3} \, \eps^4, & \qquad \quad
|h(k)| \leq C_{\eps,\Lambda,q,A} \, \frac{1}{\rho}.
\end{align*}
\end{prop}

We now derive conditions for the existence of solutions. Suppose that
$F(\eta(y),y)=0$. Then, since $\eta(y) + i(-1)^\nu y=w(\eta(y),y)$ and
$\eps<\Lambda/6$, using the above proposition we obtain
\begin{align*}
|\eta(y) + i(-1)^\nu y| = |w(\eta(y),y)| & = |F(\eta(y),y) -
w(\eta(y),y)| \\ & \leq \frac{\eps^2}{100 \Lambda} + \frac{\eps^3}{40
\Lambda^2} + \frac{\eps^4}{7^4 \Lambda^3} + \frac{C}{\rho} \leq
\frac{\eps^2}{50 \Lambda} + \frac{C}{\rho}.
\end{align*}
Hence, by choosing the constant $\rho$ sufficiently large we find that
$$ |\eta(y) + i(-1)^\nu y| < \frac{\eps^2}{40 \Lambda}. $$
In view of \eqref{inc2}, there is no solution in $M_\nu(y)$ if for some
$d \in \Gdual \setminus \{0\}$ we have
$$ |y+(-1)^\nu \theta_\nu(d)|<\eps \qquad \text{and} \qquad |\eta(y) +
d_1 + i(-1)^{\nu'}(y+d_2) | < \eps. $$
This happens if
$$ |y+(-1)^\nu \theta_\nu(d)| \leq \frac{1}{2}\left( \eps -
\frac{\eps^2}{40\Lambda} \right) $$
because in this case
\begin{align*}
| \eta(y) + d_1 + i(-1)^{\nu'}(y+d_2) | & = |\eta(y) + i(-1)^\nu y -
2i(-1)^\nu y + d_1 - i(-1)^\nu d_2| \\
& \leq |\eta(y) + i(-1)^\nu y| + 2 |y+(-1)^\nu \theta_\nu(d)| < \eps.
\end{align*}
Therefore, the image set of $pr$ is contained in
$$ \Omega_1 \coloneqq \Bigg \{ z \in \C \,\, \Bigg | \,\,
|z+(-1)^\nu\theta_\nu(b)| > \frac{1}{2}\left( \eps -
\frac{\eps^2}{40\Lambda} \right) \, \text{ for all } b \in \Gdual
\setminus \{0\} \Bigg \}. $$
On the other hand, in view of \eqref{inc1}, there is a solution in
$M_\nu(y)$ if $|y+(-1)^\nu \theta_\nu(b)| > \eps$ for all $b \in \Gdual
\setminus \{0\}$. Recall from Proposition \ref{p:order}(a) that $\rho <
|v| < 8|k_2|$. Thus, the image set of $pr$ contains the set
$$ \Omega_2 \coloneqq \Big \{ z \in \C \,\, \Big| \,\, 8 |z| > \rho \,
\text{ and } |z+(-1)^\nu \theta_\nu(b)| > \eps \, \text{ for all } b \in
\Gdual \setminus \{0\} \Big \}. $$

{\tt Step 5.} Summarizing, we have the following
biholomorphic correspondence:
\begin{align*}
\mathcal{M}_\nu \ni k \, \xrightarrow{\qquad pr \qquad} \,\, k_2 & \in
\Omega, \\
\mathcal{M}_\nu \ni (\eta(y),y) \, \xleftarrow{\qquad pr^{-1} \qquad}
\,\, y & \in \Omega,
\end{align*}
where
$$ \Omega_2 \subset \Omega \subset \Omega_1 \qquad \text{and} \qquad
\eta(y) = -\beta_2^{(1,0)} -i(-1)^\nu y - r(y), $$
with the constant $\beta_2^{(1,0)}$ given by \eqref{cte1},
$$ |\beta^{(1,0)}_2| < \frac{\eps^2}{100 \Lambda} \qquad \text{and}
\qquad |r(y)| \leq \frac{\eps^3}{50 \Lambda^2} + \frac{C}{\rho}. $$
This completes the proof of the theorem.
\end{proof}


\begin{proof}[Proof of Proposition \ref{p:imp}]
(a) Recall that $\beta_2 = J_\nu^{00}$. First observe that, by
Proposition \ref{p:newv}, Lemma \ref{l:b1}, and \eqref{alpha}, we have
\begin{equation} \label{beta21}
\beta_2^{(1)}(k) = (J_\nu^{00})^{(1)}(k) = \sum_{b,c \in G'_1}
\frac{(1,i(-1)^\nu) \cdot \hat{A}(-b)}{N_b(k)} \, S_{b,c} \,
(1,-i(-1)^\nu) \cdot \hat{A}(c).
\end{equation}
Thus, by \eqref{rec} and \eqref{estS},
\begin{equation} \label{as1}
|\beta_2^{(1)}(k)| \leq \sqrt{2} \| \hat{A} \|_{l_1} \frac{2}{\Lambda
(2|z(k)|-R)} \, \frac{45}{44} \, \sqrt{2} \| \hat{A} \|_{l_1} \leq
\frac{4}{\Lambda |z(k)|} \, \frac{44}{45} \, \left( \frac{2\eps}{63}
\right)^2 \leq \frac{\eps}{900} \, \frac{1}{|z(k)|}.
\end{equation}
Now recall that $|g(k)| \leq C_{\eps,\Lambda,q,A} \, \frac{1}{\rho}$.
Hence,
$$ |F(k) - w(k)| = |\beta_2^{(1)}(k) z(k) + g(k)| \leq \frac{\eps}{900}
+ C_{\eps,\Lambda,q,A} \, \frac{1}{\rho}. $$
This proves part (a).

(b) We first compute
\begin{equation} \label{ee1}
\begin{aligned}
\frac{\partial g}{\partial k_1} & = \frac{\partial \beta_1}{\partial
k_1} \frac{w^2}{z} + \beta_1 \, \frac{2wz-w^2}{z^2} + \left(
\frac{\partial \beta_2^{(2)}}{\partial k_1} + \frac{\partial
\beta_2^{(3)}}{\partial k_1} \right) z + \beta_2^{(2)} + \beta_2^{(3)} +
\frac{\partial \beta_3}{\partial k_1} w + \beta_3 \\
& \quad + \frac{\partial \beta_4}{\partial k_1} \, \frac{w}{z} + \beta_4
\, \frac{z-w}{z^2} + \frac{\partial \beta_5}{\partial k_1} +
\frac{\partial \beta_6}{\partial k_1} \, \frac{1}{z} -
\frac{\beta_6}{z^2} - \frac{\hat{q}(0)}{z^2}.
\end{aligned}
\end{equation}
Now observe that, since $k \in T_\nu(0) \setminus \mathcal{K}_\rho$ we
have $|w(k)| < \eps$, $3|v| \geq |z|$ and $\rho < |v| \leq |z|$.
Furthermore, by Lemmas \ref{l:b1}(i), \ref{l:der}(i) and
\ref{l:der2}(i), for $1 \leq i \leq 6$ and $1 \leq j \leq 2$,
\begin{equation} \label{ee2}
\begin{aligned}
|\beta_i(k)| & \leq \frac{C}{|z(k)|}, \qquad & |\beta_i^{(j)}(k)| & \leq
\frac{C}{|z(k)|^j}, \qquad & |\beta_i^{(3)}(k)| & \leq \frac{C}{|z(k)|
\rho^2}, \\
\left| \frac{\partial \beta_i(k)}{\partial k_1} \right| & \leq
\frac{C}{|z(k)|}, \qquad & \left| \frac{\partial
\beta_i^{(j)}(k)}{\partial k_1} \right| & \le \frac{C}{|z(k)|^j}, \qquad
& \left| \frac{\partial \beta_i^{(3)}(k)}{\partial k_1} \right| & \le
\frac{C}{|z(k)| \rho^2},
\end{aligned}
\end{equation}
where $C = C_{\eps,\Lambda,q,A}$ in all cases. Hence,
\begin{equation} \label{ee3}
\left| \frac{\partial g(k)}{\partial k_1} \right| \leq
C_{\eps,\Lambda,q,A} \, \frac{1}{\rho}.
\end{equation}
By Lemma \ref{l:der2}(i) with $f=g=(1,-i(-1)^\nu) \cdot \hat{A}$, we
obtain
\begin{equation} \label{ee4}
\left | z(k) \frac{ \partial \beta_2^{(1)}(k)}{\partial k_1} \right|
\leq |z(k)| \frac{13}{\Lambda^2 |z(k)|} \| (1,-i(-1)^\nu) \cdot \hat{A}
\|_{l^1}^2 \leq \frac{26}{\Lambda^2} \| \hat{A} \|_{l^1}^2 < \frac{1}{7
\cdot 3^4}. 
\end{equation}
Therefore,
\begin{align*}
\left| \frac{\partial F}{\partial k_1}(k) - 1 \right| & = \left|
\frac{\partial}{\partial k_1} ( F(k)-w(k) ) \right| = \left|
\frac{\partial}{\partial k_1} ( \beta_2^{(1)}(k) z(k) + g(k) ) \right|\\
& = \left| \frac{ \partial \beta_2^{(1)}}{\partial k_1}(k) z(k) +
\beta_2^{(1)}(k) + \frac{\partial g}{\partial k_1}(k) \right| \leq
\frac{1}{7 \cdot 3^4} + C_{\eps,\Lambda,q,A} \, \frac{1}{\rho}.
\end{align*}
This proves part (b) and completes the proof of the proposition.
\end{proof}


\begin{proof}[Proof of Proposition \ref{p:imp2}]
First observe that
$$ (1,i(-1)^\nu) \cdot A = A_1 + i(-1)^\nu A_2 = A_1 - i(-1)^{\nu'} A_2
= -2i \theta_{\nu'}(A). $$
Thus, recalling \eqref{beta21},
$$ \beta_2^{(1)}(k) = (J_\nu^{00})^{(1)}(k) = \sum_{b,c \in G'_1}
\frac{2i \theta_{\nu'}(\hat{A}(-b))}{N_b(k)} \, S_{b,c} \, 2i
\theta_\nu(\hat{A}(c)). $$
Now, by Lemma \ref{l:b1b} we have
$$ z(k) \beta_2^{(1)}(k) = \beta_2^{(1,0)} + \beta_2^{(1,1)}(w(k)) +
\beta_2^{(1,2)}(k) + \beta_3^{(1,3)}(k), $$
where
$$ \beta_2^{(1,0)} = - 2i \sum_{b,c \in G_1'}
\frac{\theta_{\nu'}(\hat{A}(-b))}{\theta_{\nu'}(b)} \left[ \delta_{b,c}
+ \frac{\theta_{\nu'}(\hat{A}(b-c))}{\theta_{\nu'}(c)} \right]
\theta_{\nu}(\hat{A}(c)) $$
and
$$ |\beta_3^{(1,3)}(k)| \leq C_{\Lambda,A} \, \frac{1}{|z(k)|} <
C_{\Lambda,A} \, \frac{1}{\rho}. $$
Hence,
$$ F(k)-w(k) = z(k) \beta_2^{(1)}(k) + g(k) = \beta_2^{(1,0)} +
\beta_2^{(1,1)}(w(k)) + \beta_2^{(1,2)}(k) + h(k) $$
with $h \coloneqq \beta_3^{(1,3)} + g$. Furthermore, in view of
\eqref{decg},
$$ |h(k)| \leq |\beta_3^{(1,3)}(k)| + |g(k)| < C_{\eps,\Lambda,q,A} \,
\frac{1}{\rho}. $$
This proves the first part of the proposition. Finally, by \eqref{cts},
since $\| \hat{A} \|_{l^1} < 2\eps/63$ and $\eps<\Lambda/6$, we find
that
$$ |\beta_2^{(1,0)}| \leq \frac{1}{2 \Lambda} \left( 1 + \frac{1}{2
\Lambda} \| \theta_{\nu'}(\hat{A}) \|_{l^1} \right) \| 2i
\theta_{\nu'}(\hat{A}) \|_{l^1} \| 2i \theta_{\nu}(\hat{A}) \|_{l^1}
\leq \frac{4}{\Lambda} \| \hat{A} \|_{l^1}^2 < \frac{1}{100 \Lambda} \,
\eps^2, $$
$$ |\beta_2^{(1,1)}| \leq \frac{\eps}{\Lambda^2} \left( 1 + \frac{7}{6
\Lambda} \| \theta_{\nu'}(\hat{A}) \|_{l^1} \right) \| 2i
\theta_{\nu'}(\hat{A}) \|_{l^1} \| 2i \theta_{\nu}(\hat{A}) \|_{l^1}
\leq \frac{8}{\Lambda^2} \, \eps \| \hat{A} \|_{l^1}^2 < \frac{1}{40
\Lambda^2} \, \eps^3 $$
and
$$ |\beta_2^{(1,2)}| \leq \frac{64}{\Lambda^3} \| \theta_{\nu'}(\hat{A})
\|_{l^1}^2 \| 2i \theta_{\nu'}(\hat{A}) \|_{l^1} \| 2i
\theta_{\nu}(\hat{A}) \|_{l^1} \leq \frac{256}{\Lambda^3} \| \hat{A}
\|_{l^1}^4 < \frac{1}{7^4 \Lambda^3} \, \eps^4. $$
This completes the proof.
\end{proof}


\section{The handles} \label{s:hand}

\begin{proof}[Proof of Theorem \ref{t:hand}]
{\tt Step 1 (defining equation).} Let $G=\{0,d\}$ and consider the
region $(T_\nu(0) \cap T_{\nu'}(d)) \setminus \mathcal{K}_\rho$, where
$\rho$ is a constant to be chosen sufficiently large obeying $\rho \ge
R$. Observe that, this requires $d$ being sufficiently large for
$(T_\nu(0) \cap T_{\nu'}(d)) \setminus \mathcal{K}_\rho$ being not
empty. In fact, by Proposition \ref{p:order}(ii), for $k$ in this region
we have $\rho < |v| \le 2|d|$. Now, recall from Proposition
\ref{p:Gprime}(ii) that $G' = \{ b \in \Gdual \; | \;\, |N_b(k)| \geq
\eps |v| \}$, and to simplify the notation write
$$ \mathcal{H}_\nu \coloneqq \Fhat(A,V) \cap (T_\nu(0) \cap T_{\nu'}(d))
\setminus \mathcal{K}_\rho. $$
By Lemma \ref{l:eqn}(ii), a point $k$ is in $\mathcal{H}_\nu$ if and
only if
\begin{equation} \label{eqh}
(N_0(k)+D_{0,0}(k)) (N_d(k)+D_{d,d}(k)) - D_{0,d}(k)D_{d,0}(k) = 0.
\end{equation}
Define
\begin{equation} \label{ch1}
\begin{aligned}
w_1(k) & \coloneqq w_{\nu,0} = k_1 + i (-1)^\nu k_2, \\
z_1(k) & \coloneqq z_{\nu,0} = k_1 - i (-1)^\nu k_2, \\
w_2(k) & \coloneqq w_{\nu',d} = k_1+d_1 + i(-1)^{\nu'} (k_2+d_2), \\
z_2(k) & \coloneqq z_{\nu',d} = k_1+d_1 - i(-1)^{\nu'} (k_2+d_2).
\end{aligned}
\end{equation}
Note that, by Proposition \ref{p:order}(ii),
$$ |v| \le |z_1| \le 3|v|, \qquad |v| \le |z_2| \le 3|v| \qquad
\text{and} \qquad |d| \le |z_2| \le 2|d|. $$
By Proposition \ref{p:newv},
\begin{equation} \label{term1}
\begin{aligned}
N_0 + D_{0,0} & = \beta_1 w_1^2 + \beta_2 z_1^2 + (1 + \beta_3) w_1 z_1
+ \beta_4 w_1 + \beta_5 z_1 + \beta_6 + \hat{q}{(0)}, \\
N_d + D_{d,d} & = \eta_1 w_2^2 + \eta_2 z_2^2 + (1 + \eta_3) w_2 z_2 +
\eta_4 w_2 + \eta_5 z_2 + \eta_6 + \hat{q}{(0)},
\end{aligned}
\end{equation}
where
\begin{alignat*}{3}
\beta_1 & \coloneqq J_{\nu'}^{00}, & \qquad \beta_2 & \coloneqq
J_\nu^{00}, & \qquad \beta_3 & \coloneqq K^{00}, \\
\beta_4 & \coloneqq L_{\nu'}^{00}, & \qquad \beta_5 & \coloneqq
L_\nu^{00}, & \qquad \beta_6 & \coloneqq M^{00} - \hat{q}(0),
\end{alignat*}
and
\begin{alignat*}{3}
\eta_1 & \coloneqq J_\nu^{dd}, & \qquad \eta_2 & \coloneqq
J_{\nu'}^{dd}, & \qquad \eta_3 & \coloneqq K^{dd}, \\
\eta_4 & \coloneqq L_\nu^{dd}, & \qquad \eta_5 & \coloneqq
L_{\nu'}^{dd}, & \qquad \eta_6 & \coloneqq M^{dd} - \hat{q}(0),
\end{alignat*}
with $J_\nu^{d'd'}$, $K^{d'd'}$, $L_\nu^{d'd'}$ and $M^{d'd'}$ given by
Proposition \ref{p:newv}. Observe that all the coefficients $\beta_1,
\dots, \beta_6$ and $\eta_1, \dots, \eta_6$ have exactly the same form
as the function $\Phi_{d',d'}(k)$ of Lemma \ref{l:b1}(ii) (see
\eqref{Phi}). Thus, by this lemma, for $1 \leq i \leq 6$ we have
\begin{equation} \label{coef2}
\beta_i = \beta_i^{(1)} + \beta_i^{(2)} + \beta_i^{(3)} \qquad
\text{and} \qquad \eta_i = \eta_i^{(1)} + \eta_i^{(2)} + \eta_i^{(3)},
\end{equation}
where the functions $\beta_i^{(j)}$ and $\eta_i^{(j)}$ are analytic in
the region under consideration with
\begin{align*}
|\beta_i^{(j)}(k)| \leq \frac{C}{(2|z_1(k)|-\rho)^j} \leq
\frac{C}{|z_1(k)|^j} \quad \text{for} \quad 1 \le j \le 2 \qquad
\text{and} \qquad |\beta_i^{(3)}(k)| \leq \frac{C}{|z_1(k)| \rho^2}, \\
|\eta_i^{(j)}(k)| \leq \frac{C}{(2|z_2(k)|-\rho)^j} \leq
\frac{C}{|z_2(k)|^j} \quad \text{for} \quad 1 \le j \le 2 \qquad
\text{and} \qquad |\eta_i^{(3)}(k)| \leq \frac{C}{|z_2(k)| \rho^2}, 
\end{align*}
where $C = C_{\eps,\Lambda, q, A}$ is a constant. The exact expressions
for $\beta_i^{(j)}$ and $\eta_i^{(j)}$ can be easily obtained from the
definitions and from Lemma \ref{l:b1}(ii). Substituting \eqref{coef2}
into \eqref{term1} yields
\begin{equation} \label{term1b}
\begin{aligned}
\frac{1}{z_1} (N_0 + D_{0,0}) & = w_1 + \beta_2^{(1)} \, z_1 + g_1, \\
\frac{1}{z_2} (N_d + D_{d,d}) & = w_2 + \eta_2^{(1)} \, z_2 + g_2,
\end{aligned}
\end{equation}
where
\begin{equation} \label{defg1g2}
\begin{aligned}
g_1 & \coloneqq \frac{\beta_1 w_1^2}{z_1} + (\beta_2^{(2)} +
\beta_2^{(3)}) z_1 + \beta_3 w_1 + \frac{\beta_4 w_1}{z_1} + \beta_5 +
\frac{\beta_6}{z_1} + \frac{\hat{q}(0)}{z_1}, \\
g_2 & \coloneqq \frac{\eta_1 w_2^2}{z_2} + (\eta_2^{(2)} + \eta_2^{(3)})
z_2 + \eta_3 w_2 + \frac{\eta_4 w_2}{z_2} + \eta_5 + \frac{\eta_6}{z_2}
+ \frac{\hat{q}(0)}{z_2}
\end{aligned}
\end{equation}
obey
\begin{equation} \label{decg1g2}
|g_1(k)| \le \frac{C}{\rho} \qquad \text{and} \qquad |g_2(k)| \le
\frac{C}{\rho},
\end{equation}
with a constant $C = C_{\eps,\Lambda, q,A}$. This gives us more
information about the first term in \eqref{eqh}. We next consider the
second term in that equation.

Write
\begin{equation} \label{term2}
D_{0,d} = c_1(d) + p_1 \qquad \text{and} \qquad D_{d,0} = c_2(d) + p_2
\end{equation}
with
\begin{alignat*}{2}
c_1(d) & \coloneqq \hat{q}(-d) - 2d \cdot \hat{A}(-d), \qquad \quad &
p_1 & \coloneqq D_{0,d} - \hat{q}(-d) + 2d \cdot \hat{A}(-d), \\
c_2(d) & \coloneqq \hat{q}(d) + 2d \cdot \hat{A}(d), \qquad \quad & p_2
& \coloneqq D_{d,0} - \hat{q}(d) - 2d \cdot \hat{A}(d).
\end{alignat*}
We have the following estimates.

\begin{prop} \label{p:p1p2}
Under the hypotheses of Theorem \ref{t:hand} we have, for any integers
$n$ and $m$ with $n+m \ge 0$ and for $1 \le j \le 2$,
$$ \left| \frac{\partial^{n+m}}{\partial k_1^n \partial k_2^m} \, p_j(k)
\right| \le \frac{C_1}{|d|} \qquad \text{and} \qquad |c_j(d)| \le
\frac{C_2}{|d|}, $$
where the constants $C_1$ and $C_2$ depend only on $\eps$, $\Lambda$,
$q$ and $A$.
\end{prop}

Thus, by dividing both sides of \eqref{eqh} by $z_1z_2$ and substituting
\eqref{term1b} and \eqref{term2} we find that
\begin{equation} \label{eqh2}
\begin{aligned} 0 & = \frac{1}{z_1 z_2} \big[ (N_0+D_{0,0})
(N_d+D_{d,d}) - D_{0,d} D_{d,0} \big] \\
& = (w_1 + \beta_2^{(1)} \, z_1 + g_1)(w_2 + \eta_2^{(1)} \, z_2 + g_2)
- \frac{1}{z_1 z_2} (c_1(d) + p_1)(c_2(d) + p_2).
\end{aligned}
\end{equation}

We now introduce a (nonlinear) change of variables in $\C^2$. Set
\begin{equation} \label{x1x2}
\begin{aligned}
x_1(k) & \coloneqq w_1(k) + \beta_2^{(1)}(k) \, z_1(k) + g_1(k), \\
x_2(k) & \coloneqq w_2(k) + \eta_2^{(1)}(k) \, z_2(k) + g_2(k).
\end{aligned}
\end{equation}
This transformation obeys the following estimates.

\begin{prop} \label{p:jac}
Under the hypotheses of Theorem \ref{t:hand} we have:

\vspace{0.2cm}

\noindent {\rm (i)} For $1 \leq j \leq 2$ and for $\rho$ sufficiently
large,
$$ |x_j(k) - w_j(k)| \le \frac{\eps}{900} + \frac{C}{\rho} <
\frac{\eps}{8}. $$

\noindent {\rm (ii)}
\begin{align*}
\begin{pmatrix}
\frac{\partial x_1}{\partial k_1} & \frac{\partial x_1}{\partial k_2}\\
\frac{\partial x_2}{\partial k_1} & \frac{\partial x_2}{\partial k_2}
\end{pmatrix} & = \begin{pmatrix} 1 & i(-1)^\nu \\ 1 & i(-1)^{\nu'}
\end{pmatrix} ( I + M ) \shortintertext{and}
\begin{pmatrix}
\frac{\partial k_1}{\partial x_1} & \frac{\partial k_1}{\partial x_2}\\
\frac{\partial k_2}{\partial x_1} & \frac{\partial k_2}{\partial x_2}
\end{pmatrix} & = \frac{1}{2} \begin{pmatrix} 1 & 1 \\ i(-1)^{\nu'} &
i(-1)^\nu \end{pmatrix} (I + N)
\end{align*}
with
$$ \| M \| \leq \frac{4}{7 \cdot 3^4} + \frac{C}{\rho} < \frac{1}{2}
\qquad \text{and} \qquad \| N \| \le 4 \| M \|. $$
Furthermore, for all $m,i,j \in \{1,2\}$,
$$ \left| \frac{\partial^2 k_m}{\partial x_i \partial x_j} \right| \le
\frac{3}{\Lambda^3} \, \eps^2 +  \frac{C}{\rho}. $$
Here, all the constants $C$ depend only on $\eps$, $\Lambda$, $q$ and
$A$.  \end{prop}

By the inverse function theorem, these estimates imply that the above
transformation is invertible. Therefore, by rewriting the equation
\eqref{eqh2} in terms of these new variables, we conclude that a point
$k$ is in $\mathcal{H}_\nu$ if and only if $x_1(k)$ and $x_2(k)$ satisfy
the equation
\begin{equation} \label{eqh3}
x_1 x_2 + r(x_1,x_2) = 0,
\end{equation}
where
$$ r(x_1,x_2) \coloneqq - \frac{1}{z_1 z_2} (c_1(d) + p_1)(c_2(d) +
p_2). $$
In order to study this defining equation we need some estimates.

{\tt Step 2 (estimates).} Using the above inequalities we have, for
$i,j,l \in \{1,2\}$,
$$ \left| \frac{\partial}{\partial x_i} \, p_j(k(x)) \right| \le
\sum_{m=1}^2 \left| \frac{\partial p_j}{\partial k_m} \, \frac{\partial
k_m}{\partial x_i} \right| \le \frac{C}{|d|} $$
and
$$ \left| \frac{\partial^2}{\partial x_i \partial x_l} \, p_j(k(x))
\right| \le \sum_{m,n=1}^2 \left| \frac{\partial^2 p_j}{\partial k_m
\partial k_n} \, \frac{\partial k_m}{\partial x_i} \, \frac{\partial
k_n}{\partial x_l} \right| + \sum_{m=1}^2 \left| \frac{\partial
p_j}{\partial k_m} \, \frac{\partial^2 k_m}{\partial x_i \partial x_l}
\right| \le \frac{C}{|d|}, $$
so that
$$ |r(x)| \le C \, \frac{1}{|d|^2} \, \frac{1}{|d|} \, \frac{1}{|d|} \le
\frac{C}{|d|^4}, $$
$$ \left| \frac{\partial}{\partial x_i} \, r(x) \right| \le C \,
\frac{1}{|d|^3} \, \frac{1}{|d|} \, \frac{1}{|d|} + C \, \frac{1}{|d|^2}
\, \frac{1}{|d|} \, \frac{1}{|d|} \le \frac{C}{|d|^4} $$
and
$$ \left| \frac{\partial^2}{\partial x_i \partial x_j} \, r(x) \right|
\le \frac{C}{|d|^4}. $$
Here, all the constants depend only on $\eps$, $\Lambda$, $q$ and $A$.

{\tt Step 3 (Morse lemma).} We now apply the quantitative Morse lemma in
Appendix \ref{s:morse} for studying the equation \eqref{eqh3}. We
consider this lemma with $a = b = C/|d|^4$, $\delta = \eps$, and $d$
sufficiently large so that $b < \max \{ \frac{2}{3} \, \frac{1}{55},
\frac{\eps}{4} \}$. Observe that, under this condition we have
$$ (\delta-a)(1-19b) > \frac{\eps}{2} \qquad \text{and} \qquad
(\delta-a)(1-55b) > \frac{\eps}{4}. $$
According to this lemma, there is a biholomorphism $\Phi_\nu$ defined on
$$ \Omega_1 \coloneqq \left \{ (z_1,z_2) \in \C^2 \; \Big | \;\; |z_1| <
\frac{\eps}{2} \; \text{ and } \; |z_2| < \frac{\eps}{2} \right\} $$
with range containing
\begin{equation} \label{range}
\left \{ (x_1,x_2) \in \C^2 \; \Big | \;\; |x_1| < \frac{\eps}{4} \;
\text{ and } \; |x_2| < \frac{\eps}{4} \right \}
\end{equation}
such that
\begin{equation} \label{propPhi}
\begin{aligned}
\| D \Phi_\nu - I \| & \le \frac{C}{|d|^2}, \\
((x_1 x_2 + r) \circ \Phi_\nu)(z_1,z_2) & = z_1 z_2 + t_d, \\
|t_d| & \le \frac{C}{|d|^4}, \\
|\Phi_\nu(0)| & \le \frac{C}{|d|^4},
\end{aligned}
\end{equation}
where $D \Phi_\nu$ is the derivative of $\Phi_\nu$ and $t_d$ is a
constant that depends on $d$. Hence, if for $\nu=1$ we define
$$ \phi_{d,1} \, : \, \Omega_1 \longrightarrow T_1(0) \cap T_2(d) $$
as
$$ \phi_{d,1}(z_1,z_2) \coloneqq
(k_1(\Phi_1(z_1,z_2)),k_2(\Phi_1(z_1,z_2))), $$
where $k(x)$ is the inverse of the transformation \eqref{x1x2}, we
obtain the desired map. Note that the conclusion (ii) of the theorem is
immediate. We next prove (i) and (iii).

{\tt Step 4 (proof of (i)).} By Proposition \ref{p:jac}(i), for $1 \leq
j \le 2$ we have $| x_j(k)-w_j(k) | \le \frac{\eps}{8}$. Now, recall
from \eqref{ch1} the definition of $w_1(k)$ and $w_2(k)$. Then, since
$$ |x_j(k)| \le |x_j(k)-w_j(k)| + |w_j(k)| < \frac{\eps}{8} + |w_j(k)|,
$$
the set
$$ \left \{ (k_1,k_2) \in \C^2 \; \Big | \;\; |w_1(k)| < \frac{\eps}{8}
\; \text{ and } \; |w_2(k)| < \frac{\eps}{8} \right \} $$
is contained in the set \eqref{range}. This proves the first part of
(i). To prove the second part we use Proposition \ref{p:jac} and
\eqref{propPhi}. First observe that
$$ D \phi_{d,1} = \frac{\partial k}{\partial x} D \Phi_1 = \frac{1}{2}
\begin{pmatrix} 1 & 1 \\ i & -i \end{pmatrix} (I + N) (I + D\Phi_1
- I) = \frac{1}{2} \begin{pmatrix} 1 & 1 \\ i & -i \end{pmatrix} (I + N
  + \mathcal{R}), $$
where
$$ \| N \| \le \frac{1}{3^3} + \frac{C}{\rho} \qquad \text{and} \qquad
\| \mathcal{R} \| \le \frac{C}{|d|^2}. $$
Furthermore, from \eqref{ch1} and \eqref{x1x2} we have
$$ k_1 = i \theta_\nu(d) + \frac{1}{2}(w_1+w_2) = i \theta_\nu(d) +
\frac{1}{2} ( x_1 + x_2 + \beta_2^{(1)} z_1 + \eta_2^{(1)} z_2 + g_1 +
g_2) $$
and similarly
$$ k_2 = -(-1)^\nu \theta_\nu(d) + \frac{(-1)^\nu}{2i} ( x_1 - x_2 -
\beta_2^{(1)} z_1 + \eta_2^{(1)} z_2 - g_1 + g_2), $$
so that
$$ \phi_{d,1}(0) = k(\Phi_1(0)) = k \left( O \left( \frac{1}{|d|^4}
\right) \right) = (i \theta_\nu(d), -(-1)^\nu \theta_\nu(d)) + O \left(
\frac{\eps}{900} \right) + O \left( \frac{1}{\rho} \right). $$

{\tt Step 5 (proof of (iii)).} To prove part (iii) it suffices to note
that $T_1(0) \cap T_2(d) \cap \Fhat(A,V)$ is mapped to $T_1(-d) \cap
T_2(0) \cap \Fhat(A,V)$ by translation by $d$ and define $\phi_{d,2}$ by
$$ \phi_{d,2}(z_1,z_2) \coloneqq \phi_{d,1}(z_2,z_1) + d. $$
This completes the proof of the theorem.
\end{proof}


\begin{proof}[Proof of Proposition \ref{p:p1p2}]
It suffices to estimate
$$ c_{d',d''} \coloneqq \hat{q}(d'-d'') - 2(d'-d'') \cdot
\hat{A}(d'-d'') \qquad \text{and} \qquad p_{d',d''} \coloneqq
D_{d',d''}- c_{d',d''} $$
for $d',d'' \in \{0,d\}$ with $d' \neq d''$. Define $l^{d'd''}_\nu
\coloneqq (1,i(-1)^\nu) \cdot \hat{A}(d'-d'')$. Observe that, since
$$ |\hat{q}(d'-d'')| = \frac{1}{|d'-d''|^2} |d'-d''|^2 \,
|\hat{q}(d'-d'')| \leq \frac{1}{|d'-d''|^2} \sum_{b \in \Gdual} |b|^2 \,
|\hat{q}(b)| \leq \| b^2 \hat{q}(b) \|_{l^1} \frac{1}{|d|^2}, $$
and similarly
$$ |\hat{A}(d'-d'')| \leq \| b^2 \hat{A}(b) \|_{l^1} \frac{1}{|d|^2}, $$
it follows that
$$ |c_{d',d''}| \leq \frac{C_{A,q}}{|d|} \qquad \text{and} \qquad
|l^{d'd''}_\nu| \leq \frac{C_{A}}{|d|^2}. $$
This gives the desired bounds for $c_1$ and $c_2$.

Now, by Proposition \ref{p:newv} we have
$$ p = J^{d'd''}_{\nu'} w_{\nu,d'}^2 + J^{d'd''}_\nu z_{\nu,d'}^2 +
K^{d'd''} w_{\nu,d} z_{\nu,d'} + (\tilde{L}^{d'd''}_{\nu'} -
l^{d'd''}_{\nu'}) w_{\nu,d'} + (\tilde{L}^{d'd''}_\nu - l^{d'd''}_\nu)
z_{\nu,d'} + \tilde{M}^{d'd''} $$
with $\tilde{L}^{d'd''}_\nu \coloneqq L^{d'd''}_\nu + l^{d'd''}_\nu$ and
$\tilde{M}^{d'd''} \coloneqq M^{d'd''} - c$. Observe that all the
coefficients $J^{d'd''}_\nu$, $K^{d'd''}$,  $\tilde{L}^{d'd''}_\nu$ and
$\tilde{M}^{d'd''}$ have exactly the same form as the function
$\Phi_{d',d''}(k)$ of Lemma \ref{l:derii} (see Proposition \ref{p:newv}
and \eqref{Phi}). Thus, by this lemma with $\beta=2$, for any integers
$n$ and $m$ with $n+m \ge 0$, the absolute value of the
$\frac{\partial^{n+m}}{\partial k_1^n \partial k_2^m}$-derivative of
each of these functions is bounded above by $C_{\eps, \Lambda, A, q,m,n}
\, \frac{1}{|d|^3}$. Hence, if we recall from Proposition
\ref{p:order}(ii) that $|z_1(k)| \le 6|d|$ and $|z_2(k)| \le 2|d|$, and
apply the Leibniz rule we find that
$$ \left| \frac{\partial^{n+m}}{\partial k_1^n \partial k_2^m} \,
p_{d',d''}(k) \right| \leq C_{m,n} \, \frac{C}{|d|}. $$
This yields the desired bounds for $p_1$ and $p_2$ and completes the proof.
\end{proof}


\begin{proof}[Proof of Proposition \ref{p:jac}]
(i) Similarly as in \eqref{as1} we have
$$ |\beta_2^{(1)}(k)| \le \frac{\eps}{900} \, \frac{1}{|z_1(k)|} \qquad
\text{and} \qquad |\eta_2^{(1)}(k)| \le \frac{\eps}{900} \,
\frac{1}{|z_2(k)|}. $$
Thus, in view of \eqref{decg1g2}, and by choosing $\rho$ sufficiently
large,
$$ |x_1(k)-w_1(k)| \le |\beta_2^{(1)}(k) \, z_1(k) + g_1(k)| \le
\frac{\eps}{900} + \frac{C}{\rho} < \frac{\eps}{8}, $$
and similarly $|x_2(k)-w_2(k)| < \eps/8$. This proves part (i).

(ii) Recall \eqref{ch1} and \eqref{x1x2}. Then, for $1 \le j \le 2$,
\begin{align*}
\frac{\partial x_1}{\partial k_j} & = \frac{\partial}{\partial k_j} (w_1
+ z_1 \beta_2^{(1)} + g_1) = \frac{\partial w_1}{\partial k_j} + z_1
\frac{\partial \beta_2^{(1)}}{\partial k_j} + \frac{\partial
z_1}{\partial k_j} \beta_2^{(1)} + \frac{\partial g_1}{\partial k_j}, \\
\frac{\partial x_2}{\partial k_j} & = \frac{\partial}{\partial k_j} (w_2
+ z_2 \eta_2^{(1)} + g_2) = \frac{\partial w_2}{\partial k_j} + z_2
\frac{\partial \eta_2^{(1)}}{\partial k_j} + \frac{\partial
z_2}{\partial k_j} \eta_2^{(1)} + \frac{\partial g_2}{\partial k_j}.
\end{align*}
First observe that the functions $g_1$ and $g_2$ are similar to the
function $g$ (see \eqref{defg1g2} and \eqref{defg}). Thus, it is easy to
see that $\frac{\partial g_1}{\partial k_j}$ and $\frac{\partial
g_2}{\partial k_j}$ are given by expressions similar to \eqref{ee1}.
Since $k \in T_\nu(0) \cap T_{\nu'}(d)$ we have $|w_1(k)| < \eps$ and
$|w_2(k)| < \eps$. Recall also the inequalities in Proposition
\ref{p:order}(ii). Hence, by Lemmas \ref{l:b1}(ii), \ref{l:der}(ii) and
\ref{l:der2}(ii), we obtain \eqref{ee2} with $k_1$ and $z(k)$ replaced
by $k_j$ and $z_1(k)$, respectively, and for $k_1$, $z(k)$ and $\beta$
replaced by $k_j$, $z_2(k)$ and $\eta$, respectively. Consequently,
similarly as in \eqref{ee3} and using again Lemma \ref{l:b1}(ii), for $1
\le j \le 2$ we have
$$ \left| \frac{\partial z_1}{\partial k_j} \beta_2^{(1)} +
\frac{\partial g_1}{\partial k_j} \right| \leq C_{\eps,\Lambda,q,A} \,
\frac{1}{\rho} \qquad \text{and} \qquad \left| \frac{\partial
z_2}{\partial k_j} \eta_2^{(1)} + \frac{\partial g_2}{\partial k_j}
\right| \leq C_{\eps,\Lambda,q,A} \, \frac{1}{\rho}. $$ 
Now recall that $\beta_2 = J_\nu^{00}$ and $\eta_2 = J_{\nu'}^{dd}$.
Then, by Proposition \ref{p:newv}, Lemma \ref{l:b1}(ii), and
\eqref{alpha}, it follows that
\begin{align*}
\beta_2^{(1)}(k) & = (J_\nu^{00})^{(1)}(k) = \sum_{b,c \in G'_1}
\frac{(1,i(-1)^\nu) \cdot \hat{A}(-b)}{N_b(k)} \, S_{b,c} \,
(1,-i(-1)^\nu) \cdot \hat{A}(c), \\
\eta_2^{(1)}(k) & = (J_{\nu'}^{dd})^{(1)}(k) = \sum_{b,c \in G'_1}
\frac{(1,i(-1)^{\nu'}) \cdot \hat{A}(d-b)}{N_b(k)} \, S_{b,c} \,
(1,-i(-1)^{\nu'}) \cdot \hat{A}(c-d).
\end{align*}
Hence, by Lemma \ref{l:der2}(ii), similarly as in \eqref{ee4}, for $1
\le j \le 2$,
$$ \left | z_1(k) \frac{ \partial \beta_2^{(1)}(k)}{\partial k_j}
\right| \leq \frac{13}{\Lambda^2} \| (1,-i(-1)^\nu) \cdot \hat{A}
\|_{l^1}^2 < \frac{1}{7 \cdot 3^4} \qquad \text{and} \qquad \left |
z_2(k) \frac{ \partial \eta_2^{(1)}(k)}{\partial k_j} \right| <
\frac{1}{7 \cdot 3^4}. $$
Therefore,
\begin{align*}
\begin{pmatrix}
\frac{\partial x_1}{\partial k_1} & \frac{\partial x_1}{\partial k_2} \\
\frac{\partial x_2}{\partial k_1} & \frac{\partial x_2}{\partial k_2}
\end{pmatrix} & = \begin{pmatrix} 1 & i(-1)^\nu \\ 1 & i(-1)^{\nu'}
\end{pmatrix} + \begin{pmatrix} z_1(k) \frac{ \partial
\beta_2^{(1)}(k)}{\partial k_1} & z_1(k) \frac{\partial
\beta_2^{(1)}(k)}{\partial k_2} \\ z_2(k) \frac{ \partial
\eta_2^{(1)}(k)}{\partial k_1} & z_2(k) \frac{ \partial
\eta_2^{(1)}(k)}{\partial k_2} \end{pmatrix} \\ & \qquad +
\begin{pmatrix} \beta_2^{(1)} & - i(-1)^\nu \beta_2^{(1)} \\
\eta_2^{(1)} & - i(-1)^{\nu'} \eta_2^{(1)} \end{pmatrix} +
\begin{pmatrix} \frac{\partial g_1}{\partial k_1} & \frac{\partial
g_1}{\partial k_2} \\ \frac{\partial g_2}{\partial k_1} & \frac{\partial
g_2}{\partial k_2} \end{pmatrix} \\
& \eqqcolon \begin{pmatrix} 1 & i(-1)^\nu \\ 1 & i(-1)^{\nu'}
\end{pmatrix} ( I + M_1 + M_2 + M_3),
\end{align*}
where
$$ \| M_1 \| \leq 2 \, \frac{2}{7 \cdot 3^4} \qquad \text{and} \qquad \|
M_2 + M_3 \| \leq C_{\eps,\Lambda,q,A} \, \frac{1}{\rho}. $$
Set $M \coloneqq M_1 + M_2 + M_3$. This proves the first claim.

Now, by choosing $\rho$ sufficiently large we can make $\| M \| <
\frac{1}{2}$. Write
$$ P \coloneqq \begin{pmatrix} 1 & i(-1)^\nu \\ 1 & i(-1)^{\nu'}
\end{pmatrix}. $$
Then, by the inverse function theorem and using the Neumann series,
\begin{align*}
\begin{pmatrix} \frac{\partial k_1}{\partial x_1} & \frac{\partial
k_1}{\partial x_2} \\ \frac{\partial k_2}{\partial x_1} & \frac{\partial
k_2}{\partial x_2} \end{pmatrix} & = \begin{pmatrix} \frac{\partial
x_1}{\partial k_1} & \frac{\partial x_1}{\partial k_2} \\ \frac{\partial
x_2}{\partial k_1} & \frac{\partial x_2}{\partial k_2}
\end{pmatrix}^{-1} = (I + M)^{-1} P^{-1} = (I + \tilde{M}) P^{-1} \\ &
\eqqcolon P^{-1} (I + P \tilde{M} P^{-1}) = \frac{1}{2} \begin{pmatrix}
1 & 1 \\ i(-1)^{\nu'} & i(-1)^\nu \end{pmatrix} (I + P \tilde{M}
P^{-1}),
\end{align*}
with
$$ \| P \tilde{M} P^{-1} \| \leq 2 \| \tilde{M} \| 1 \le \frac{2 \| M \|
}{1 - \| M \|} \le 4 \| M \|. $$
Set $N \coloneqq P \tilde{M} P^{-1}$. This proves the second claim.

Differentiating the matrix identity $T T^{-1}=I$ and applying the chain
rule we find that
$$ \frac{\partial^2 k_m}{\partial x_i \partial x_j} = - \sum_{l,p=1}^2
\frac{\partial k_m}{\partial x_l} \, \frac{\partial}{\partial x_i}
\left( \frac{\partial x_l}{\partial k_p} \right) \frac{\partial
k_p}{\partial x_j} = - \sum_{l,p=1}^2 \frac{\partial k_m}{\partial x_l}
\, \frac{\partial^2 x_l}{\partial k_r \partial x_p} \, \frac{\partial
k_r}{\partial x_i} \, \frac{\partial k_p}{\partial x_j}. $$
Furthermore, in view of the above calculations we have
$$ \left| \frac{\partial k_i}{\partial x_j} \right| \leq \frac{1}{2} ( 1
+ \| N \| ) \leq \frac{1}{2} ( 1 + 4 \| M \| ) \le \frac{1}{2} \left( 1
+ 4 \, \frac{1}{2} \right) < \frac{3}{2}. $$
Thus,
$$ \left| \frac{\partial^2 k_m}{\partial x_i \partial x_j} \right| \leq
4 \left( \frac{3}{2} \right)^3 \sup_{l,r,p} \left| \frac{\partial^2
x_l}{\partial k_r \partial x_p} \right|. $$

We now estimate
$$ \frac{\partial^2 x_1}{\partial k_i \partial k_j} = \frac{\partial
z_1}{\partial k_i} \, \frac{\partial \beta_2^{(1)}}{\partial k_j} + z_1
\frac{\partial^2 \beta_2^{(1)}}{\partial k_i \partial k_j} +
\frac{\partial z_1}{\partial k_j} \, \frac{\partial
\beta_2^{(1)}}{\partial k_i} + \frac{\partial^2 g_1}{\partial k_i
\partial k_j} \qquad \text{and} \qquad \frac{\partial^2 x_2}{\partial
k_i \partial k_j}. $$
From \eqref{ee1} with $g$, $w$ and $z$ replaced by $g_1$, $w_1$ and
$z_1$, respectively, we obtain
\begin{align*}
\frac{\partial^2 g_1}{\partial k_1^2} & = \frac{\partial^2
\beta_1}{\partial k_1^2} \frac{w_1^2}{z_1} + 2 \frac{\partial
\beta_1}{\partial k_1} \frac{2w_1 z_1 - w_1^2}{z_1^2} + \beta_1 \,
\frac{2 z_1^2 - 6w_1z_1 + 4 w_1^2}{z_1^3} + \left( \frac{\partial^2
\beta_2^{(2)}}{\partial k_1^2} + \frac{\partial^2
\beta_2^{(3)}}{\partial k_1^2} \right) z_1 \\
& \quad + 2 \left( \frac{\partial \beta_2^{(2)}}{\partial k_1} +
\frac{\partial \beta_2^{(3)}}{\partial k_1} \right) + \frac{\partial^2
\beta_3}{\partial k_1^2} w_1 + 2 \frac{\partial \beta_3}{\partial k_1} +
\frac{\partial^2 \beta_4}{\partial k_1^2} \, \frac{w_1}{z_1} + 2
\frac{\partial \beta_4}{\partial k_1} \, \frac{z_1-w_1}{z_1^2} \\ &
\quad + \beta_4 \, \frac{2(w_1-z_1)}{z_1^3} + \frac{\partial^2
\beta_5}{\partial k_1^2} + \frac{\partial^2 \beta_6}{\partial k_1^2} \,
\frac{1}{z_1} - 2 \frac{\partial \beta_6}{\partial k_1} \,
\frac{1}{z_1^2} + 2 \frac{\beta_6}{z_1^3} + \frac{2 \hat{q}(0)}{z_1^3}.
\end{align*}
Hence, by Lemmas \ref{l:b1}(ii), \ref{l:der}(ii) and \ref{l:der2}(ii),
$$ \left| \frac{\partial^2 g_1}{\partial k_1^2} \right| \le C_{\eps,
\Lambda, q, A} \, \frac{1}{\rho}. $$
Similarly we prove that
$$ \left| \frac{\partial^2 g_l}{\partial k_i \partial k_j} \right| \le
C_{\eps, \Lambda, q, A} \, \frac{1}{\rho} $$
for all $l,i,j \in \{1,2\}$ because all the derivatives acting on $g_l$
are essentially the same up to constant factors (see \cite{deO}).
Furthermore, again by Lemma \ref{l:der2}(ii),
$$ \left| \frac{\partial \beta_2^{(1)}}{\partial k_j} \right| \le
C_{\eps, \Lambda, q, A} \, \frac{1}{\rho} , \qquad \left| \frac{\partial
\eta_2^{(1)}}{\partial k_j} \right| \le C_{\eps, \Lambda, q, A} \,
\frac{1}{\rho}, $$
and
$$ \left | z_1(k) \frac{ \partial^2 \beta_2^{(1)}(k)}{\partial k_1
\partial k_j} \right| \leq \frac{65}{\Lambda^3} \| (1,-i(-1)^\nu) \cdot
\hat{A} \|_{l^1}^2 < \frac{1}{5 \Lambda^3} \, \eps^2, \qquad \left |
z_2(k) \frac{ \partial^2 \eta_2^{(1)}(k)}{\partial k_i \partial k_j}
\right| < \frac{1}{5 \Lambda^3} \, \eps^2. $$
Hence,
$$ \left| \frac{\partial^2 x_l}{\partial k_i \partial k_j} \right| \leq
\frac{1}{5 \Lambda^3} \, \eps^2 + C_{\eps, \Lambda, q, A} \,
\frac{1}{\rho}. $$
Therefore,
$$ \left| \frac{\partial^2 k_m}{\partial x_i \partial x_j} \right| \leq
4 \left( \frac{3}{2} \right)^3 \sup_{l,r,p} \left| \frac{\partial^2
x_l}{\partial k_r \partial x_p} \right| \le \frac{3}{\Lambda^3} \,
\eps^2 +  C_{\eps, \Lambda, q, A} \, \frac{1}{\rho}. $$
This completes the proof of the proposition.
\end{proof}


\appendix

\section{Quantitative Morse lemma} \label{s:morse}

\begin{lem}[Quantitative Morse lemma \cite{deO}] \label{l:morse}
Let $\delta$ be a constant with $0 < \delta < 1$ and assume that 
$$ f(x_1,x_2) = x_1 x_2 + r(x_1,x_2) $$
is an holomorphic function on $D_\delta = \{ (x_1,x_2) \in \C^2 \;\; |
\;\; |x_1| \le \delta \text{ and } \, |x_2| \le \delta \}$. Suppose
further that, for all $x \in D_\delta$ and $1 \le i \le 2$, the function
$r$ satisfies
\[ \left| \frac{\partial r}{\partial x_i}(x) \right| \leq a < \delta
\qquad \text{and} \qquad \left\| \left[ \frac{\partial^2 r}{\partial x_i
\partial x_j}(x) \right]_{i,j \in \{1,2\}} \right \| \le b <
\frac{1}{55}, \]
where $a$ and $b$ are constants. Then $f$ has a unique critical point
$\xi = (\xi_1,\xi_2) \in D_\delta$ with $|\xi_1| \le a$ and $|\xi_2|
\leq a$. Furthermore, let $s = \max \{ |\xi_1|, |\xi_2| \}$. Then there
is a biholomorphic map $\Phi$ from the domain $D_{(\delta-s)(1-19b)}$ to
a neighbourhood of $\xi \in D_\delta$ that contains
$$ \{ (z_1,z_2) \in \C^2 \; | \;\; |z_i-\xi_i| < (\delta-s)(1-55b)
\text{ for } 1 \le i \le 2 \} $$
such that $(f \circ \Phi)(z_1,z_2) = z_1 z_2 + c$, where $c \in \C$ is a
constant fulfilling $|c-r(0,0)| \le a^2$. The differential $D \Phi$
obeys $\| D \Phi - I \| \leq 18 b$. If $\frac{\partial r}{\partial
x_1}(0,0) = 0$ and $\frac{\partial r}{\partial x_2}(0,0) = 0$, then
$\xi=0$ and $s=0$. 
\end{lem}


\section{Asymptotics for the coeficients: proofs} \label{s:app2}

\begin{proof}[Proof of Proposition \ref{p:order}]
We first derive a more general inequality and then we prove parts (i)
and (ii). First observe that, if $k \in T_\mu(d') \setminus
\mathcal{K}_R$ then
$$ |v + (-1)^\mu(u+d')^\perp| = |N_{d',\mu}(k)| < \eps < |v|. $$
Hence,
$$ |v| \leq |2v - (v + (-1)^\mu(u+d')^\perp)| \leq 3|v|. $$
But
\begin{align*}
|2v - (v + (-1)^\mu(u+d')^\perp)| & = |v-(-1)^\mu(u+d')^\perp| \\
& = |k_1+d'_1 - i(-1)^\mu (k_2 + d'_2)| = |z_{\mu,d'}(k)|.
\end{align*}
Therefore,
\begin{equation} \label{pre}
\frac{1}{|z_{\mu,d'}(k)|} \leq \frac{1}{|v|} \leq
\frac{3}{|z_{\mu,d'}(k)|}.  \end{equation}
We now prove parts (i) and (ii).

(i) The first inequality of part (i) follows from the above estimate
setting $(\mu,d')=(\nu,0)$. To prove the second inequality observe that,
since $|v|>R \geq 2\Lambda > 12 \eps$ by hypothesis and $|v| \leq
|z_{\nu,0}(k)|$ by \eqref{pre}, on the one hand we have
\begin{align*}
\tfrac{1}{4} |v| \leq \tfrac{11}{12} |v| = |v| - \tfrac{1}{12}|v| \leq
|v| - \tfrac{1}{6} \Lambda \leq |v| - \eps & \leq |z_{\nu,0}(k)| -
|k_1+i(-1)^\nu k_2| \\ & \leq |z_{\nu,0}(k)-k_1-i(-1)^\nu k_2| = 2|k_2|.
\end{align*}
On the other hand, since $|z_{\nu,0}(k)| < 3|v|$ by \eqref{pre},
\begin{align*}
|k_2| = |2i(-1)^\nu k_2| & = |k_1+i(-1)^\nu k_2 -( k_1 -i(-1)^\nu k_2)|
\\ & = |k_1+i(-1)^\nu k_2 - z_{\nu,0}(k)| \leq \eps + 3|v| \leq 4|v|.
\end{align*}
Combining these estimates we obtain the second inequality of part (i).

(ii) Similarly, in view of \eqref{pre}, if $k \in T_\mu(d') \setminus
\mathcal{K}_R$ for $(\mu,d') \in \{ (\nu,0), (\nu',d) \}$ then
\begin{equation} \label{zv}
\frac{1}{|z_{\nu,0}(k)|} \leq \frac{1}{|v|} \leq
\frac{3}{|z_{\nu,0}(k)|} \qquad \text{and} \qquad
\frac{1}{|z_{\nu',d}(k)|} \leq \frac{1}{|v|} \leq
\frac{3}{|z_{\nu',d}(k)|}.
\end{equation}
These are the first two inequalities of part (ii). Now, since
\begin{align*}
z_{\nu',d}(k) & = k_1 - i(-1)^{\nu'} k_2 + d_1 - i(-1)^{\nu'} d_2 \\ & =
z_{\nu',0}(k) + d_1 - i(-1)^{\nu'} d_2 = w_{\nu,0}(k) + d_1 -
i(-1)^{\nu'} d_2,
\end{align*}
$|w_{\nu,0}(k)| < \eps$, and $|d_1 - i(-1)^{\nu'} d_2| = |d|$, it
follows that
$$ |z_{\nu',d}(k)| - \eps \leq |d| \leq |z_{\nu',d}(k)| + \eps. $$
Furthermore, by \eqref{zv},
$$ \eps < \frac{\Lambda}{6} \leq \frac{|v|}{12} \leq
\frac{|z_{\nu',d}(k)|}{12}. $$
Thus,
$$ \frac{1}{2} |z_{\nu',d}(k)| \leq |d| \leq 2 |z_{\nu',d}(k)|. $$
This yields the third inequality of part (ii) and completes the proof.
\end{proof}


\begin{proof}[Proof of Lemma \ref{l:b1}]
We consider all cases at the same time. Therefore, we have either
hypothesis (i) with $(\mu,d')=(\nu,0)$ or hypothesis (ii) with $(\mu,d')
\in \{ (\nu,0), (\nu',d) \}$. Observe that either $(\nu,\nu')=(1,2)$ or
$(\nu,\nu')=(2,1)$. {\tt Step 1.} Recall the change of variables
\eqref{change} and set
$$ G'_1 \coloneqq \big \{ b \in G' \; | \;\; |b-d'| < \tfrac{1}{4} R
\big \}, \qquad G'_2 \coloneqq \big \{ b \in G' \; | \;\; |b-d'| \geq
\tfrac{1}{4} R \big \}. $$
Then $G' = G'_1 \cup G'_2$ and $G'_1, G'_2 \subset \{ b \in \Gdual \; |
\;\; |N_b(k)| \geq \eps |v| \}$ by Proposition \ref{p:Gprime}.
Furthermore, by Proposition \ref{p:order}, for $(\mu,d')=(\nu,0)$ if (i)
or $(\mu,d') \in \{ (\nu,0), (\nu',d) \}$ if (ii) we have $|z_{\mu,d'}|
\le 3|v|$. Thus, observing the definition of $G'_2$,
\begin{align}
|\mathcal{R}_1(k)| & \coloneqq \Bigg| \sum_{b \in G'_1} \sum_{c \in
G'_2} \frac{f(d'-b)}{N_b(k)} (R_{G'G'}^{-1})_{b,c} \, g(c-d') \Bigg|
\label{RR1} \\
& \leq \frac{1}{\eps |v|} \| R^{-1}_{G'G'} \| \sum_{b \in G'_1}
|f(d'-b)| \sum_{c \in G'_2} \frac{|c-d'|^2}{|c-d'|^2} |g(c-d')| \notag\\
& \leq \frac{1}{\eps |v|} \| R^{-1}_{G'G'} \| \, \| f \|_{l^1}
\frac{16}{R^2} \| c^2 g(c) \|_{l^1} \notag \leq
\frac{C_{\eps,f,g}}{|z_{\mu,d'}| R^2}, \notag
\end{align}
and similarly
\begin{equation} \label{RR2}
|\mathcal{R}_2(k)| \leq \frac{C_{\eps,f,g}}{|z_{\mu,d'}| R^2}.
\end{equation}
Hence,
\begin{equation} \label{exp1}
\begin{split}
\Phi_{d',d'}(k) & = \Bigg[ \; \sum_{b,c \in G'_1} + \sum_{\substack{b
\in G'_1 \\ c \in G'_2}} + \sum_{\substack{b \in G'_2 \\ c \in G'}} \;
\Bigg] \; \frac{f(d'-b)}{N_b(k)} (R_{G'G'}^{-1})_{b,c} \, g(c-d')\\ 
& = \sum_{b,c \in G'_1} \frac{f(d'-b)}{N_b(k)} (R_{G'G'}^{-1})_{b,c} \,
g(c-d') + \mathcal{R}_1(k) + \mathcal{R}_2(k)
\end{split}
\end{equation}
with
\begin{equation} \label{decR1R2}
|\mathcal{R}_1(k) + \mathcal{R}_2(k)| \le
\frac{C_{\eps,f,g}}{|z_{\mu,d'}|R^2}.
\end{equation}

Now, if we set $T_{G'G'} \coloneqq \pi_{G'} - R_{G'G'}$ and recall the
convergent series expansion
$$ R_{G'G'}^{-1} = (\pi_{G'} - T_{G'G'})^{-1} = \sum_{j=0}^\infty
T_{G'G'}^j, $$
we can write
\begin{equation} \label{exp2}
\sum_{b,c \in G'_1} \frac{f(d'-b)}{N_b(k)} (R_{G'G'}^{-1})_{b,c} \,
g(c-d') = \sum_{j=0}^\infty \sum_{b,c \in G'_1} \frac{f(d'-b)}{N_b(k)}
(T^j_{G'G'})_{b,c} \, g(c-d').
\end{equation}
Note, the above equality is fine because $G'_1$ is finite set. Let
$$ G'_3 \coloneqq \{ b \in G' \; | \;\; |b-d'| < \tfrac{1}{2} R \},
\qquad G'_4 \coloneqq \{ b \in G' \; | \;\; |b-d'| \geq \tfrac{1}{2} R
\}. $$
Again, observe that $G' = G'_3 \cup G'_4$. Thus, we can break $T_{G'G'}$
into
$$ T_{G'G'} = \pi_{G'} T \pi_{G'} = (\pi_{G'_3} + \pi_{G'_4}) T
(\pi_{G'_3} + \pi_{G'_4}) = T_{33} + T_{43} + T_{34} + T_{44}, $$
where $T_{ij} \coloneqq \pi_{G_i} T \pi_{G_j}$ for $i,j \in \{3,4\}$.
Using this decomposition we prove the following.

\begin{prop} \label{p:2}
Under the hypotheses of Lemma \ref{l:b1} we have
$$ \sum_{j=0}^\infty \sum_{b,c \in G'_1} \frac{f(d'-b)}{N_b(k)}
(T^j_{G'G'})_{b,c} \, g(c-d') = \sum_{j=0}^\infty \sum_{b,c \in G'_1}
\frac{f(d'-b)}{N_b(k)} (T^j_{33})_{b,c} \, g(c-d') + \mathcal{R}_3(k) $$
with $\mathcal{R}_3(k)$ given by \eqref{RR3} and
\begin{equation} \label{decR3}
|\mathcal{R}_3(k)| \leq \frac{C_{\Lambda,f,g}}{|z_{\mu,d'}| R^2}.
\end{equation}
\end{prop}

This proposition will be proved below. Combining this with \eqref{exp1}
and \eqref{exp2} we obtain
\begin{equation} \label{exp3}
\Phi_{d',d'}(k) = \sum_{j=0}^\infty \sum_{b,c \in G'_1}
\frac{f(d'-b)}{N_b(k)} (T^j_{33})_{b,c} \, g(c-d') + \sum_{j=1}^3
\mathcal{R}_j(k).
\end{equation}

{\tt Step 2.} We now look in detail to the operator $T_{33}$ and its
powers $T_{33}^j$. Recall that $\theta_\mu(b) = \frac{1}{2} ((-1)^\mu
b_2 + ib_1)$ and set $\mu' \coloneqq \mu-(-1)^\mu$ so that $(-1)^\mu =
-(-1)^{\mu'}$. Then,
\begin{align*}
N_b(k) & = N_{b,\mu}(k) N_{b,\mu'}(k) \\
& = (w_{\mu,d'} - 2i\theta_{\mu'}(b-d'))(z_{\mu,d'} 
- 2i \theta_\mu(b-d')).
\end{align*}
Extend the definition of $\theta_\mu(y)$ to any $y \in \C^2$. Thus,
$$ 2(k+d') \cdot \hat{A}(b-c) = -2i \theta_\mu (\hat{A}(b-c)) \,
w_{\mu,d'} - 2i \theta_{\mu'} (\hat{A}(b-c)) \, z_{\mu,d'}. $$
Hence,
\begin{equation} \label{TXY0}
\begin{split}
T_{b,c} & = \frac{1}{N_c(k)} (2(c + k) \cdot \hat{A}(b-c)-\hat{q}(b-c) )
\\ 
& = \frac{2 (c-d') \cdot \hat{A}(b-c) - \hat{q}(b-c) + 2 (k+d') \cdot
\hat{A}(b-c)}{(w_{\mu,d'} - 2i \theta_{\mu'}(c-d')) (z_{\mu,d'}-2i
\theta_\mu(c-d'))} = X_{b,c} + Y_{b,c},
\end{split}
\end{equation}
where
\begin{align}
X_{b,c} & \coloneqq \frac{2 (c-d') \cdot \hat{A}(b-c) - \hat{q}(b-c) -2i
\theta_\mu (\hat{A}(b-c)) \, w_{\mu,d'}}{(w_{\mu,d'} - 2i
\theta_{\mu'}(c-d')) (z_{\mu,d'}- 2i \theta_\mu(c-d'))}, \label{Xbc} \\
Y_{b,c} & \coloneqq \frac{-2i \theta_{\mu'} (\hat{A}(b-c)) \,
z_{\mu,d'}}{(w_{\mu,d'} - 2i \theta_{\mu'}(c-d')) (z_{\mu,d'} - 2i
\theta_\mu(c-d'))}. \label{Ybc}
\end{align}
Let $X$ and $Y$ be the operators whose matrix elements are,
respectively, $X_{b,c}$ and $Y_{b,c}$. Set
$$ X_{33} \coloneqq \pi_{G'_3} X \pi_{G'_3} \qquad \text{and} \qquad Y_{33}
\coloneqq \pi_{G'_3} Y \pi_{G'_3}. $$
We next prove the following estimates,
\begin{equation} \label{est2}
\begin{aligned}
\| X_{33} \| & \leq \left( 20 \| \hat{A} \|_{l^1} + \frac{4}{\Lambda}
\| \hat{q} \|_{l^1} \right) \frac{1}{|z_{\mu,d'}|_R} < \frac{1}{3},\\
\| Y_{33} \| & \leq \frac{8}{\Lambda} \| \theta_{\mu'}(\hat{A})
\|_{l^1} < \frac{1}{14},
\end{aligned}
\end{equation}
where
$$ |z_{\mu,d'}|_R \coloneqq 2|z_{\mu,d'}|-R. $$

First observe that the ``vector'' $b \in \Gdual$ has the same length
as the complex number $2i \theta_\mu(b)$:
\begin{equation} \label{est5/2}
|b| = |(b_1,b_2)| = |b_1 + i(-1)^\mu b_2| = |2i \theta_\mu(b)|.
\end{equation}
Thus, for $b \in G'_3$,
$$ \frac{|2i \theta_\mu(b-d')|}{R} = \frac{|b-d'|}{R} < \frac{1}{2}. $$
Consequently,
\begin{equation} \label{est3}
\frac{1}{|z_{\mu,d'}-2i \theta_\mu(b-d')|} \leq
\frac{1}{|z_{\mu,d'}|-|2i\theta_\mu(b-d')|} < \frac{1}{|z_{\mu,d'}| -
\tfrac{1}{2} R} = \frac{2}{|z_{\mu,d'}|_R}.
\end{equation}
Furthermore, for $b \in G'$,
\begin{align}
\frac{1}{|w_{\mu,d'}-2i \theta_{\mu'}(b-d')|} \leq
\frac{1}{|b-d'|-|w_{\mu,d'}|} & \leq\frac{1}{|b-d'|-\eps}\label{est4} \\
& \leq \frac{1}{2\Lambda - \Lambda} = \frac{1}{\Lambda}. \label{est5}
\end{align}
Here we have used that $|w_{\mu,d'}|<\eps<\Lambda$ and $|b-d'| \geq 2
\Lambda$ for all $b \in G'$. Using again that $\eps < \Lambda \leq
|c-d'|/2$ for all $c \in G'$ we have
\begin{equation} \label{est6}
\frac{|c-d'|}{|c-d'|-\eps} < 2.
\end{equation}
Finally recall that
\begin{equation} \label{est7}
\frac{\eps}{\Lambda} < \frac{1}{6} \qquad \text{and} \qquad
\frac{1}{|z_{\mu,d'}|} \leq \frac{1}{|v|} < \frac{1}{R},
\end{equation}
where the last inequality follows from Proposition \ref{p:order} since
$|v|>R$ by hypothesis. Then, using the above inequalities and
Proposition \ref{p:ineq}, the bounds \eqref{est2} for $\| X_{33}\|$ and
$\| Y_{33}\|$ follow from the estimates
\begin{align*}
& \left[ \sup_{c \in G'_3} \sum_{b \in G'_3} + \sup_{b \in G'_3} \sum_{c
\in G'_3} \right] |X_{b,c}| \\
& \leq \left[ \sup_{c \in G'_3} \sum_{b \in G'_3} + \sup_{b \in G'_3}
\sum_{c \in G'_3} \right] \frac{2|c-d'| \, |\hat{A}(b-c)| +
|\hat{q}(b-c)| + |2i \theta_\mu (\hat{A}(b-c))| \,
|w_{\mu,d'}|}{|w_{\mu,d'} - 2i \theta_{\mu'}(c-d')| \, |z_{\mu,d'} - 2i
\theta_\mu(c-d')|} \\
& \leq \frac{2}{|z_{\mu,d'}|_R} \left[ \sup_{c \in G'_3} \sum_{b \in
G'_3} + \sup_{b \in G'_3} \sum_{c \in G'_3} \right] \left[ \frac{2|c-d'|
\, |\hat{A}(b-c)|}{|w_{\mu,d'} - 2i \theta_{\mu'}(c-d')|} +
\frac{|\hat{q}(b-c)| + \eps \sqrt{2} \, |\hat{A}(b-c)|}{|w_{\mu,d'} - 2i
\theta_{\mu'}(c-d')|} \right]\\
& \leq \frac{2}{|z_{\mu,d'}|_R} \left[ \sup_{c \in G'_3} \sum_{b \in
G'_3} + \sup_{b \in G'_3} \sum_{c \in G'_3} \right] \left[ \frac{2|c-d'|
\, |\hat{A}(b-c)|}{|c-d'|-\eps} + \frac{|\hat{q}(b-c)| + \eps \sqrt{2}
\, |\hat{A}(b-c)|}{\Lambda} \right] \\
& \leq \frac{2}{|z_{\mu,d'}|_R} \, 2 \left[ \left[ 4 + \frac{\eps
\sqrt{2}}{\Lambda} \right] \| \hat{A} \|_{l^1} + \frac{ \| \hat{q}
\|_{l^1}}{\Lambda} \right] \leq \left[ 20 \| \hat{A} \|_{l^1} + \frac{ 4
\| \hat{q} \|_{l^1}}{\Lambda} \right] \frac{1}{|z_{\mu,d'}|_R} \\ & \leq
\left[ 20 \| \hat{A} \|_{l^1} + \frac{ 4 \| \hat{q} \|_{l^1}}{\Lambda}
\right] \frac{1}{R} < \frac{1}{7} + \frac{1}{4} = \frac{1}{3}
\end{align*}
and similarly
$$ \left[ \sup_{c \in G'_3} \sum_{b \in G'_3} + \sup_{b \in G'_3}
\sum_{c \in G'_3} \right] |Y_{b,c}| \leq \frac{8}{\Lambda} \|
\theta_{\mu'} (\hat{A}) \|_{l^1} < \frac{1}{14}. $$

{\tt Step 3.} We now look in detail to $T_{33}^j$. For each integer $j
\geq 1$ write
\begin{equation} \label{TXY}
T_{33}^j = ( X_{33} + Y_{33} )^j = Z_j + W_j + Y_{33}^j,
\end{equation}
where $W_j$ is the sum of the $j$ terms containing only one factor
$X_{33}$ and $j-1$ factors $Y_{33}$,
$$ W_j \coloneqq \sum_{m=1}^j (Y_{33})^{m-1} X_{33} (Y_{33})^{j-m}, $$
%
%
%
%
$$ Z_j \coloneqq (X_{33}+Y_{33})^j - W_j - Y_{33}^j. $$
In view of \eqref{est2} we have
\begin{align*}
\| Y_{33} \|^j & \leq \left( \frac{1}{14} \right)^j, \\ \| W_j \| & \leq
j \| X_{33} \| \, \| Y_{33} \|^{j-1} \leq
\frac{C_{\Lambda,A,q}}{|z_{\mu,d'}|_R} \; j \left( \frac{1}{14}
\right)^{j-1}, \\
\| Z_j \| & \leq (2^j-j-1) \, \| X_{33} \|^2 \left( \frac{1}{3}
\right)^{j-2} \leq \frac{C_{\Lambda,A,q}}{|z_{\mu,d'}|_R^2} \left(
\frac{2}{3} \right)^j.
\end{align*}
Hence, the series
\begin{equation} \label{WS}
S \coloneqq \sum_{j=0}^\infty Y_{33}^j = (I - Y_{33})^{-1}, \qquad W
\coloneqq \sum_{j=1}^\infty W_j \qquad \text{and} \qquad Z \coloneqq
\sum_{j=2}^\infty Z_j
\end{equation}
converge, and the operator norm of $W$ and $Z$ decay with respect to
$|z_{\mu,d'}|$. Indeed,
\begin{align*}
\| S \| & \leq \sum_{j=0}^\infty \| Y_{33}\|^j \leq \sum_{j=0}^\infty
\left( \frac{1}{14} \right)^j < C, \\
\| W \| & \leq \sum_{j=1}^\infty \| W_j \| \leq \frac{C'_{\Lambda, A,
q}}{2|z_{\mu,d'}|-R} \sum_{j=1}^\infty \, j \left( \frac{1}{14}
\right)^{j-1} < \frac{C_{\Lambda, A, q}}{|z_{\mu,d'}|_R}, \\
\| Z \| & \leq \sum_{j=2}^\infty \| Z_j \| \leq
\frac{C_{\Lambda,A,q}'}{|z_{\mu,d'}|_R^2} \sum_{j=2}^\infty \left(
\frac{2}{3} \right)^j \leq \frac{C_{\Lambda,A,q}}{|z_{\mu,d'}|_R^2}.
\end{align*}
Thus, we have the expansion
$$ \sum_{j=0}^\infty T_{33}^j = S + W + Z. $$

{\tt Step 4.} Consequently,
\begin{equation} \label{exp4}
\begin{split}
\sum_{j=0}^\infty \sum_{b,c \in G'_1} \frac{f(d'-b)}{N_b(k)}
(T^j_{33})_{b,c} \, g(c-d') & = \sum_{b,c \in G'_1} \frac{f(d'-b) \, (S
+ W + Z)_{b,c} \, g(c-d') }{(w_{\mu,d'} -2i
\theta_{\mu'}(b-d'))(z_{\mu,d'} -2i \theta_\mu(b-d'))} \\
& = \alpha^{(1)}_{\mu,d'} + \alpha^{(2)}_{\mu,d'} + \mathcal{R}_4,
\end{split}
\end{equation}
where
\begin{equation} \label{alpha}
\begin{aligned}
\alpha^{(1)}_{\mu,d'}(k) & \coloneqq \sum_{b,c \in G'_1} \frac{f(d'-b)
\, S_{b,c}(k) \, g(c-d') }{(w_{\mu,d'}(k) -2i
\theta_{\mu'}(b-d'))(z_{\mu,d'}(k) -2i \theta_\mu(b-d'))}, \\
\alpha^{(2)}_{\mu,d'}(k) & \coloneqq \sum_{b,c \in G'_1} \frac{f(d'-b)
\, W_{b,c}(k) \, g(c-d') }{(w_{\mu,d'}(k) -2i
\theta_{\mu'}(b-d'))(z_{\mu,d'}(k) -2i \theta_\mu(b-d'))}
\end{aligned}
\end{equation}
and
\begin{equation} \label{RR4}
\mathcal{R}_4(k) \coloneqq \sum_{b,c \in G'_1} \frac{f(d'-b) \,
Z_{b,c}(k) \, g(c-d') }{(w_{\mu,d'}(k) -2i
\theta_{\mu'}(b-d'))(z_{\mu,d'}(k) -2i \theta_\mu(b-d'))}.
\end{equation}
By a short calculation as in \eqref{done}, using \eqref{est3} and
\eqref{est5} we find that
\begin{equation} \label{dec}
\begin{split}
|\alpha^{(1)}_{\mu,d'}(k)| & \leq \frac{1}{\Lambda} \,
\frac{2}{2|z_{\mu,d'}|-R} \| f \|_{l^1} \| g \|_{l^1} \| S \| \leq
\frac{C_{\Lambda,f,g}}{|z_{\mu,d'}|_R}, \\
|\alpha^{(2)}_{\mu,d'}(k)| & \leq \frac{1}{\Lambda} \,
\frac{2}{2|z_{\mu,d'}|-R} \| f \|_{l^1} \| g \|_{l^1} \| W \| \leq
\frac{C_{\Lambda,A,q,f,g}}{|z_{\mu,d'}|_R^2}, \\
|\mathcal{R}_4(k)| & \leq \frac{1}{\Lambda} \,
\frac{2}{2|z_{\mu,d'}|-R} \| f \|_{l^1} \| g \|_{l^1} \| Z \| \leq
\frac{C_{\Lambda,A,q,f,g}}{|z_{\mu,d'}|_R^3}.
\end{split}
\end{equation}
Hence, recalling \eqref{exp3} we conclude that
$$ \Phi_{d',d'} = \alpha^{(1)}_{\mu,d'} + \alpha^{(2)}_{\mu,d'} +
\alpha^{(3)}_{\mu,d'}, $$
where
\begin{equation} \label{alpha3}
\alpha^{(3)}_{\mu,d'}(k) \coloneqq \sum_{j=1}^4 \mathcal{R}_j(k).
\end{equation}
Furthermore, in view of \eqref{decR1R2}, \eqref{decR3} and \eqref{dec},
since
$$\frac{1}{|z_{\mu,d'}|_R^3} = \frac{1}{(2|z_{\mu,d'}|-R)^3} <
\frac{1}{|z_{\mu,d'}| R^2}, $$
for $1 \leq j \leq 2$ we have
$$ |\alpha^{(j)}_{\mu,d'}(k)| \leq \frac{C_j}{|z_{\mu,d'}(k)|_R^j}
\qquad \text{and} \qquad |\alpha^{(3)}_{\mu,d'}(k)| \leq
\frac{C_3}{|z_{\mu,d'}(k)| R^2}, $$
where $C_j = C_{j;\Lambda,A,q,f,g}$ and $C_3 =
C_{3;\eps,\Lambda,A,q,f,g}$ are constants. This proves the main
statement of the lemma. Finally observe that, since $G'_3$ is a finite
set, the matrices $X_{33}$ and $Y_{33}$ are analytic in $k$ because
their matrix elements are analytic functions of $k$. (Note, the
functions $w_{\mu,d'}(k)$ and $z_{\mu,d'}(k)$ are analytic.)
Consequently, the matrices $W_j$ and $Z_j$ are also analytic and so are
$S_{b,c}$, $W_{b,c}$ and $Z_{b,c}$ because the series \eqref{WS}
converge uniformly with respect to $k$. Thus, all the functions
$\alpha^{(j)}_{\mu,d'}(k)$ are analytic in the region under
consideration. This completes the proof of the lemma.
\end{proof}


\begin{proof}[Proof of Proposition \ref{p:2}]
{\tt Step 1.} Recall that $T_{G'G'} = T_{33} + T_{34} + T_{43} + T_{44}$
with $T_{ij} = \pi_{G'_i} T \pi_{G'_j}$ and set $X^{(0)}_{33} \coloneqq
0$, $Y^{(0)}_{34} \coloneqq T_{34}$, $W^{(0)}_{43} \coloneqq T_{43},$
and $Z^{(0)}_{44} \coloneqq T_{44}$. It is straightforward to verify
that, for any integer $j \ge 0$,
\begin{equation} \label{Tj}
T_{G'G'}^{j+1} = T_{33}^{j+1} + X^{(j)}_{33} + Y^{(j)}_{34} +
W^{(j)}_{43} + Z^{(j)}_{44},
\end{equation}
where
\begin{equation} \label{XYWZ}
\begin{aligned}
X^{(j)}_{33} & \coloneqq T_{33} X^{(j-1)}_{33} + T_{34} W^{(j-1)}_{43} &
: \, L^2_{G'_3} \to L^2_{G'_3}, \\
Y^{(j)}_{34} & \coloneqq T_{33} Y^{(j-1)}_{34} + T_{34} Z^{(j-1)}_{44} &
: \, L^2_{G'_3} \to L^2_{G'_4}, \\
W^{(j)}_{43} & \coloneqq T_{43} T_{33}^j + T_{43} X^{(j-1)}_{33} +
T_{44} W^{(j-1)}_{43} & : \, L^2_{G'_4} \to L^2_{G'_3}, \\
Z^{(j)}_{44} & \coloneqq T_{43} Y^{(j-1)}_{34} + T_{44} Z^{(j-1)}_{44} &
: \, L^2_{G'_4} \to L^2_{G'_4}.
\end{aligned}
\end{equation}

{\tt Step 2.} Since $\pi_{G'_1} \pi_{G'_4} = \pi_{G'_4} \pi_{G'_1} = 0$
and $\pi_{G'_1} \pi_{G'_3} = \pi_{G'_3} \pi_{G'_1} = \pi_{G'_1}$,
substituting \eqref{Tj} into the sum below for the terms where $j \ge 1$
we have, recalling that $X^{(0)}_{33}=0$,
\begin{multline} \label{split}
\sum_{j=0}^\infty \sum_{b,c \in G'_1} \frac{f(d'-b)}{N_b(k)}
(T^j_{G'G'})_{b,c} \, g(c-d') \\
= \sum_{j=0}^\infty \sum_{b,c \in G'_1} \frac{f(d'-b)}{N_b(k)}
(T^j_{33})_{b,c} \, g(c-d') + \sum_{j=1}^\infty \sum_{b,c \in G'_1}
\frac{f(d'-b)}{N_b(k)} (X^{(j)}_{33})_{b,c} \, g(c-d').
\end{multline}
Now recall from \eqref{est3} and \eqref{est5} that, for all $b \in
G'_3$,
\begin{equation} \label{b4}
\frac{1}{|N_b(k)|} \leq \frac{2}{\Lambda} \, \frac{1}{|z_{\mu,d'}|_R},
\end{equation}
and observe that $G'_1 \subset G'_3$. Let $\mathcal{M}$ be either
$T_{G'G'}$ or $T_{33}$. Then, the estimate
\begin{equation} \label{done}
\begin{aligned}
\Bigg| \sum_{b,c \in G'_1} \frac{f(d'-b)}{N_b(k)} (\mathcal{M}^j)_{b,c}
\, g(c-d') \Bigg| & = \Bigg| \sum_{b \in G'_1} \frac{f(d'-b)}{N_b(k)}
\sum_{c \in G'_1} \left \la \frac{e^{i b \cdot x}}{|\Gamma|^{1/2}}, \,
\mathcal{M}^j \, \frac{e^{i c \cdot x}}{|\Gamma|^{1/2}} \right \ra
g(c-d') \Bigg| \\
& \leq \frac{2}{\Lambda} \, \frac{1}{|z_{\mu,d'}|_R} \| f \|_{l_1} \| g
\|_{l_1} \| \mathcal{M} \|^j
\end{aligned}
\end{equation}
implies that the left hand side and the first term on the right hand
side of \eqref{split} converge because $\| \mathcal{M} \| < 17/18$.
Thus, the last term in \eqref{split} also converges. Hence, we are left
to show that
\begin{equation} \label{RR3}
\mathcal{R}_3(k) \coloneqq \sum_{j=1}^\infty \sum_{b,c \in G'_1}
\frac{f(d'-b)}{N_b(k)} (X^{(j)}_{33})_{b,c} \, g(c-d')
\end{equation}
obeys
$$ |\mathcal{R}_3(k)| \leq \frac{C_{\Lambda,f,g}}{|z_{\mu,d'}| \, R^2}.
$$
In order to do this we need the following inequality, which we prove
later.

\begin{prop} \label{p:decayT}
Consider a constant $\beta \geq 0$ and suppose that $\| (1+|b|^\beta)
\hat{q}(b) \|_{l^1} < \infty$ and $\| (1+|b|^\beta) \hat{A}(b) \|_{l^1}
< 2 \eps / 63$. Suppose further that $|v| > \frac{2}{\eps} \|
(1+|b|^\beta) \hat{A}(b) \|_{l^1}$. Then, for any $B,C \subset G'$ and
$m \geq 1$,
\[ \| \pi_B T_{G'G'}^m \pi_C \| \leq (1 + (2\Lambda)^{\beta-\lceil \beta
\rceil} \lceil \beta \rceil m^{\lceil \beta \rceil-1}) \left(
\frac{17}{18} \right)^m \sup_{\substack{b \in B \\ c \in C}} \frac{1}{1
+ |b-c|^\beta}, \]
where $\lceil \beta \rceil$ is the smallest integer greater or equal
than $\beta$.
\end{prop}

{\tt Step 3.} Now observe that, if $b \in G'_1$ and $c \in G'_4$ then
$$ |b-c| = |b-d' - (c-d')| \geq |c-d'| - |b-d'| \geq \frac{R}{2} -
\frac{R}{4} = \frac{R}{4}. $$
Thus, applying the last proposition with $\beta=2$ and recalling that
$G'_3 \subset G'$, for $m \geq 0$ we have
$$ \| \pi_{G'_1} T_{33}^m T_{34} \| \leq \| \pi_{G'_1} T_{G'G'}^m
T_{G'G'_4} \| = \| \pi_{G'_1} T_{G'G'}^{m+1} \pi_{G'_4} \| \leq
\frac{3(m+1)}{1 + \tfrac{1}{16} R^2} \left( \frac{17}{18} \right)^{m+1}.
$$
Furthermore, since $\pi_{G'_4} \pi_{G'_3} = \pi_{G'_4} \pi_{G'_1} = 0$
and $\pi_{G'_3} \pi_{G'_1} = \pi_{G'_1}$, from \eqref{Tj} we obtain
$$ W^{(j)}_{43} \pi_{G'_1} = \pi_{G'_4} T_{G'G'}^{j+1} \pi_{G'_3}
\pi_{G'_1} = \pi_{G'_4} T_{G'G'}^{j+1} \pi_{G'_1}. $$
Hence,
$$ \| W^{(j)}_{43} \pi_{G'_1} \| = \| \pi_{G'_4} T_{G'G'}^{j+1}
\pi_{G'_1} \| \leq \| T_{G'G'} \|^{j+1} < \left( \frac{17}{18}
\right)^{j+1}. $$
Therefore, for $0 \leq m < j$,
\[ \| \pi_{G'_1} T_{33}^m T_{34} W^{(j-m-1)} \pi_{G'_1} \| \leq \|
\pi_{G'_1} T_{33}^m T_{34} \| \, \| W^{(j-m-1)}_{43} \pi_{G'_1} \| \leq
\frac{3(m+1)}{1 + \tfrac{1}{16} R^2} \left( \frac{17}{18} \right)^{j+1}.
\]
Iterating the first expression in \eqref{XYWZ} we find that
\begin{equation} \label{itera}
\begin{split}
X^{(j)}_{33} & = T_{34} W^{(j-1)}_{43} + T_{33} X^{(j-1)}_{33} \\ & =
T_{34} W^{(j-1)}_{43} + T_{33} T_{34} W^{(j-2)}_{43} + T_{33}^2
X^{(j-2)}_{33} \\
& \;\; \vdots \\
& = T_{34} W^{(j-1)}_{43} + T_{33} T_{34} W^{(j-2)}_{43} + \cdots +
T_{33}^{j-2} T_{34} W_{43}^{(1)} + T_{33}^{j-1} T_{34} W_{43}^{(0)} \\ &
= \sum_{m=0}^{j-1} T_{33}^m T_{34} W^{(j-m-1)}_{43}.
\end{split}
\end{equation}
Thus, using the above inequality,
\begin{align*}
\| \pi_{G'_1} X^{(j)}_{33} \pi_{G'_1} \| & = \left \| \sum_{m=0}^{j-1}
\pi_{G'_1} T_{33}^m T_{34} W^{(j-m-1)}_{43} \pi_{G'_1} \right \| \leq
\sum_{m=0}^{j-1} \| \pi_{G'_1} T_{33}^m T_{34} W^{(j-m-1)}_{43}
\pi_{G'_1} \| \\
& \leq \frac{3}{1 + \tfrac{1}{16} R^2} \left( \frac{17}{18}
\right)^{j+1} \sum_{m=0}^{j-1} (m+1) = \frac{3}{2 + \tfrac{1}{8} R^2}
(j^2+j) \left( \frac{17}{18} \right)^{j+1}.
\end{align*}
Consequently,
$$ \left \| \pi_{G'_1} \left[ \sum_{j=1}^\infty X^{(j)}_{33} \right]
\pi_{G'_1} \right \| \leq \sum_{j=1}^\infty \| \pi_{G'_1} X^{(j)}_{33}
\pi_{G'_1} \| \leq \frac{3}{2 + \tfrac{1}{8} R^2} \sum_{j=1}^\infty
(j^2+j) \left( \frac{17}{18} \right)^{j+1} \leq \frac{C}{R^2}, $$
where $C$ is an universal constant. Finally, using this and \eqref{b4},
since $|z_{\mu,d'}| \leq 3|v|$ we have
$$ |\mathcal{R}_3(k)| = \left| \sum_{b,c \in G'_1}
\frac{f(d'-b)}{N_b(k)} \left[ \sum_{j=1}^\infty X^{(j)}_{33}
\right]_{b,c} \, g(c-d') \right| \leq \frac{6 C}{\Lambda} \| f \|_{l^1}
\| g \|_{l^1} \frac{1}{|z_{\mu,d'}| \, R^2}. $$
In view of \eqref{split} and \eqref{RR3} this completes the proof.
\end{proof}


\begin{proof}[Proof of Proposition \ref{p:decayT}]
For any $b,c \in \Gdual$ set $Q_{b,c} \coloneqq (1+|b-c|^\beta)
T_{b,c}$. We first claim that, for any $B,C \subset G'$,
\begin{equation} \label{Qbc1}
\sup_{b \in B} \sum_{c \in C} |Q_{b,c}| < \frac{17}{18} \qquad
\text{and} \qquad \sup_{c \in C} \sum_{b \in B} |Q_{b,c}| <
\frac{17}{18}.
\end{equation}
In fact, using the bounds \eqref{b5}, \eqref{b6} and $|k| \leq 3|v|$, it
follows that
\begin{align*}
\sup_{b \in B} \sum_{c \in C} |Q_{b,c}| & = \sup_{b \in B} \sum_{c \in
C} (1+|b-c|^\beta) \left| \frac{\hat{q}(b-c)}{N_c(k)} - \frac{2c \cdot
\hat{A}(b-c)}{N_c(k)} - \frac{2 k \cdot \hat{A}(b-c)}{N_c(k)} \right| \\
& \leq \| (1+|b|^\beta) \hat{q}(b) \|_{l^1} \frac{1}{\eps |v|} +
\frac{14}{\eps} \| (1+|b|^\beta) \hat{A}(b) \|_{l^1} < \frac{1}{2} +
\frac{4}{9} = \frac{17}{18},
\end{align*}
and similarly we prove the second bound in \eqref{Qbc1}. Furthermore,
since $|T_{b,c}| \leq |Q_{b,c}|$ for all $b,c \in \Gdual$, for any
integer $m \geq 1$ we have
$$ \sup_{b \in B} \sum_{c \in C} |(T_{BC}^m)_{b,c}| < \left(
\frac{17}{18} \right)^m \qquad \text{and} \qquad \sup_{c \in C} \sum_{b
\in B} |(T_{BC}^m)_{b,c}| < \left( \frac{17}{18} \right)^m. $$

Now, let $p$ be the smallest integer greater or equal than $\beta$, and
for any integer $m \geq 1$ and any $\xi_0, \xi_1, \dots, \xi_m \in
\Gdual$, let $b=\xi_0$ and $c=\xi_m$. Then,
\begin{equation} \label{bminusc}
\begin{split}
|b-c|^\beta & = (2\Lambda)^\beta \left[ \frac{|b-c|}{2\Lambda}
\right]^\beta \le (2\Lambda)^\beta \left[ \frac{|b-c|}{2\Lambda}
\right]^p = \frac{(2\Lambda)^\beta}{(2\Lambda)^p} \sum_{i_1,\dots,i_p
=1}^m |\xi_{i_1-1}-\xi_{i_1}| \cdots \, |\xi_{i_p-1}-\xi_{i_p}| \\
& \leq (2\Lambda)^{\beta-p} \sum_{i_1,\dots,i_p =1}^m (
|\xi_{i_1-1}-\xi_{i_1}|^p + \cdots + |\xi_{i_p-1}-\xi_{i_p}|^p) \\ & =
(2\Lambda)^{\beta-p} p \, m^{p-1} \sum_{i=1}^m |\xi_{i-1}-\xi_i|^p \le
(2\Lambda)^{\beta-p} p \, m^{p-1} \prod_{i=1}^m (1+|\xi_{i-1}-\xi_i|^p).
\end{split}
\end{equation}
To simplify the notation write $s \coloneqq \sup_{b \in B, \, c \in C}
\frac{1}{1 + |b-c|^\beta}$. Hence,
\begin{align*}
& \sup_{b \in B} \sum_{c \in C} |(T_{G'G'}^m)_{b,c}| \leq
\sup_{\substack{b \in B \\ c \in C}} \frac{1}{1 + |b-c|^\beta} \;
\sup_{b \in B} \sum_{c \in C} (1+|b-c|^\beta) |(T_{G'G'}^m)_{b,c}| \\
& \qquad \leq s \left[ \sup_{b \in B} \sum_{c \in C}
|(T_{G'G'}^m)_{b,c}| + (2\Lambda)^{\beta-p} p \, m^{p-1} \; \sup_{b \in
B} \, \sum_{\xi_1 \in G'} (1+|b-\xi_1|^\beta) |T_{b,\xi_1}| \right. \\
& \qquad \qquad \left. \times \sum_{\xi_2 \in G'} (1+|\xi_1-\xi_2|^2)
|T_{\xi_1,\xi_2}| \cdots \sum_{c \in C} (1+|\xi_{m-1} -c|^2)
|T_{\xi_{m-1},c}| \right] \\
& \qquad \leq s \left[ \left( \frac{17}{18} \right)^m +
(2\Lambda)^{\beta-p} p \, m^{p-1} \; \sup_{b \in B} \, \sum_{\xi_1 \in
G'} (1+|b-\xi_1|^2) |T_{b,\xi_1}| \right. \\
& \qquad \qquad \left. \times \sup_{\xi_1 \in G'} \sum_{\xi_2 \in G'}
(1+|\xi_1-\xi_2|^2) |T_{\xi_1,\xi_2}| \cdots \sup_{\xi_{m-1} \in G'}
\sum_{c \in C} (1+|\xi_{m-1} -c|^2) |T_{\xi_{m-1},c}| \right] \\
& \qquad \leq s \, ( 1 + (2\Lambda)^{\beta-p} p \, m^{p-1} ) \left(
\frac{17}{18} \right)^m,
\end{align*}
and similarly we prove the other inequality. Therefore, by Proposition
\ref{p:ineq},
\[ \| \pi_B T_{G'G'}^m \pi_C \| \leq (1 + (2\Lambda)^{\beta-\lceil \beta
\rceil} \lceil \beta \rceil m^{\lceil \beta \rceil-1}) \left(
\frac{17}{18} \right)^m \sup_{\substack{b \in B \\ c \in C}} \frac{1}{1
+ |b-c|^\beta}, \]
where $\lceil \beta \rceil$ is the smallest integer greater or equal
than $\beta$. This is the desired estimate.
\end{proof}


\begin{proof}[Proof of Lemma \ref{l:b1b}]
To simplify the notation write $w = w_{\mu,d'}$, $z = z_{\mu,d'}$, and
$|z|_R = 2|z|-R$. First observe that
$$ \frac{1}{w - 2i \theta_{\mu'}(c-d')} = \frac{-1}{2i
\theta_{\mu'}(c-d')} + \frac{w}{2i \theta_{\mu'}(c-d')(w - 2i
\theta_{\mu'}(c-d'))}, $$
so that
\begin{align*}
\frac{z}{N_c(k)} & = \frac{-1}{2i \theta_{\mu'}(c-d')} + \frac{w}{2i
\theta_{\mu'}(c-d')(w - 2i \theta_{\mu'}(c-d'))} + \frac{2i
\theta_\mu(c-d')}{w
- 2i \theta_{\mu'}(c-d')} \; \frac{1}{z - 2i \theta_\mu(c-d')} \\ &
\eqqcolon \eta_c^{(0)} + \eta_c^{(w)} + \eta_c^{(z)},
\end{align*}
where, in view of \eqref{est3} to \eqref{est6}, since $|w|<\eps$,
$$ |\eta_c^{(0)}| \leq \frac{1}{2 \Lambda}, \qquad |\eta_c^{(w)}| \leq
\frac{\eps}{2\Lambda^2} \qquad \text{and} \qquad |\eta_c^{(z)}| \leq
\frac{4}{|z|_R}. $$
Hence,
\begin{align*}
Y_{b,c} = \frac{-2i \theta_{\mu'}(\hat{A}(b-c)) z}{N_c(k)} & = -2i
\theta_{\mu'}(\hat{A}(b-c)) \eta_c^{(0)} -2i
\theta_{\mu'}(\hat{A}(b-c))\eta_c^{(w)} -2i
\theta_{\mu'}(\hat{A}(b-c))\eta_c^{(z)} \\
& \eqqcolon Y_{b,c}^{(0)} + Y_{b,c}^{(w)} + Y_{b,c}^{(z)}.
\end{align*}

Let $Y^{(\,\cdot \,)}$ be the operator whose matrix elements are $Y^{(\,
\cdot \,)}_{b,c}$ and set $Y_{33}^{(\, \cdot \,)} \coloneqq \pi_{G'_3}
Y^{(\, \cdot \,)} \pi_{G'_3} $. Then, similarly as we estimated $\|
Y_{33} \|$, using \eqref{est3} to \eqref{est6} and Proposition
\ref{p:ineq}, it follows easily that
$$ \| Y^{(0)}_{33} \| \leq \frac{1}{2 \Lambda} \| \theta_{\mu'}(\hat{A})
\|_{l^1}, \qquad \| Y^{(w)}_{33} \| \leq \frac{\eps}{2\Lambda^2} \|
\theta_{\mu'}(\hat{A}) \|_{l^1}, \qquad \| Y^{(z)}_{33} \| \leq
\frac{4}{|z|_R} \| \theta_{\mu'}(\hat{A}) \|_{l^1}. $$
Furthermore,
\begin{align*}
S & = (I-Y_{33})^{-1} = 1 + (1-Y_{33})^{-1} Y_{33} = 1 + SY_{33} \\
& = 1 + (1 + SY_{33}) Y_{33} = 1 + Y_{33}^{(0)} + Y_{33}^{(w)} +
Y_{33}^{(z)} + S Y_{33}^2,
\end{align*}
where, recalling \eqref{est2},
$$ \| S Y_{33}^2 \| \leq \| (1-Y_{33})^{-1} \| \, \| Y_{33} \|^2 \leq
\frac{\| Y_{33} \|^2}{1-\|Y_{33}\|} < \frac{14}{13} \left(
\frac{8}{\Lambda} \right)^2 \| \theta_{\mu'}(\hat{A}) \|_{l^1}^2. $$
Combining all this we have
\begin{align*}
\frac{z \, S_{b,c}}{N_b(k)} & = (\eta_b^{(0)} +
\eta_b^{(w)})(\delta_{b,c} + Y_{b,c}^{(0)} + Y_{b,c}^{(w)} +
Y_{b,c}^{(z)} + (SY_{33}^2)_{b,c} ) + \eta_b^{(z)} S_{b,c} \\
& = \Big[ \eta_b^{(0)} (\delta_{b,c} + Y_{b,c}^{(0)}) \Big] + \Big[
\eta_b^{(0)} Y_{b,c}^{(w)} + \eta_b^{(w)} (\delta_{b,c} + Y_{b,c}^{(0)}
+ Y_{b,c}^{(w)}) \Big] \\
& \qquad + \Big[ (\eta_b^{(0)} + \eta_b^{(w)}) (SY_{33}^2)_{b,c} \Big] +
\Big[ (\eta_b^{(0)} + \eta_b^{(w)}) Y_{b,c}^{(z)} + \eta_b^{(z)} S_{b,c}
\Big] \\ & \eqqcolon K_{b,c}^{(0)} + K_{b,c}^{(1)} + K_{b,c}^{(2)} +
K_{b,c}^{(3)}
\end{align*}
with
\begin{align*}
|K_{b,c}^{(0)}| & \leq \frac{1}{2 \Lambda} \left( 1 + \frac{1}{2
\Lambda} \| |\theta_{\mu'}(\hat{A}) \|_{l^1} \right), \\
|K_{b,c}^{(1)}| & \leq \frac{\eps}{4 \Lambda^3} \|
\theta_{\mu'}(\hat{A}) |\|_{l^1} + \frac{\eps}{2\Lambda^2} \left( 1 +
\frac{1}{2 \Lambda} \| |\theta_{\mu'}(\hat{A}) \|_{l^1} +
\frac{\eps}{\Lambda^2} \| |\theta_{\mu'}(\hat{A}) \|_{l^1} \right) \\
& < \frac{\eps}{2\Lambda^2} \left(1 + \frac{7}{6 \Lambda} \|
\theta_{\mu'}(\hat{A}) \|_{l^1} \right), \\
|K_{b,c}^{(2)}| & \leq \frac{1}{\Lambda} \left( \frac{8}{\Lambda}
\right)^2 \| \theta_{\mu'}(\hat{A}) \|_{l^1}^2 < \frac{64}{\Lambda^3} \,
\| \theta_{\mu'}(\hat{A}) \|_{l^1}^2, \\
|K_{b,c}^{(3)}| & \leq \frac{3}{2 \Lambda} \| \theta_{\mu'}(\hat{A})
\|_{l^1} \frac{4}{|z|_R} + \frac{14}{13} \, \frac{4}{|z|_R} <
\frac{C_{\Lambda,A}}{|z|_R}
\end{align*}
for all $b,c \in G_3'$. Here, to estimate $|K_{b,c}^{(1)}|$ we have used
that $\eps < \Lambda/6$.

Finally, recalling \eqref{alpha} and using the above estimates we find
that
\begin{equation} \label{a1j}
\begin{aligned}
z_{\mu,d'}(k) \alpha^{(1)}_{\mu,d'}(k) & = \sum_{b,c \in G'_1} f(d'-b)
\frac{z \, S_{b,c}}{N_b(k)} g(c-d') \\
& = \sum_{b,c \in G'_1} f(d'-b) \Bigg[ \sum_{j=0}^3 K_{b,c}^{(j)} 
\Bigg] g(c-d') \\
& \eqqcolon \alpha^{(1,0)}_{\mu,d'} + \alpha^{(1,1)}_{\mu,d'}(w(k)) +
\alpha^{(1,2)}_{\mu,d'}(k) + \alpha^{(1,3)}_{\mu,d'}(k),
\end{aligned}
\end{equation}
where, in particular,
\begin{equation} \label{a10}
\alpha^{(1,0)}_{\mu,d'} = - \sum_{b,c \in G'_1} \frac{f(d'-b)}{2i
\theta_{\mu'}(b-d')} \left[ \delta_{b,c} +
\frac{\theta_{\mu'}(\hat{A}(b-c))}{ \theta_{\mu'}(c-d')} \right]
g(c-d').
\end{equation}
Furthermore, for $0 \le j \le 2$, it follows easily from \eqref{a1j}
that $|\alpha^{(1,j)}_{\mu,d'}| \leq C_j$ with
\begin{equation} \label{cts}
\begin{aligned}
C_0 & \coloneqq \frac{1}{2 \Lambda} \left( 1 + \frac{1}{2 \Lambda} \|
\theta_{\mu'}(\hat{A}) \|_{l^1} \right) \| f \|_{l^1} \| g \|_{l^1}, \\
C_1 & \coloneqq \frac{\eps}{2\Lambda^2} \left(1 + \frac{7}{6 \Lambda} \|
\theta_{\mu'}(\hat{A}) \|_{l^1} \right) \| f \|_{l^1} \| g \|_{l^1}, \\
C_2 & \coloneqq \frac{64}{\Lambda^3} \, \| \theta_{\mu'}(\hat{A})
\|_{l^1}^2 \| f \|_{l^1} \| g \|_{l^1},
\end{aligned}
\end{equation}
while for $j=3$,
\[ |\alpha^{(1,3)}_{\mu,d'} | \leq C_{\Lambda,A,f,g} \frac{1}{|z|_R}. \]
This completes the proof of the lemma.
\end{proof}


\begin{proof}[Proof of Lemma \ref{l:b2}]
To prove this lemma we apply the following (well-known) inequality (see
\cite{deO} for a proof).

\begin{prop} \label{p:decay}
Let $\alpha$ and $\delta$ be constants with $1 < \alpha \leq 2$ and $1 <
\delta \le 2$. Suppose that $f$ is a function on $\Gdual$ obeying $\|
|b|^\alpha f(b) \|_{l^1} < \infty$. Then, for any $\xi_1, \xi_2 \in
\Gdual$ with $\xi_1 \neq \xi_2$,
$$ \sum_{b \in \Gdual \setminus \{\xi_1,\xi_2\}}
\frac{|f(b-\xi_1)|}{|b-\xi_2|^\delta} \leq \frac{C}{|\xi_1 -
\xi_2|^{\alpha + \delta -2}} \times \begin{cases} 1 & \text{\rm if}
\quad \alpha, \delta < 2, \\ \ln |\xi_1-\xi_2| & \text{\rm if} \quad
\alpha = 2 \text{ or } \delta = 2, \end{cases} $$
where $C = C_{\Gdual, \alpha, \delta, f}$ is a constant.
\end{prop}

First observe that $\| \pi_{\{b\}} T_{G'G'}^m \pi_{\{c\}} \| =
|(T_{G'G'}^m)_{b,c} |$. Hence, by Proposition \ref{p:decayT} with
$\beta=2$, for all $b,c \in G'$ and $m \geq 1$,
$$ |(T_{G'G'}^m)_{b,c} | = \| \pi_{\{b\}} T_{G'G'}^m \pi_{\{c\}} \| \leq
(1 + 2m) \left( \frac{17}{18} \right)^m \frac{1}{1 + |b-c|^2}. $$
Note that this inequality is also valid for $m=0$. Thus,
\begin{equation} \label{e1}
\begin{split}
| \Phi_{d',d''}(k) | & = \left| \sum_{m=0}^\infty \sum_{b,c \in G'}
\frac{f(d'-b)}{N_{b}(k)} (T_{G'G'}^m)_{b,c} \, g(c-d'') \right| \\ &
\leq \frac{1}{\eps |v|} \left[ \sum_{m=0}^\infty (1 + 2m) \left(
\frac{17}{18} \right)^m \right] \sum_{b \in G'} |f(d'-b)| \sum_{c \in
G'} \frac{|g(c-d'')|}{1 + |b-c|^2} \\
& \leq \frac{C}{\eps |v|} \sum_{b \in G'} |f(d'-b)| \left[ |g(b-d'')| +
\sum_{c \in G' \setminus \{b\}} \frac{|g(c-d'')|}{|b-c|^2} \right],
\end{split}
\end{equation}
where $C$ is an universal constant.

Now, by the triangle inequality, H\"older's inequality, and since $\|
\cdot \|_{l_2} \leq \| \cdot \|_{l_1}$,
\begin{equation} \label{fg}
\begin{split}
& \sum_{b \in G'} |f(d'-b)| \, |g(b-d'')| \\
& = \sum_{b \in G' } \frac{|d'-d''|^2}{|d'-d''|^2} \, |f(d'-b)| \, 
|g(b-d'')| \\ 
& \leq \frac{4}{|d'-d''|^2} \sum_{b \in G'} ( |d'-b|^2 + |b-d''|^2) \,
|f(d'-b)| \, |g(b-d'')| \\
& \leq \frac{4}{|d'-d''|^2} ( \| b^2 f(b) \|_{l^2} \| g \|_{l^2} + \| f
\|_{l^2} \| b^2 g(b) \|_{l^2} ) \\
& \leq \frac{4}{|d'-d''|^2} ( \| b^2 f(b) \|_{l^1} \| g \|_{l^1} + \| f
\|_{l^1} \| b^2 g(b) \|_{l^1} ) \leq \frac{C_{f,g}}{|d'-d''|^2}.
\end{split}
\end{equation}
Furthermore, by Proposition \ref{p:decay} with $\alpha=\delta=2$, for
any $0 < \epsilon_1 < 2$,
$$ \sum_{c \in G' \setminus \{b\}} \frac{|g(c-d'')|}{|b-c|^2} \leq
C_{\Gdual,g} \frac{\ln |b-d''|}{|b-d''|^2} \leq
\frac{C_{\Gdual,g,\epsilon_1}}{|b-d''|^{2-\epsilon_1}}. $$
Applying this inequality and \eqref{fg} to \eqref{e1} we obtain
$$ | \Phi_{d',d''}(k) | \leq \frac{C}{\eps |v|} \left[
\frac{C_{f,g}}{|d'-d''|^2} + C_{\Gdual,g,\epsilon_1} \sum_{b \in G'}
\frac{|f(d'-b)|}{|b-d''|^{2-\epsilon_1}} \right]. $$
Again, by Proposition \ref{p:decay} with $\alpha=2$ and
$\delta=2-\epsilon_1$ we conclude that, for any $0 < \epsilon_2 <
2-\epsilon_1$,
\[ | \Phi_{d',d''}(k) | \leq \frac{C}{\eps |v|} \left[
\frac{C_{f,g}}{|d'-d''|^2} + C_{\Gdual,f,g,\epsilon_1}
\frac{\ln|d'-d''|}{|d'-d''|^{2-\epsilon_1}} \right] \leq
\frac{C_{\eps,\Gdual,f,g,\epsilon_1, \epsilon_2}}{|v| \,
|d'-d''|^{2-\epsilon_1-\epsilon_2}}. \]
Finally, recall from Proposition \ref{p:order}(ii) that $|z_{\nu',d}| <
3|d|$ and $|z_{\nu',d}| < 3|v|$, observe that $|d'-d''|=|d|$, and set
$\epsilon = \epsilon_1 + \epsilon_2$. Then, for any $0 < \epsilon < 2$,
$$ | \Phi_{d',d''}(k) | \leq \frac{C_{\eps,\Gdual,f,g,\epsilon_1,
\epsilon_2}}{|d| \, |d|^{2-\epsilon_1-\epsilon_2}} \leq
\frac{C_{\eps,\Gdual,f,g,\epsilon}}{|z_{\nu',d}|^{3-\epsilon}}. $$
Choosing $\epsilon = 10^{-1}$ we obtain the desired inequality.
\end{proof}


\section{Bounds on the derivatives: proofs} \label{s:app3}

\begin{proof}[Proof of Lemma \ref{l:der}]
{\tt Step 0.} When there is no risk of confusion we shall use the same
notation to denote an operator or its matrix. Define
$$ \mathcal{F}_{BC} \coloneqq \left[ f(b-c) \right]_{b \in B, c \in C},
\qquad \mathcal{G}_{BC} \coloneqq [ g(b-c) ]_{b \in B, c \in C}, \qquad
\Phi_G(k) \coloneqq \left[ \Phi_{d',d''}(k;G) \right]_{d',d'' \in G}. $$
Here $\mathcal{F}_{BC}$ and $\mathcal{G}_{BC}$ are $|B| \times |C|$
matrices and $\Phi_G(k)$ is a $|G| \times |G|$ matrix. First observe
that
$$ \Phi_G(k) = \left[ \sum_{b,c \in G'} \frac{f(d'-b)}{N_b(k)}
(R_{G'G'}^{-1})_{b,c} \, g(c-d'') \right]_{d',d'' \in G} $$
can be written as the product of matrices $\F_{GG'} \Delta_k^{-1}
R_{G'G'}^{-1} \G_{G'G}$. Furthermore, since on $L^2_{G'}$ we have
$\Delta_k^{-1} R_{G'G'}^{-1} = ( R_{G'G'} \Delta_k )^{-1} = H_k^{-1}$,
we can write $\Phi_G(k)$ as $\F_{GG'} H_k^{-1} \G_{G'G}$. Hence,
\begin{equation} \label{drop1}
\frac{\partial^{n+m}}{\partial k_1^n \partial k_2^m} \Phi_G(k) =
\F_{GG'} \frac{\partial^{n+m} H_k^{-1}}{\partial k_1^n \partial k_2^m}
\G_{G'G}.
\end{equation}
This is the quantity we want to estimate.

{\tt Step 1.} Let $T=T(k)$ be an invertible matrix. Then applying
$\frac{\partial^{m_0}}{ \partial k_i^{m_0}}$ to the identity $T
T^{-1}=I$ and using the Leibniz rule for $\frac{\partial^{m_0}}{
\partial k_i^{m_0}} (T T^{-1})$ we find that 
$$ \frac{\partial^{m_0} T^{-1}}{\partial k_i^{m_0}} = -T^{-1} \sum_{m_1
= 0}^{m_0-1} \binom{m_0}{m_1} \frac{\partial^{m_0 - m_1} T}{\partial
k_i^{m_0 - m_1}} \; \frac{\partial^{m_1} T^{-1}}{\partial k_i^{m_1}}. $$
Iterating this formula $m_0-1$ times we obtain
\begin{equation} \label{rule}
\begin{split}
\frac{\partial^{m_0} T^{-1}}{\partial k_i^{m_0}} & = \left [
\prod_{j=1}^{m_0} \sum_{m_j=0}^{m_{j-1}-1} \binom{m_{j-1}}{m_j}
(-T^{-1}) \frac{\partial^{m_{j-1}-m_j} T}{\partial k_i^{m_{j-1} - m_j}}
\right ] \frac{\partial^{m_{m_0}} T^{-1}}{\partial k_i^{m_{m_0}}} \\
& = \left [ \prod_{j=1}^{m_0-1} \sum_{m_j=0}^{m_{j-1}-1}
\binom{m_{j-1}}{m_j} (-T^{-1}) \frac{\partial^{m_{j-1}-m_j} T}{\partial
k_i^{m_{j-1} - m_j}} \right ] \\
& \qquad \qquad \times \sum_{m_{m_0}=0}^{m_{m_0-1}-1}
\binom{m_{m_0-1}}{m_{m_0}} (-T^{-1}) \frac{\partial^{m_{m_0-1}-m_{m_0}}
T}{\partial k_i^{m_{m_0-1} - m_{m_0}}} \; \frac{\partial^{m_{m_0}}
T^{-1}}{\partial k_i^{m_{m_0}}}\\
& = (-1)^{m_0} \left [ \prod_{j=1}^{m_0-1} \sum_{m_j=0}^{m_{j-1}-1}
\binom{m_{j-1}}{m_j} T^{-1} \frac{\partial^{m_{j-1}-m_j} T}{\partial
k_i^{m_{j-1} - m_j}} \right ] T^{-1} \frac{\partial^{m_{m_0}-1}
T}{\partial k_i^{m_{m_0-1}}} T^{-1}.
\end{split}
\end{equation}

{\tt Step 2.} In view of \eqref{rule}, it is not difficult to see that
$\frac{\partial^m H_k^{-1}}{\partial k_2^m}$ is given by a finite linear
combination of terms of the form
\begin{equation} \label{form1}
\left [ \prod_{j=1}^{m} \, H_k^{-1} \frac{\partial^{n_j} H_k}{\partial
k_2^{n_j}} \right ] H_k^{-1},
\end{equation}
where $\sum_{j=1}^{m} n_j = m$. Thus, when we compute
$\frac{\partial^n}{\partial k_1^n} \, \frac{\partial^m
H_k^{-1}}{\partial k_2^m}$, the derivative $\frac{\partial^n}{\partial
k_1^n}$ acts either on $H_k^{-1}$ or $\frac{\partial^{n_j} H_k}{\partial
k_2^{n_j}}$. However, since $\big( \frac{\partial H_k}{\partial k_2}
\big)_{b,c} = 2(k_2 + b_2) \delta_{b,c} - 2 \hat{A}_2(b-c)$, we have
$\frac{\partial^n}{\partial k_1^n} \frac{\partial^{n_j} H_k}{\partial
k_2^{n_j}} = 0$ if $n_j \geq 1$ and $\frac{\partial^n}{\partial k_1^n}
\frac{\partial^{n_j} H_k}{\partial k_2^{n_j}} = \frac{\partial^n
H_k}{\partial k_1^n}$ if $n_j=0$. Similarly, using again \eqref{rule}
one can see that $\frac{\partial^n H_k^{-1}}{\partial k_1^n}$ is given
by a finite linear combination of terms of the form \eqref{form1}, with
$m$ and $k_2$ replaced by $n$ and $k_1$, respectively, and
$\sum_{j=1}^{n} n_j = n$. Therefore, combining all this we conclude that
$\frac{\partial^{n+m} H_k^{-1}}{\partial k_1^n \partial k_2^m}$ is given
by a finite linear combination of terms of the form
\begin{equation} \label{prod} 
\left [ \prod_{j=1}^{n+m} \, \Delta_k^{-1} R_{G'G'}^{-1}
\frac{\partial^{n_j} H_k}{\partial k_{i_j}^{n_j}} \right ] \Delta_k^{-1}
R_{G'G'}^{-1},
\end{equation}
where $\sum_{j=1}^{n+m} n_j \delta_{2,i_j} = m$ and $\sum_{j=1}^{n+m}
n_j \delta_{1,i_j} = n$, that is, where the sum of $n_j$ for which
$i_j=2$ is equal to $m$, and the sum of $n_j$ for which $i_j=1$ is equal
to $n$.

{\tt Step 3.} The first step in bounding \eqref{prod} is to estimate
$\Big \| \frac{\partial^{n_j}H_k}{\partial k_{i_j}^{n_j}} \Delta_k^{-1}
\pi_{G'} \Big \|$. A simple calculation shows that
$$ \left( \frac{\partial^{n_j} H_k}{\partial k_{i_j}^{n_j}}
\Delta_k^{-1} \right)_{b,c} = \frac{1}{N_c(k)} \times \begin{cases}
2(k_{i_j} + b_{i_j}) \delta_{b,c} + 2 \hat{A}_{i_j}(b-c) & \text{if}
\quad n_j = 1, \\ 2 \delta_{b,c} & \text{if} \quad n_j = 2, \\ 0 &
\text{if} \quad n_j \geq 3. \end{cases} $$
Furthermore, by Proposition \ref{p:Gprime},
$$ \frac{1}{|N_b(k)|} \leq \frac{1}{\eps |v|} $$
for all $b \in G'$, while by Proposition \ref{p:Nb} we have
\begin{equation} \label{1/N}
\frac{1}{|N_b(k)|} \leq \frac{2}{\Lambda |v|}
\end{equation}
and
$$ |k_i + b_i| \leq |u_i + b_i| + |v_i| \leq |v| + |u+b| \leq
\frac{2}{\Lambda} |N_b(k)| $$
for all $b \in G'$ if $G=\{0,d\}$, and for all $b \in G' \setminus
\{\tilde{b}\}$ if $G=\{0\}$. Furthermore,
\begin{equation} \label{tildeb}
|\tilde{b}| \leq \Lambda + |u| + |v| < \Lambda + 3|v|,
\end{equation}
since $|u| < 2|v|$ because $k \in T_0$. Now, let $\mathbf{1}_B(x)$ be
the characteristic function of the set $B$. Then, using the above
estimates we have
%
%
\begin{align*}
& \sup_{c \in G'} \sum_{b \in G'} \Big| \Big(
\frac{\partial^{n_j}H_k}{\partial k_{i_j}^{n_j}} \Delta_k^{-1} \pi_{G'}
\Big)_{b,c} \Big| \\
\displaybreak
& \leq \sup_{c \in G'} \sum_{b \in G'} \left[ \frac{2|k_{i_j} + b_{i_j}|
\delta_{n_j,1} + 2 \delta_{n_j,2}}{|N_b(k)|} \delta_{b,c} + \frac{2
|\hat{A}_{i_j}(b-c)|}{|N_b(k)|} \delta_{n_j,1} \right] \\
& \leq \sup_{c \in G'} \left[ \frac{2|k_{i_j} + \tilde{b}_{i_j}| +
2}{|N_{\tilde{b}}(k)|} \delta_{\tilde{b},c} + \frac{2
|\hat{A}_{i_j}(\tilde{b}-c)|}{|N_{\tilde{b}}(k)|} \right]
\mathbf{1}_{G'}(\tilde{b}) \\
& \qquad \qquad \qquad \qquad \qquad \qquad \qquad \qquad + \sup_{c \in
G'} \sum_{b \in G'\setminus\{\tilde{b}\}} \left[ \frac{2|k_{i_j} +
b_{i_j}| + 2}{|N_b(k)|} \delta_{b,c} + \frac{2
|\hat{A}_{i_j}(b-c)|}{|N_b(k)|} \right] \\
& \leq \frac{2|k_{i_j} + \tilde{b}_{i_j}| + 2 + 2\| \hat{A}
\|_{l^1}}{\eps |v|} \mathbf{1}_{G'}(\tilde{b}) + \sup_{c \in G'} \sum_{b
\in G' \setminus \{ \tilde{b} \}} \left[ \left[ \frac{4}{\Lambda} +
\frac{2}{|N_b(k)|} \right] \delta_{b,c} + \frac{2
|\hat{A}_{i_j}(b-c)|}{|N_b(k)|} \right] \\
& \leq \frac{2}{\eps |v|}(2(|u| + |v| + |\tilde{b}|) + 2 + 2 \| \hat{A}
\|_{l^1}) \mathbf{1}_{G'}(\tilde{b}) + \frac{4}{\Lambda} +
\frac{4}{\Lambda |v|} + \frac{4}{ \Lambda |v|} \| \hat{A} \|_{l_1} \\
& \leq \frac{2}{\eps |v|}(12|v| + 2 \Lambda + 2 + 2 \| \hat{A} \|_{l_1}
) \mathbf{1}_{G'}(\tilde{b}) + \frac{4}{\Lambda} + \frac{4}{\Lambda |v|}
+ \frac{4}{ \Lambda |v|} \| \hat{A} \|_{l_1} \leq
\mathbf{1}_{G'}(\tilde{b}) \, \eps^{-1} C_{\Lambda, A} + C_{\Lambda, A}.
\end{align*}
%
%
Similarly,
$$ \sup_{b \in G'} \sum_{c \in G'} \Big| \Big(
\frac{\partial^{n_j}H_k}{\partial k_{i_j}^{n_j}} \Delta_k^{-1} \pi_{G'}
\Big)_{b,c} \Big| \leq \mathbf{1}_{G'}(\tilde{b}) \, \eps^{-1}
C_{\Lambda, A} + C_{\Lambda, A}. $$
Hence, by Proposition \ref{p:ineq},
$$ \left\| \frac{\partial^{n_j}H_k}{\partial k_{i_j}^{n_j}}
\Delta_k^{-1} \pi_{G'} \right \| \leq \mathbf{1}_{G'}(\tilde{b}) \,
\eps^{-1} C_{\Lambda, A} + C_{\Lambda, A}. $$

{\tt Step 4.} By a similar (and much simpler) calculation (using
Proposition \ref{p:ineq}) we get
\begin{equation} \label{FG}
\begin{split}
\| \F_{GG'} \| & \leq \| f \|_{l^1}, \\
\| \G_{GG'} \| & \leq \| g \|_{l^1}, \\
\| \Delta_k^{-1} \pi_{G'} \| & \leq \mathbf{1}_{G'}(\tilde{b}) \,
\frac{1}{\eps |v|} + (1-\mathbf{1}_{G'}(\tilde{b})) \frac{2}{\Lambda
|v|}.
\end{split}
\end{equation}
From Lemma \ref{l:invR} we have $\| (R_{G'G'})^{-1} \| \leq 18$. Thus,
the operator norm of \eqref{prod} is bounded by
$$ \left \| \left[ \prod_{j=1}^{n+m} \, \Delta_k^{-1} R_{G'G'}^{-1}
\frac{\partial^{n_j} H_k}{\partial k_{i_j}^{n_j}} \right] \Delta_k^{-1}
R_{G'G'}^{-1} \right \| \leq \| \Delta_k^{-1} \| \, \| R_{G'G'}^{-1} \|
\left[ \prod_{j=1}^{n+m} \left\| \frac{\partial^{n_j}H_k}{\partial
k_{i_j}^{n_j}} \Delta_k^{-1} \pi_{G'} \right \| \| R_{G'G'}^{-1} \|
\right], $$
which is bounded either by
$$ \frac{1}{\eps |v|} \, 18 \left[ \prod_{j=1}^{n+m} (\eps^{-1}
C_{\Lambda, A} + C_{\Lambda,A}) \, 18 \right] \leq
\eps^{-(n+m+1)}C_{\Lambda, A,n,m} \, \frac{1}{|v|} $$
if $G = \{0\}$, or by
$$ \frac{1}{\Lambda |v|} \, 18 \left[ \prod_{j=1}^{n+m} C_{\Lambda,A} \,
18 \right] \| g \|_{l^1} \leq  C_{\Lambda, A,n,m} \, \frac{1}{|v|} $$
if $G = \{0, d\}$.
Therefore,
\begin{equation} \label{drop2}
\left \| \frac{\partial^{n+m} H_k^{-1}}{\partial k_1^n \partial k_2^m}
\right \| \leq \sum_{\substack{\text{finite sum where} \\ \text{\# of
terms depend} \\ \text{ on $n$ and $m$}}} \frac{C'}{|v|} \leq C_{n,m}
\frac{C'}{|v|} \leq \frac{C}{|v|},
\end{equation}
with $C = C_{\eps,\Lambda, A,n,m}$ if $G=\{0\}$ or $C = C_{\Lambda, A,
n,m}$ if $G = \{0, d\}$. Finally, recalling \eqref{drop1} and \eqref{FG}
we have
$$ \left| \frac{\partial^{n+m}}{\partial k_1^n \partial k_2^m} \Phi_G(k)
\right| = \left| \F_{GG'} \frac{\partial^{n+m} H_k^{-1}}{\partial k_1^n
\partial k_2^m} \G_{G'G} \right| \leq \| \F_{GG'} \| \left\|
\frac{\partial^{n+m} H_k^{-1}}{\partial k_1^n \partial k_2^m} \right \|
\| \G_{G'G} \| \leq \frac{C}{|v|}, $$
where
$C = C_{\eps,\Lambda, A,n,m,f,g}$ if $G=\{0\}$ or $C=C_{\Lambda,
A,n,m,f,g}$ if $G = \{0, d\}$. This is the desired inequality. The proof
of the lemma is complete.
\end{proof}


\begin{proof}[Proof of Lemma \ref{l:derii}] 
Let $\R^+$ be the set of non-negative real numbers and let $\sigma$ be a
real-valued function on $\R^+$ such that:
\begin{itemize}
\item[]
\begin{itemize}
\item[\rm (i)] $\sigma(t) \ge 1$ for all $t \in \R^+$ with
$\sigma(0)=1$;

\item[\rm (ii)] $\sigma(s)\sigma(t) \ge \sigma(s+t)$ for all $s,t \in
\R^+$;

\item[\rm (iii)] $\sigma$ increases monotonically.
\end{itemize}
\end{itemize}
For example, for any $\beta \geq 0$ the functions $t \mapsto e^{\beta
t}$ and $t \mapsto (1+t)^\beta$ satisfy these properties. Now, let $T$
be a linear operator from $L^2_C$ to $L^2_B$ with $B, C \subset \Gdual$
(or a matrix $T=[T_{b,c}]$ with $b \in B$ and $c \in C$) and consider
the $\sigma$-norm
$$ \| T \|_\sigma \coloneqq \max \left \{ \sup_{b \in B} \sum_{c \in C}
|T_{b,c}| \sigma(|b-c|), \; \sup_{c \in C} \sum_{b \in B} |T_{b,c}|
\sigma(|b-c|) \right \}. $$
In \cite{deO} we prove that this norm has the following properties.

\begin{prop}[Properties of $\| \cdot \|_\sigma$] \label{p:nsigma}
Let $S$ and $T$ be linear operators from $L^2_C$ to $L^2_B$ with $B, C
\subset \Gdual$. Then:
\begin{itemize}
\item[\rm (a)] $\| T \| \le \| T \|_{\sigma \equiv 1} \le \| T
\|_\sigma$;

\item[\rm (b)] If $B=C$, then $\| S \, T \|_\sigma \leq \|S \|_\sigma \|
T \|_\sigma$;

\item[\rm (c)] If $B=C$, then $\| (I+T)^{-1} \|_\sigma \le
(1-\|T\|_\sigma)^{-1}$ if $\| T \|_\sigma < 1$;

\item[\rm (d)] $|T_{b,c}| \leq \frac{1}{\sigma(|b-c|)} \| T \|_\sigma$
for all $b \in B$ and all $c \in C$.
\end{itemize}
\end{prop}

Now, by using these properties we prove Lemma \ref{l:derii}. We follow
the same notation as above. First observe that, similarly as in the last
proof we can write
$$ \Phi_{d',d''}(k) = \mathcal{F}_{\{d'\} G'} \Delta_k^{-1}
R_{G'G'}^{-1} \mathcal{G}_{G' \{d''\}} = \mathcal{F}_{\{d'\} G' }
H_k^{-1} \mathcal{G}_{G' \{d''\}}. $$
Now, let $\sigma(|b|) = (1+|b|)^\beta$, and observe that there is a
positive constant $C_\beta$ such that $\sigma(|b|) \le C_\beta
(1+|b|^\beta)$ for all $b \in \Gdual$. Then, it is easy to see that
\begin{align*}
\| \mathcal{F}_{\{d'\} G' } \|_\sigma & = \| f \|_\sigma \le C_\beta \|
(1+|b|^\beta) f(b) \|_{l^1}, \\
\| \mathcal{G}_{G'\{d''\} } \|_\sigma & = \| g \|_\sigma \le C_\beta\|
(1+|b|^\beta) g(b) \|_{l^1}.
\end{align*}
Furthermore, by \eqref{Qbc1} and Proposition \ref{p:ineq},
\begin{equation} \label{Rsigma}
\| R_{G'G'}^{-1} \|_\sigma = \| (I+T_{G'G'})^{-1} \|_\sigma \le
\sum_{j=0}^\infty \| T_{G'G} \|^j_\sigma < 18,
\end{equation}
and since for diagonal operators the $\sigma$-norm and the operator norm
agree, from \eqref{FG} we have
$$ \| \Delta_k^{-1} \pi_{G'} \|_\sigma \le \frac{2}{\Lambda |v|}. $$
Hence, in view of Proposition \ref{p:nsigma}(b) and Proposition
\ref{p:order}(ii),
$$ |\Phi_{d',d''}(k)| \le \| \mathcal{F}_{\{d'\} G'} \Delta_k^{-1}
R_{G'G'}^{-1} \mathcal{G}_{G' \{d''\}} \| \le
C_{\beta,f,g,\Lambda,A,m,n} \, \frac{1}{|d|}, $$
and by repeating the proof of Lemma \ref{l:der} with the operator norm
replaced by the $\sigma$-norm we obtain 
$$ \left \| \frac{\partial^{n+m}}{\partial k_1^n \partial k_2^m}
\Phi_{d',d''}(k) \right \|_\sigma \le C_{\beta,f,g,\Lambda,A,m,n} \,
\frac{1}{|d|}. $$
Therefore, by Proposition \ref{p:nsigma}(d), for any integers $n$ and
$m$ with $n+m \ge 0$,
$$ \left| \frac{\partial^{n+m}}{\partial k_1^n \partial k_2^m}
\Phi_{d',d''}(k) \right| \le \frac{1}{1+|d'-d''|^\beta} \left \|
\frac{\partial^{n+m}}{\partial k_1^n \partial k_2^m} \Phi_{d',d''}(k)
\right \|_\sigma \leq C_{\beta,f,g,\Lambda,A,m,n} \,
\frac{1}{|d|^{1+\beta}}. $$
This is the desired inequality.
\end{proof}


\subsection*{Proof of Lemma \ref{l:der2}}

Define the operator $M^{(j)} : L^2_{G'_3} \to L^2_{G'_3}$ as
$$ M^{(j)} \coloneqq \begin{cases} S & \text{if} \quad j=1, \\ W &
\text{if} \quad j=2, \\ Z & \text{if} \quad j=3, \end{cases} $$
where $S$, $W$ and $Z$ are given by \eqref{WS}. In order to prove Lemma
\ref{l:der2} we first prove the following proposition.

\begin{prop} \label{p:dSWZ}
Assume the same hypotheses of Lemma \ref{l:der2}. Then, for any integers
$n$ and $m$ with $n+m \ge 1$ and for $1 \le j \le 3$,
$$ \left \| \frac{\partial^{n+m}}{\partial k_1^n \partial k_2^m}
\Delta_k^{-1} M^{(j)} \right \| \leq \frac{C_j}{(2|z_{\mu,d'}(k)|-R)^j},
$$
where $C_1 = C_{1;\Lambda, A,n,m}$ and $C_j = C_{j;\Lambda,A,q,n,m}$ for
$2 \le j \le 3$ are constants. Furthermore,
$$ C_{1;\Lambda, A, 1, 0} \leq \frac{13}{\Lambda^2}, \qquad
C_{1;\Lambda, A, 0, 1} \leq \frac{13}{\Lambda^2} \qquad \text{and}
\qquad C_{1;\Lambda, A, 1, 1} \leq \frac{65}{\Lambda^3}. $$
\end{prop}

\begin{proof}
{\tt Step 0.} To simplify the notation write $w = w_{\mu,d'}$, $z =
z_{\mu,d'}$ and $|z|_R = 2|z|-R$. First observe that, for any analytic
function of the form $h(k)=\tilde{h}(w(k),z(k))$ we have
$$ \frac{\partial}{\partial k_1} h = \left( \frac{\partial}{\partial w}
+ \frac{\partial}{\partial z} \right) \tilde{h}, \qquad
\frac{\partial}{\partial k_2} h = i(-1)^\nu \left(
\frac{\partial}{\partial w} - \frac{\partial}{\partial z} \right)
\tilde{h}. $$
Thus,
\begin{align*}
\left\| \frac{\partial^{n+m}}{\partial k_1^n \partial k_2^m}
\Delta_k^{-1} M^{(j)} \right \| & = \left\| (i(-1)^\nu)^m \sum_{p=0}^m
\sum_{r=0}^n \binom{m}{p} \binom{n}{r} (-1)^{m-p}
\frac{\partial^{n-r+m-p}}{\partial z^{n-r+m-p}}
\frac{\partial^{r+p}}{\partial w^{r+p}} \Delta_k^{-1} M^{(j)} \right\|
\\
& \leq 2^{n+m} \sup_{p \le r} \sup_{r \le n} \left\|
\frac{\partial^{n-r+m-p}}{\partial z^{n-r+m-p}}
\frac{\partial^{r+p}}{\partial w^{r+p}} \Delta_k^{-1} M^{(j)} \right\|. 
\end{align*}
Now, by the Leibniz rule,
\begin{align*}
\left \| \frac{\partial^n}{\partial z^n} \frac{\partial^m}{\partial w^m}
\Delta_k^{-1} M^{(j)} \right \| & = \left \| \sum_{p=0}^m \sum_{r=0}^n
\binom{m}{p} \binom{n}{r} \frac{\partial^{n-r+m-p} \Delta_k^{-1}}{
\partial z^{n-r} \partial w^{m-p}} \, \frac{\partial^{r+p}
M^{(j)}}{\partial z^r \partial w^p} \right \| \\
& \leq 2^{n+m} \, \sup_{p \leq m} \sup_{r \leq n} \, \Bigg \|
\frac{\partial^{n-r+m-p} \Delta_k^{-1}}{ \partial z^{n-r} \partial
w^{m-p}} \Bigg \| \, \Bigg \| \frac{\partial^{r+p} M^{(j)}}{\partial z^r
\partial w^p} \Bigg \|.
\end{align*}
Furthermore, we shall prove below that
\begin{equation} \label{der}
\sup_{p \leq m} \sup_{r \leq n} \Bigg \| \frac{\partial^{n-r+m-p}
\Delta_k^{-1}}{ \partial z^{n-r} \partial w^{m-p}} \Bigg \| \, \Bigg \|
\frac{\partial^{r+p} M^{(j)}}{\partial z^r \partial w^p} \Bigg \| \leq
\frac{C_{j,n,m}}{|z|_R^{n+j}},
\end{equation}
with constants $C_{1,n,m} = C_{1,n,m;\Lambda,A}$ and $C_{j,n,m} =
C_{j,n,m;\Lambda,A,q}$ for $2 \le j \le 3$. Hence,
$$ \left \| \frac{\partial^n}{\partial z^n} \frac{\partial^m}{\partial
w^m} \Delta_k^{-1} M^{(j)} \right \| \leq 2^{n+m}
\frac{C_{j,n,m}}{|z|_R^{n+j}}. $$
Therefore, being careful with the indices,
$$ \left\| \frac{\partial^{n+m}}{\partial k_1^n \partial k_2^m}
\Delta_k^{-1} M^{(j)} \right \| \leq 2^{n+m} \sup_{p \leq m} \sup_{r
\leq n} \, 2^{n-r+m-p+r+p} \frac{C_{j,n-r+m-p,r+p}}{|z|_R^{n-r+m-p+j}}
\leq \frac{C_j}{|z|_R^j}, $$
where $C_1 = C_{1;\Lambda,A,n,m}$ and $C_j = C_{j;\Lambda,A,q,n,m}$ for
$2 \le j \le 3$. This is the desired inequality. We are left to prove
\eqref{der} and estimate the constants $C_{1;\Lambda,A,i,j}$ for $i,j\in
\{0,1\}$ to finish the proof of the proposition.

{\tt Step 1.} The first step for obtaining \eqref{der} is to estimate
$\left \| \frac{\partial^{r+p} \Delta_k^{-1}}{ \partial z^{r} \partial
w^{p}} \pi_{G'_3} \right \|$. Observe that
\begin{align*}
\left| \left( \frac{\partial^{r+p} \Delta_k^{-1}}{ \partial z^r \partial
w^p} \right)_{b,c} \right| & = \left| \frac{\partial^{r+p}
(\Delta_k^{-1})_{b,c}}{ \partial z^r \partial w^p} \right| = \left|
\frac{\partial^p}{\partial w^p} \, \frac{1}{w - 2i \theta_{\mu'}(b-d')}
\, \frac{\partial^r}{\partial z^r} \, \frac{\delta_{b,c}}{z - 2i
\theta_\mu(b-d')} \right| \\
& = \left| \frac{(-1)^p \, p!}{(w-2i \theta_{\mu'}(b-d'))^{p+1}} \,
\frac{(-1)^r \, r! \, \delta_{b,c}}{(z-2i \theta_\mu(b-d'))^{r+1}}
\right|\\ & \leq \frac{p! \, r! \, \delta_{b,c} }{|w - 2i
\theta_{\mu'}(b-d')|^{p+1} |z - 2i \theta_\mu(b-d')|^{r+1}},
\end{align*}
and recall from \eqref{est3} and \eqref{est4} that, for all $b \in
G'_3$,
\begin{equation} \label{rec}
\frac{1}{|z-2i \theta_\mu(b-d')|} \leq \frac{2}{|z|_R} \qquad \text{and}
\qquad \frac{1}{|w - 2i \theta_{\mu'}(b-d')|} \leq \frac{1}{\Lambda}.
\end{equation}
Then,
$$ \left| \left( \frac{\partial^{r+p} \Delta_k^{-1}}{ \partial z^r
\partial w^p} \right)_{b,c} \right| \leq  \frac{p! \, r! \, 2^{r+1} \,
\delta_{b,c} }{\Lambda^{p+1} |z|_R^{r+1}}, $$
and consequently,
\begin{align*}
\left[ \sup_{b \in G'_3} \sum_{c \in G'_3} + \sup_{c \in G'_3}
\sum_{b\in G'_3} \right] \, \left| \left( \frac{\partial^{r+p}
\Delta_k^{-1}}{ \partial z^r \partial w^p} \right)_{b,c} \right| & \leq
\frac{p! \, r! \, 2^{r+1} }{\Lambda^{p+1} |z|_R^{r+1}} \left[ \sup_{b
\in G'_3} \sum_{c \in G'_3} + \sup_{c \in G'_3} \sum_{b \in G'_3}
\right] \delta_{b,c} \\
& = \frac{p! \, r! \, 2^{r+2} }{\Lambda^{p+1} |z|_R^{r+1}}.
\end{align*}
Therefore, by Proposition \ref{p:ineq},
\begin{equation} \label{estD}
\left \| \frac{\partial^{r+p} \Delta_k^{-1}}{ \partial z^{r} \partial
w^{p}} \pi_{G'_3} \right \| \leq \frac{p! \, r! \, 2^{r+2}
}{\Lambda^{p+1}} \, \frac{1}{|z|_R^{r+1}}.
\end{equation}

{\tt Step 2.} We now estimate the second factor in \eqref{der}. Let us
first consider the case $j=1$, that is, $M^{(1)} = S$. Since
$S=(I-Y_{33})^{-1}$, the operator $S$ is clearly invertible. Thus, by
applying \eqref{rule} with $T=S^{-1}$, one can see that
$\frac{\partial^p S}{\partial w^p}$ is given by a finite linear
combination of terms of the form
\begin{equation} \label{prodS0}
\left [ \prod_{j=1}^{p} \, S \frac{\partial^{n_j} S^{-1}}{\partial
w^{n_j}} \right ] S,
\end{equation}
where $\sum_{j=1}^{p} n_j = p$. Hence, when we compute
$\frac{\partial^r}{\partial z^r} \frac{\partial^p S}{\partial w^p}$, the
derivative $\frac{\partial^r}{\partial z^r}$ acts either on $S$ or
$\frac{\partial^{n_j} S^{-1}}{\partial w^{n_j}}$. Similarly, using again
\eqref{rule} with $T=S^{-1}$, one can see that $\frac{\partial^r
S}{\partial z^r}$ is given by a finite linear combination of terms of
the form \eqref{prodS0}, with $p$ and $w$ replaced by $r$ and $z$,
respectively, and $\sum_{j=1}^{r} m_j = r$. Thus, we conclude that
$\frac{\partial^{r+p} S}{\partial z^r \partial w^p}$ is given by a
finite linear combination of terms of the form
\begin{equation} \label{prodS}
\left [ \prod_{j=1}^{r+p} \, S \frac{\partial^{m_j + n_j}
S^{-1}}{\partial z^{m_j} \partial w^{n_j}} \right ] S,
\end{equation}
where $\sum_{j=1}^{r+p} m_j = r$ and $\sum_{j=1}^{r+p} n_j = p$. Indeed,
observe that the general form of the terms \eqref{prodS} follows
directly from \eqref{rule} because that identity is also valid for mixed
derivatives.

Since $S = (I-Y_{33})^{-1}$ with $\| Y_{33} \| < 1/14$ and
\begin{equation} \label{Ybc2}
Y_{b,c} = \frac{-2i \theta_{\mu'} (\hat{A}(b-c)) \, z}{(w - 2i
\theta_{\mu'}(c-d')) (z - 2i \theta_\mu(c-d'))},
\end{equation}
we have
\begin{equation} \label{estS}
\| S \| = \| (I-Y_{33})^{-1} \| \leq \frac{1}{1-\| Y_{33} \|} \leq
\frac{14}{13}
\end{equation}
and
\begin{align*}
\left| \left( \frac{\partial^{j + l}}{\partial z^j \partial w^l} S^{-1}
\right)_{b,c} \right| & = \left| \frac{\partial^{j + l}}{\partial z^j
\partial w^l} Y_{b,c} \right| \\
& = \left| \frac{\partial^j}{\partial z^j}
\, \frac{-2i \theta_{\mu'}(\hat{A}(b-c)) \, z}{z -2i \theta_\mu(c-d')}
\, \frac{\partial^l}{\partial w^l} \, \frac{1}{ w - 2i
\theta_{\mu'}(c-d')} \right|.
\end{align*}
Furthermore,
\begin{alignat*}{2}
\frac{\partial^j}{\partial z^j} \, \frac{-2i \theta_{\mu'}(\hat{A}(b-c))
\, z}{z -2i \theta_\mu(c-d')} & = \frac{(-1)^{j-1} j! \, 2i
\theta_{\mu'}(\hat{A}(b-c)) \, 2 i \theta_{\nu}(c-d')}{(z - 2 i
\theta_{\nu}(c-d'))^{j+1}} & \qquad \text{for} \quad j \geq 1, \\
\frac{\partial^l}{\partial w^l} \, \frac{1}{ w - 2i \theta_{\mu'}(c-d')}
& = \frac{(-1)^l \, l!}{ ( w - 2i \theta_{\mu'}(c-d'))^{l+1}} & \qquad
\text{for} \quad l \geq 0.
\end{alignat*}
Recall from \eqref{est4} and \eqref{est6} that, for all $c \in G'$,
\begin{equation} \label{rec2}
\frac{ |c-d'| }{ |w - 2i \theta_{\mu'}(c-d')| } \leq \frac{ |c-d'| }{
|c-d'| - \eps} \leq 2.
\end{equation}
Then, using this and \eqref{rec}, for $j \geq 1$ and $l \geq 0$,
\begin{equation} \label{estYbc}
\begin{aligned}
\left| \left( \frac{\partial^{j + l}}{\partial z^j \partial w^l} S^{-1}
\right)_{b,c} \right| & \leq \frac{j! \, l! \, |\hat{A}(b-c)|}{ |z -2i
\theta_\mu(c-d')|^{j+1} |w - 2i \theta_{\mu'}(c-d')|^l} \, \frac{ |c-d'|
}{ |w - 2i \theta_{\mu'}(c-d')| } \\
& \leq \frac{2^{j+2} j! \, l! \, |\hat{A}(b-c)|}{ \Lambda^l |z|_R^{j+1}
}, 
\end{aligned}
\end{equation}
while for $j=0$ and $l \geq 0$,
\begin{equation} \label{estYbc2}
\left| \left( \frac{\partial^{j + l}}{\partial z^j \partial w^l} S^{-1}
\right)_{b,c} \right| \leq \frac{l! \, | \hat{A}(b-c) | \, |z|}{ |z -2i
\theta_\mu(c-d')| \, |w - 2i \theta_{\mu'}(c-d')|^{l+1}} \leq \frac{2 \,
l! \, | \hat{A}(b-c) |}{\Lambda^{l+1} }.
\end{equation}
Consequently,
\begin{align*}
& \left[ \sup_{b \in G'_3} \sum_{c \in G'_3} + \sup_{c \in G'_3} \sum_{b
\in G'_3} \right] \left| \left( \frac{\partial^{j + l}}{\partial z^j
\partial w^l} S^{-1} \right)_{b,c} \right| \\
& \qquad \qquad \leq \left( 1 - \delta_{0,j} + \frac{|z|_R}{2 \Lambda}
\delta_{0,j} \right) \frac{2^{j+2} j! \, l!}{ \Lambda^l |z|_R^{j+1} }
\left[ \sup_{b \in G'_3} \sum_{c \in G'_3} + \sup_{c \in G'_3} \sum_{b
\in G'_3} \right] | \hat{A}(b-c) | \\
& \qquad \qquad \leq \left( 1 - \delta_{0,j} + \frac{|z|_R}{2 \Lambda}
\delta_{0,j} \right) \frac{2^{j+3} j! \, l!}{ \Lambda^l |z|_R^{j+1} } \,
\| \hat{A} \|_{l^1}.
\end{align*}
Therefore, by Proposition \eqref{p:ineq},
\begin{equation} \label{Sinv}
\left \| \frac{\partial^{j + l}}{\partial z^j \partial w^l} S^{-1}
\right \| \leq \left( 1 - \delta_{0,j} + \frac{|z|_R}{2 \Lambda}
\delta_{0,j} \right) \frac{2^{j+3} j! \, l!}{ \Lambda^l |z|_R^{j+1} } \,
\| \hat{A} \|_{l^1}.
\end{equation}
Thus, for $r \geq 1$, in view of \eqref{prodS} where $\sum_{j=1}^{r+p}
m_j = r$,
\begin{align*}
& \left \| \frac{\partial^{r+p}}{\partial z^r \partial w^p} S \right \|
\leq C_{r,p} \left [ \prod_{j=1}^{r+p} \, \| S \| \, \left\|
\frac{\partial^{m_j + n_j}}{\partial z^{m_j} \partial w^{n_j}} S^{-1}
\right \| \, \right ] \| S \| \\
& \leq C_{r,p} \left [ \prod_{j=1}^{r+p} \, C_{\Lambda,A} \,
\frac{2^{m_j+3} m_j! \, n_j!}{\Lambda^{n_j}} \| \hat{A} \|_{l^1} \right
] C_{\Lambda,A} \, \prod_{j=1}^{r+p} \left( 1 - \delta_{0,m_j} +
\frac{|z|_R}{2 \Lambda} \delta_{0,m_j} \right) \frac{1}{|z|_R^{m_j+1}}
\\
& \leq C_{\Lambda, A, r, p} \, \frac{1}{|z|_R^{r+1}},
\end{align*}
since $m_j \geq 1$ for at least one $1 \leq j \leq r+p$. Similarly, if
$r = 0$ then
$$ \left \| \frac{\partial^{r+p}}{\partial z^r \partial w^p} S \right \|
\leq C_{\Lambda, A, r, p}. $$
Hence, in view of \eqref{estD},
\begin{align*}
& \sup_{p \leq m} \sup_{r \leq n} \left \| \frac{\partial^{n-r+m-p}
\Delta_k^{-1}}{ \partial z^{n-r} \partial w^{m-p}} \right \| \left\|
\frac{\partial^{r+p} M^{(1)}}{\partial z^r \partial w^p} \right \| \\
& \leq \sup_{p \leq m} \sup_{r \leq n} \frac{(m-p)! \, (n-r)! \,
2^{n-r+2} }{\Lambda^{m-p+1} |z|_R^{n-r+1}} \, C_{\Lambda, A, r,p} \|
\hat{A} \|_{l^1} \left( 1 - \delta_{0,r} + \frac{|z|_R}{2 \Lambda}
\delta_{0,r} \right) \frac{1}{|z|_R^{r+1}} \\
& \leq C_{\Lambda, A, n,m} \, \frac{1}{|z|_R^{n+1}}.
\end{align*}
This proves \eqref{der} for $j=1$.

{\tt Step 3.} We now estimate the constant $C_{1;\Lambda, A,i, j}$ for
$i,j \in \{0,1\}$. First observe that
$$ \left| \frac{\partial w}{\partial k_j} \right| = | \delta_{1,j} +
i(-1)^\nu \delta_{2,j} | = 1 \qquad \text{and} \qquad \left|
\frac{\partial z}{\partial k_j} \right| = | \delta_{1,j} - i(-1)^\nu
\delta_{2,j} | = 1. $$
Thus, in view of \eqref{estS} and \eqref{Sinv}, since $|z| \ge |v| > R
\ge 2\Lambda$,
\begin{align*}
\left \| \frac{\partial S}{\partial k_j} \right \| & = \left \| -S
\frac{\partial S^{-1}}{\partial k_j} S \right \| = \left \| - S \left(
\frac{\partial w}{\partial k_j} \frac{\partial S^{-1}}{\partial w} +
\frac{\partial z}{\partial k_j} \frac{\partial S^{-1}}{\partial z}
\right) S \right \| \\
& \leq \| S \|^2 \left( \left\| \frac{\partial S^{-1}}{\partial w}
\right \| + \left \| \frac{\partial S^{-1}}{\partial z} \right \|
\right) \leq \left( \frac{3}{2} \right)^2 \left( \frac{2^4 \| \hat{A}
\|_{l^1}}{|z|_R^2} + \frac{2^2 \| \hat{A} \|_{l^1}}{\Lambda^2} \right)
\\
& \le \frac{18 \| \hat{A} \|_{l^1}}{\Lambda^2}.
\end{align*}
Similarly,
\begin{align*}
\frac{\partial^2 S}{\partial k_i \partial k_j} & = - \frac{\partial
S}{\partial k_i} \left( \frac{\partial w}{\partial k_j} \frac{\partial
S^{-1}}{\partial w} + \frac{\partial z}{\partial k_j} \frac{\partial
S^{-1}}{\partial z} \right) S - S \left( \frac{\partial w}{\partial k_j} 
\frac{\partial S^{-1}}{\partial w} + \frac{\partial z}{\partial k_j} 
\frac{\partial S^{-1}}{\partial z} \right) 
\frac{\partial S}{\partial k_i} \\
& \quad - S \left( \frac{\partial w}{\partial k_j} \left( \frac{\partial
w}{\partial k_i} \frac{\partial^2 S^{-1}}{\partial w^2} + \frac{\partial
z}{\partial k_i} \frac{\partial^2 S^{-1}}{\partial z \partial w} \right)
+ \frac{\partial z}{\partial k_j} \left( \frac{\partial w}{\partial k_i}
\frac{\partial^2 S^{-1}}{\partial w \partial z} + \frac{\partial
z}{\partial k_i} \frac{\partial^2 S^{-1}}{\partial z^2} \right) \right)
S,
\end{align*}
so that, using the above inequality as well,
\begin{align*}
\left \| \frac{\partial^2 S}{\partial k_i \partial k_j} \right \| & \leq
2 \| S \| \left \| \frac{\partial S}{\partial k_i} \right \| \left(
\left\| \frac{\partial S^{-1}}{\partial w} \right \| + \left \|
\frac{\partial S^{-1}}{\partial z} \right \| \right) \\ 
& \qquad + \| S \|^2 \left( \left \| \frac{\partial^2 S^{-1}}{\partial
w^2} \right \| + 2 \left \| \frac{\partial^2 S^{-1}}{\partial z \partial
w} \right \| + \left \| \frac{\partial^2 S^{-1}}{\partial z^2} \right \|
\right) \\
& \leq 2 \, \frac{3}{2} \, \frac{18 \| \hat{A} \|_{l^1}}{\Lambda^2} \,
\frac{8 \| \hat{A} \|_{l^1}}{\Lambda^2} + \left( \frac{3}{2} \right)^2
\left( \frac{2^3 \| \hat{A} \|_{l^1}}{\Lambda^3} + \frac{2^5 \| \hat{A}
\|_{l^1}}{\Lambda |z|_R^2} + \frac{2^6 \| \hat{A} \|_{l^1}}{|z|_R^3}
\right) \\
& \leq \frac{432}{\Lambda^4} \| \hat{A} \|_{l^1}^2 +
\frac{54}{\Lambda^3} \| \hat{A} \|_{l^1} \leq \frac{55 \| \hat{A}
\|_{l^1}}{\Lambda^3} \left( \frac{8 \| \hat{A} \|_{l^1}}{\Lambda} + 1
\right).
\end{align*}
Furthermore, by \eqref{estD},
\begin{align*}
\left \| \frac{\partial \Delta_k^{-1}}{\partial k_j} \right \| & \leq
\left \| \frac{\partial \Delta_k^{-1}}{\partial w} \right \| + \left \|
\frac{\partial \Delta_k^{-1}}{\partial z} \right \| \leq
\frac{2^2}{\Lambda^2 |z|_R} + \frac{2^3}{\Lambda |z|_R^2} \leq
\frac{8}{\Lambda^2 |z|_R}
\end{align*}
and
\begin{align*}
\left \| \frac{\partial^2 \Delta_k^{-1}}{\partial k_i \partial k_j}
\right \| & \leq \left \| \frac{\partial^2 \Delta_k^{-1}}{\partial w^2}
\right \| + 2 \left \| \frac{\partial^2 \Delta_k^{-1}}{\partial z
\partial w} \right \| + \left \| \frac{\partial^2
\Delta_k^{-1}}{\partial z^2} \right \| \\
& \leq \frac{2^3}{\Lambda^3 |z|_R} + \frac{2^4}{\Lambda^2 |z|_R^2} +
\frac{2^6}{\Lambda |z|_R^3} < \frac{5 \cdot 2^3}{\Lambda^3} \,
\frac{1}{|z|_R}.
\end{align*}
Hence, since $\| \hat{A} \|_{l_1} < 2\eps/63$ and $\eps<\Lambda/6$,
\begin{align*}
\left \| \frac{\partial}{\partial k_j} \Delta_k^{-1} S \right \| & \leq
\left \| \frac{\partial \Delta_k^{-1}}{\partial k_j} \right \| \| S \| +
\| \Delta_k^{-1} \| \left \| \frac{\partial S}{\partial k_j} \right \|\\
& \leq \frac{8}{\Lambda^2 |z|_R} \, \frac{3}{2} + \frac{2}{\Lambda |z|_R}
\, \frac{18 \| \hat{A} \|_{l^1}}{\Lambda^2} \leq \frac{13}{\Lambda^2} \,
\frac{1}{|z|_R}
\end{align*}
and
\begin{align*}
& \left \| \frac{\partial^2}{\partial k_i \partial k_j} \Delta_k^{-1} S
\right \| \\
& \leq \left \| \frac{\partial^2 \Delta_k^{-1}}{\partial k_i
\partial k_j} \right \| \| S \| + \left \| \frac{\partial
\Delta_k^{-1}}{\partial k_j} \right \| \, \left \| \frac{\partial
S}{\partial k_i} \right \| + \left \| \frac{\partial
\Delta_k^{-1}}{\partial k_i} \right \| \, \left \| \frac{\partial
S}{\partial k_j} \right \| + \| \Delta_k^{-1} \| \left \|
\frac{\partial^2 S}{\partial k_i \partial k_j} \right \| \\
& \leq \frac{1}{|z|_R} \Bigg( \frac{5 \cdot 2^3}{\Lambda^3} \,
\frac{3}{2} + 2 \, \frac{8}{\Lambda^2} \, \frac{18 \| \hat{A}
\|_{l^1}}{\Lambda^2} + \frac{2}{\Lambda} \, \frac{55 \| \hat{A}
\|_{l^1}}{\Lambda^3} \Big( \frac{8 \| \hat{A} \|_{l^1}}{\Lambda} + 1
\Big) \Bigg) < \frac{65}{\Lambda^3} \, \frac{1}{|z|_R}.
\end{align*}
Therefore,
$$ C_{1;\Lambda, A, 1, 0} \leq \frac{13}{\Lambda^2}, \qquad
C_{1;\Lambda, A, 0, 1} \leq \frac{13}{\Lambda^2} \qquad \text{and}
\qquad C_{1;\Lambda, A, 1, 1} \leq \frac{65}{\Lambda^3}, $$
as was to be shown.

{\tt Step 4.} To prove \eqref{der} for $j=2$ we need to bound $\Big \|
\frac{\partial^{r+p} M^{(2)}}{\partial z^r \partial w^p} \Big \| =
\left\| \frac{\partial^{r+p} W }{\partial z^r \partial w^p} \right \|$.
Recall from \eqref{WS} that
$$ W = \sum_{j=1}^\infty W_j = \sum_{j=1}^\infty \sum_{m=1}^j
(Y_{33})^{m-1} X_{33} (Y_{33})^{j-m}, $$
where $Y_{b,c}$ is given above by \eqref{Ybc2} and $\| X_{33} \| \leq
C/|z| < 1/3$ with
$$ X_{b,c} = \frac{(c-d') \cdot \hat{A}(b-c) - \hat{q}(b-c) -2i
\theta_\mu(\hat{A}(b-c)) \, w}{(w - 2i \theta_{\mu'}(c-d')) (z - 2i
\theta_\mu(c-d'))}. $$
First observe that
$$ \frac{\partial^{r+p}}{\partial z^r w^p} \, (Y_{33})^{m-1} X_{33}
(Y_{33})^{j-m} $$
is given by a sum of $j^{r+p}$ terms of the form
$$ \frac{\partial^{l_1 + n_1} Y_{33}}{\partial z^{l_1} \partial w^{n_1}}
\; \cdots \; \frac{\partial^{l_{m-1} + n_{m-1}} Y_{33}}{\partial
z^{l_{m-1}} \partial w^{n_{m-1}}} \; \frac{\partial^{l_m + n_m}
X_{33}}{\partial z^{l_m} \partial w^{n_m}} \; \frac{\partial^{l_{m+1} +
n_{m+1}} Y_{33}}{\partial z^{l_{m+1}} \partial w^{n_{m+1}}} \; \cdots \;
\frac{\partial^{l_j + n_j} Y_{33}}{\partial z^{l_j} \partial w^{n_j}},
$$
where there are $j$ factors ordered as in the product $(Y_{33})^{m-1}
X_{33} (Y_{33})^{j-m}$. Furthermore, for each term in the sum we have
$\sum_{i=1}^j l_i = r$ and $\sum_{i=1}^j n_i = p$. Thus,
\begin{align}
& \left \| \frac{\partial^{r+p}}{\partial z^r w^p} \, W \right \| =
\left \| \sum_{j=1}^\infty \frac{\partial^{r+p}}{\partial z^r w^p} \,
W_j \right \| \label{derW1} \\
& = \left \| \sum_{j=1}^\infty \sum_{m=1}^j
\frac{\partial^{r+p}}{\partial z^r w^p} \, (Y_{33})^{m-1} X_{33}
(Y_{33})^{j-m} \right \| \notag \\
& \leq \sum_{j=1}^\infty \sum_{m=1}^j \left \|
\frac{\partial^{r+p}}{\partial z^r w^p} \, (Y_{33})^{m-1} X_{33}
(Y_{33})^{j-m} \right \| \notag \\
& \leq \sum_{j=1}^\infty j^{r+p} \sum_{m=1}^j \sup_{\mathcal{I}} \left
\| \frac{\partial^{l_1 + n_1} Y_{33}}{\partial z^{l_1} \partial w^{n_1}}
\cdots \frac{\partial^{l_m + n_m} X_{33}}{\partial z^{l_m} \partial
w^{n_m}} \cdots \frac{\partial^{l_j + n_j} Y_{33}}{\partial z^{l_j}
\partial w^{n_j}} \right \| \notag \\ 
& \leq \sum_{j=1}^\infty j^{r+p} \sum_{m=1}^j \sup_{\mathcal{I}} \left
\| \frac{\partial^{l_1 + n_1} Y_{33}}{\partial z^{l_1} \partial w^{n_1}}
\right \| \cdots \left \| \frac{\partial^{l_m + n_m} X_{33}}{\partial
z^{l_m} \partial w^{n_m}} \right \| \cdots \left \| \frac{\partial^{l_j
+ n_j} Y_{33}}{\partial z^{l_j} \partial w^{n_j}} \right \|,
\label{derW2}
\end{align}
where
\begin{equation} \label{I}
\mathcal{I} \coloneqq \left \{ (l_i, n_i) \; \Bigg| \; l_i \leq r \,
\text{ and } \, n_i \leq p \, \text{ for } \, 1 \leq i \leq j \text{
with } \sum_{i=1}^j l_i = r \, \text{ and } \, \sum_{i=1}^j n_i = p
\right \}.
\end{equation}
Note, we can differentiate the series \eqref{derW1} term-by-term because
the sum $\sum_{j=1}^\infty W_j$ converges uniformly and the sum
$\sum_{m=1}^j$ is finite. We next estimate the factors in \eqref{derW2}.

Combining \eqref{estYbc} and \eqref{estYbc2} we have
\begin{equation} \label{dYbc}
\left| \frac{\partial^{l_i + n_i}}{\partial z^{l_i} \partial w^{n_i}}
Y_{b,c} \right| \leq \left( 1 - \delta_{0,l_i} + \frac{|z|_R}{2 \Lambda}
\delta_{0,l_i} \right) \frac{2^{l_i+2} l_i! \, n_i!}{ \Lambda^{n_i}
|z|_R^{l_i+1} } \, |\hat{A}(b-c)|.
\end{equation}
Furthermore, using \eqref{rec} and \eqref{rec2},
\begin{equation} \label{dXbc}
\begin{split}
& \left| \frac{\partial^{l_i+n_i}}{\partial z^{l_i} \partial w^{n_i}}
X_{b,c} \right| \\
& = \left| \frac{\partial^{l_i}}{\partial z^{l_i}} \, \frac{1}{z -2i
\theta_\mu(c-d')} \, \frac{\partial^{n_i}}{\partial w^{n_i}} \, \frac{
(c-d') \cdot \hat{A}(b-c) - \hat{q}(b-c) -2i \theta_\mu (\hat{A}(b-c))
\, w}{ w - 2i \theta_{\mu'}(c-d')} \right| \\
& = \left| \frac{(-1)^{l_i} \, l_i! \, (-1)^{n_i} \, n_i! \, ( 2
\theta_\mu(\hat{A}(b-c)) \, 2 \theta_{\mu'}(c-d') - (c-d') \cdot
\hat{A}(b-c) - \hat{q}(b-c)) }{(z -2i \theta_\mu(c-d'))^{l_i+1} ( w - 2i
\theta_{\mu'}(c-d'))^{n_i+1}} \right| \\
& \leq \frac{l_i! \, n_i! \, ( 2 |\hat{A}(b-c)| \, |c-d'| +
|\hat{q}(b-c)|) }{ |z -2i \theta_\mu(c-d')|^{l_i+1} |w - 2i
\theta_{\mu'}(c-d')|^{n_i+1}} \\
& \leq \frac{2^{l_i+1} l_i! \, n_i!}{
\Lambda^{n_i} |z|_R^{l_i+1} } \, \frac{2 |\hat{A}(b-c)| \, |c-d'| +
|\hat{q}(b-c)|}{ |w - 2i \theta_{\mu'}(c-d')| } \\
& \leq \frac{2^{l_i+1} l_i! \, n_i!}{ \Lambda^{n_i} |z|_R^{l_i+1} } \,
\Big( 4 |\hat{A}(b-c)| + \frac{1}{\Lambda} |\hat{q}(b-c)| \Big).
\end{split}
\end{equation}
Hence,
\begin{align*}
& \left[ \sup_{b \in G'_3} \sum_{c \in G'_3} + \sup_{c \in G'_3} \sum_{b
\in G'_3} \right] \left| \frac{\partial^{l_i + n_i}}{\partial z^{l_i}
\partial w^{n_i}} Y_{b,c} \right| \\ 
& \qquad \leq \left( 1 - \delta_{0,l_i} + \frac{|z|_R}{2 \Lambda}
\delta_{0,l_i} \right) \frac{2^{l_i+2} l_i! \, n_i!}{ \Lambda^{n_i}
|z|_R^{l_i+1} } \left[ \sup_{b \in G'_3} \sum_{c \in G'_3} + \sup_{c \in
G'_3} \sum_{b \in G'_3} \right] |\hat{A}(b-c)| \\
& \qquad \leq \left( 1 - \delta_{0,l_i} + \frac{|z|_R}{2 \Lambda}
\delta_{0,l_i} \right) \frac{2^{l_i+3} l_i! \, n_i!}{ \Lambda^{n_i}
|z|_R^{l_i+1} } \| \hat{A} \|_{l^1}
\end{align*}
and similarly
$$ \left[ \sup_{b \in G'_3} \sum_{c \in G'_3} + \sup_{c \in G'_3}
\sum_{b \in G'_3} \right] \left| \frac{\partial^{l_i+n_i}}{\partial
z^{l_i} \partial w^{n_i}} X_{b,c} \right| \leq \frac{2^{l_i+2} l_i! \,
n_i!}{ \Lambda^{n_i} |z|_R^{l_i+1} } \left( 4 \| \hat{A} \|_{l^1} +
\frac{ \| \hat{q} \|_{l^1} }{\Lambda} \right). $$
Thus, by Proposition \eqref{p:ineq}, since $|z| \ge |v| > R \ge 2
\Lambda$,
\begin{equation} \label{dY}
\begin{split}
\left \| \frac{\partial^{l_i + n_i}}{\partial z^{l_i} \partial w^{n_i}}
Y_{33} \right \| & \leq \left( 1 - \delta_{0,l_i} + \frac{|z|_R}{2
\Lambda} \delta_{0,l_i} \right) \frac{2^{l_i+3} l_i! \, n_i!}{
\Lambda^{n_i} |z|_R^{l_i+1} } \| \hat{A} \|_{l^1} \\ 
& \leq \left( \frac{1}{|z|_R} + \frac{1}{2 \Lambda} \right)
\frac{2^{l_i+3} l_i! \, n_i!}{ \Lambda^{n_i} |z|_R^{l_i} } \| \hat{A}
\|_{l^1} \leq \frac{2^{l_i+3} l_i! \, n_i!}{ \Lambda^{n_i+1}
|z|_R^{l_i}} \| \hat{A} \|_{l^1}
\end{split}
\end{equation}
and
\begin{equation} \label{dX}
\begin{aligned}
\left \| \frac{\partial^{l_i+n_i}}{\partial z^{l_i} \partial w^{n_i}}
X_{33} \right \| & \leq \frac{2^{l_i+2} l_i! \, n_i!}{ \Lambda^{n_i}
|z|_R^{l_i+1} } \left( 4 \| \hat{A} \|_{l^1} + \frac{ \| \hat{q}
\|_{l^1} }{\Lambda} \right) \\
& = \left( 2 \Lambda + \frac{ \| \hat{q} \|_{l^1}}{2 \| \hat{A} \|_{l^1}} 
\right) \frac{1}{|z|_R} \, \frac{2^{l_i+3} l_i! \, n_i!}{\Lambda^{n_i+1} 
|z|_R^{l_i} } \| \hat{A} \|_{l^1}.
\end{aligned}
\end{equation}
Applying these estimates to \eqref{derW2} and recalling that
$\sum_{i=1}^j l_i = r$ and $\sum_{i=1}^j n_i = p$ we have 
\begin{align*}
& \left \| \frac{\partial^{r+p}}{\partial z^r w^p} W \right \| \\
& \leq
\sum_{j=1}^\infty j^{r+p} \sum_{m=1}^j \sup_{\mathcal{I}} \left \|
\frac{\partial^{l_1 + n_1} Y_{33}}{\partial z^{l_1} \partial w^{n_1}}
\right \| \cdots \left \| \frac{\partial^{l_m + n_m} X_{33}}{\partial
z^{l_m} \partial w^{n_m}} \right \| \cdots \left \| \frac{\partial^{l_j
+ n_j} Y_{33}}{\partial z^{l_j} \partial w^{n_j}} \right \| \\
& \leq \sum_{j=1}^\infty j^{r+p} \sum_{m=1}^j \sup_{\mathcal{I}} \left
\{ \left( 2 \Lambda + \frac{ \| \hat{q} \|_{l^1}}{2 \| \hat{A} \|_{l^1}}
\right) \frac{1}{|z|_R} \prod_{i=1}^j \frac{2^{l_i+3} l_i! \, n_i!}{
\Lambda^{n_i+1} |z|_R^{l_i} } \| \hat{A} \|_{l^1} \right \} \\
& = \left( 2 \Lambda + \frac{ \| \hat{q} \|_{l^1}}{2 \| \hat{A}
\|_{l^1}} \right) \frac{1}{|z|_R} \frac{2^r}{\Lambda^p |z|_R^r}
\sum_{j=1}^\infty j^{r+p} \left( \frac{ 8 \| \hat{A} \|_{l^1}}{\Lambda}
\right)^j \sup_{\mathcal{I}} \left \{ \prod_{i=1}^j  l_i! \prod_{m=1}^j
n_m! \right \} \sum_{m=1}^j 1 \\
& \leq \left( 2 \Lambda + \frac{ \| \hat{q} \|_{l^1}}{2 \| \hat{A}
\|_{l^1}} \right) \frac{2^r r! p!}{\Lambda^p |z|_R^{r+1}}
\sum_{j=1}^\infty j^{r+p+1} \left( \frac{1}{21} \right)^j \leq
C'_{\Lambda, A, q, r, p} \, \frac{1}{|z|_R^{r+1}}. 
\end{align*}
This is the inequality we needed to prove \eqref{der} for $j=2$. In
fact, using \eqref{estD} we obtain
\begin{align*}
\sup_{p \leq m} \sup_{r \leq n} \left \| \frac{\partial^{n-r+m-p}
\Delta_k^{-1}}{ \partial z^{n-r} \partial w^{m-p}} \right \| \left \|
\frac{\partial^{r+p} M^{(2)}}{\partial z^r \partial w^p} \right \| &
\leq \sup_{p \leq m} \sup_{r \leq n} \frac{(m-p)! \, (n-r)! \, 2^{n-r+2}
}{\Lambda^{m-p+1} |z|_R^{n-r+1}} \; \frac{C'_{\Lambda, A, q, r,
p}}{|z|_R^{r+1}} \\
& \leq C_{\Lambda, A, q, m, n} \, \frac{1}{|z|_R^{n+2}}.
\end{align*}

{\tt Step 5.} To prove \eqref{der} for $j=3$ we need to estimate $\Big
\| \frac{\partial^{r+p} M^{(3)}}{\partial z^r \partial w^p} \Big \| =
\left\| \frac{\partial^{r+p} Z }{\partial z^r \partial w^p} \right \|$,
where
$$ Z = \sum_{j=2}^\infty Z_j = \sum_{j=2}^\infty (X_{33}+Y_{33})^j - W_j
- Y_{33}^j. $$
First observe that
$$ \frac{\partial^{r+p}}{\partial z^r \partial w^p} Z_j =
\frac{\partial^{r+p}}{\partial z^r \partial w^p} ( (X_{33}+Y_{33})^j -
W_j - Y_{33}^j )$$
is given by a sum of $(2^j-j-1) \cdot j^{r+p}$ terms of the form
\begin{equation} \label{XXY}
\underbrace{\frac{\partial^{l_1 + n_1} Y_{33}}{\partial z^{l_1} \partial
w^{n_1}} \; \cdots \; \frac{\partial^{l_m + n_m} X_{33}}{\partial
z^{l_m} \partial w^{n_m}} \; \cdots \; \frac{\partial^{l_j + n_j}
Y_{33}}{\partial z^{l_j} \partial w^{n_j}}}_{j \text{ factors}},
\end{equation}
where there are $j-2$ factors involving $X_{33}$ or $Y_{33}$ and two
factors containing $X_{33}$. Furthermore, for each term in the sum we
have $\sum_{i=1}^j l_i = r$ and $\sum_{i=1}^j n_i = p$. Thus,
$$ \left \| \frac{\partial^{r+p}}{\partial z^r \partial w^p} Z_j \right
\| \leq (2^j-j-1) \, j^{r+p} \sup_{\mathcal{I}} \left \|
\frac{\partial^{l_1 + n_1} Y_{33}}{\partial z^{l_1} \partial w^{n_1}}
\right \| \cdots \left \| \frac{\partial^{l_m + n_m} X_{33}}{\partial
z^{l_m} \partial w^{n_m}} \right \| \cdots \left \| \frac{\partial^{l_j
+ n_j} Y_{33}}{\partial z^{l_j} \partial w^{n_j}} \right \|, $$
where the set $\mathcal{I}$ is given above by \eqref{I}. Now observe
that, the estimate for the derivatives of $X_{33}$ in \eqref{dX} is
better then the estimate for the derivatives of $Y_{33}$ in \eqref{dY}
because the former has an extra factor $C_{\Lambda,A,q}/|z|_R < 1$.
Since the product \eqref{XXY} has at least two factors containing
$X_{33}$, we can estimate any of these products by considering the worst
case. This happens when there are exactly two factors involving
$X_{33}$. Hence, by proceeding in this way, for each $j \geq 2$ we have
\begin{align*}
\left \| \frac{\partial^{r+p}}{\partial z^r \partial w^p} Z_j \right \|
& \leq (2^j-j-1) \, j^{r+p} \sup_{\mathcal{I}} \left \{ \left( 2\Lambda
+ \frac{ \|\hat{q} \|_{l^1}}{ 2 \| \hat{A} \|_{l^1}} \right)^2
\frac{1}{|z|_R^2} \prod_{i=1}^j \frac{2^{l_i+3} l_i! n_i!}{
\Lambda^{n_i+1} |z|_R^{l_i}} \| \hat{A} \|_{l^1} \right\}\\
& \leq 2^j j^{r+p} \left( 2\Lambda + \frac{ \|\hat{q} \|_{l^1}}{ 2 \|
\hat{A} \|_{l^1}} \right)^2 \frac{1}{|z|_R^2} \, \frac{2^r r!
p!}{\Lambda^p |z|_R^r} \left( \frac{8 \| \hat{A} \|_{l^1}}{\Lambda}
\right)^j \\
& \leq C'_{\Lambda, A, q, r,p} \, j^{r+p} \left( \frac{2}{21} \right)^j
\frac{1}{|z|_R^{r+2}},
\end{align*}
since $\| A \|_{l^1} \leq 2\eps/63$ and $\eps < \Lambda/6$. Thus,
$$ \left \| \frac{\partial^{r+p}}{\partial z^r \partial w^p} Z \right \|
\leq \sum_{j=2}^\infty \left \| \frac{\partial^{r+p}}{\partial z^r
\partial w^p} Z_j \right \| \leq \frac{C'_{\Lambda, A, q,
r,p}}{|z|_R^{r+2}} \sum_{j=2}^\infty j^{r+p} \left( \frac{2}{21}
\right)^j \leq \frac{C_{\Lambda, A, q, r,p}}{|z|_R^{r+2}}. $$
Therefore, recalling \eqref{estD},
\begin{align*}
\sup_{p \leq m} \sup_{r \leq n} \left \| \frac{\partial^{n-r+m-p}
\Delta_k^{-1}}{ \partial z^{n-r} \partial w^{m-p}} \right \| \left \|
\frac{\partial^{r+p} M^{(3)}}{\partial z^r \partial w^p} \right \| &
\leq \sup_{p \leq m} \sup_{r \leq n} \frac{(m-p)! \, (n-r)! \, 2^{n-r+2}
}{\Lambda^{m-p+1} |z|_R^{n-r+1}} \; \frac{C'_{\Lambda, A, q,
r,p}}{|z|_R^{r+2}} \\
& \leq C_{\Lambda, A, q, m, n} \, \frac{1}{|z|_R^{n+3}}.
\end{align*}
This is the desired inequality for $j=3$. The proof of the proposition
is complete.
\end{proof}


We can now prove Lemma \ref{l:der2}. We first prove it for $1 \le j \le
2$ and then for $j=3$ separately.

\begin{proof}[Proof of Lemma \ref{l:der2} for $1 \le j \le 2$]
Define the $|B| \times |C|$ matrices
$$ \mathcal{F}_{BC} \coloneqq \left[ f(b-c) \right]_{b \in B, c \in C}
\qquad \text{and} \qquad \mathcal{G}_{BC} \coloneqq [ g(b-c) ]_{b \in b,
c \in C}, $$
and write $w = w_{\mu,d'}$, $z = z_{\mu,d'}$ and $|z|_R = 2|z|-R$. First
observe that, for $1 \le j \le 2$, the functions
$$ \big[ \alpha_{\mu,d'}^{(j)}(k) \big]_{d' \in G} = \left[ \sum_{b,c
\in G'_1} \frac{f(d'-b) M^{(j)}_{b,c} g(c-d')}{(w - 2i
\theta_{\mu'}(b-d')) (z - 2i \theta_{\mu}(b-d'))} \right]_{d' \in G} $$
are the diagonal entries of the matrix $\F_{GG'_1} \Delta_k^{-1} M^{(j)}
\,\G_{G'_1 G}$. Thus, similarly as in the proof of Lemma \ref{l:der}, by
Proposition \ref{p:dSWZ}, for $1 \le j \le 2$,
$$ \Bigg| \frac{\partial^{n+m}}{\partial k_1^n \partial k_2^m}
\alpha_{\mu,d'}^{(j)}(k) \Bigg | \leq \| \F_{GG'_1} \| \left \|
\frac{\partial^n}{\partial k_1^n} \frac{\partial^m}{\partial k_2^m}
\Delta_k^{-1} M^{(j)} \right \| \| \G_{G'_1 G} \| \leq
\frac{C_j}{|z|_R^j}, $$
where $C_1 = C_{1;\Lambda,A,n,m,f,g}$ and $C_2 =
C_{2;\Lambda,A,q,n,m,f,g}$ are constants. Furthermore,
\begin{gather*}
C_{1;\Lambda,A,1,0,f,g} \leq \frac{13}{\Lambda^2} \| f \|_{l^1} \| g
\|_{l^1}, \qquad C_{1;\Lambda,A,0,1,f,g} \leq \frac{13}{\Lambda^2} \| f
\|_{l^1} \| g \|_{l^1} \\
\text{and} \qquad C_{1;\Lambda,A,1,1,f,g} \leq \frac{65}{\Lambda^3} \| f
\|_{l^1} \| g \|_{l^1}.
\end{gather*}
This proves the lemma for $1 \le j \le 2$.
\end{proof}


\begin{proof}[Proof of Lemma \ref{l:der2} for $j=3$]
We need to estimate
$$ \frac{\partial^{n+m}}{\partial k_1^n \partial k_2^m}
\alpha^{(3)}_{\mu,d'}(k) = \sum_{j=1}^4 \frac{\partial^{n+m}}{\partial
k_1^n \partial k_2^m} \mathcal{R}_j(k), $$
where $\mathcal{R}_1, \dots, \mathcal{R}_4$ are given by \eqref{RR1},
\eqref{RR2}, \eqref{RR3} and \eqref{RR4}, respectively.

{\tt Step 1.} We begin with the terms involving $\mathcal{R}_1$ and
$\mathcal{R}_2$, which are easier. We follow the same notation as above.
First observe that, similarly as in the proof of Lemma \ref{l:der},
since $\Delta_k^{-1} R_{G'G'}^{-1} = H_k^{-1}$ on $L^2_{G'}$, we have
\begin{align*}
\left| \frac{\partial^{n+m}}{\partial k_1^n \partial k_2^m}
\mathcal{R}_1(k) \right| & = \left \| \mathcal{F}_{\{d'\} G'_1}
\frac{\partial^{n+m} H_k^{-1}}{\partial k_1^n \partial k_2^m} \;
\mathcal{G}_{G'_2 \{d'\}} \right\| \leq \| \mathcal{F}_{\{d'\} G'_1} \|
\left\| \frac{\partial^{n+m} H_k^{-1}}{\partial k_1^n \partial k_2^m}
\right \| \| \mathcal{G}_{G'_2 \{d'\}} \|, \\
\left| \frac{\partial^{n+m}}{\partial k_1^n \partial k_2^m}
\mathcal{R}_2(k) \right| & = \left \| \mathcal{F}_{\{d'\} G'_2}
\frac{\partial^{n+m} H_k^{-1}}{\partial k_1^n \partial k_2^m} \;
\mathcal{G}_{G' \{d'\}} \right\| \leq \| \mathcal{F}_{\{d'\} G'_2} \|
\left\| \frac{\partial^{n+m} H_k^{-1}}{\partial k_1^n \partial k_2^m}
\right \| \| \mathcal{G}_{G' \{d'\}} \|.
\end{align*}
Furthermore, we have already proved that $\| \mathcal{F}_{\{d'\} G'_1}
\| \leq \| f \|_{l^1}$ and $\| \mathcal{G}_{G' \{d'\}} \| \leq \|
g\|_{l^1}$ (see \eqref{FG} and \eqref{drop2}),  and since $|z| \leq
3|v|$, by Proposition \ref{p:order},
$$ \left \| \frac{\partial^{n+m} H_k^{-1}}{\partial k_1^n \partial
k_2^m} \right \| \le \eps^{-(n+m+1)} C_{\Lambda,A, n,m} \frac{1}{|z|}.
$$
Now recall that $G'_2 = \{ b \in G' \; | \; |b-d'| > \tfrac{1}{4} R \}$.
Then,
\begin{align*}
\sup_{b \in \{d'\}} \sum_{c \in G'_2} |f(b-c)| & \leq \sum_{c \in G'_2}
\frac{|d'-c|^2}{|d'-c|^2} |f(d'-c)| \leq \| b^2 f(b) \|_{l^1} \sup_{c
\in G'_2} \frac{1}{|d'-c|^2} \leq \frac{16}{R^2} \| b^2 f(b) \|_{l^1},
\\
\sup_{c \in G'_2} \sum_{b \in \{d'\}} |f(b-c)| & \leq \sup_{c \in G'_2}
\frac{|d'-c|^2}{|d'-c|^2} |f(d'-c)| \leq \| b^2 f(b) \|_{l^1} \sup_{c
\in G'_2} \frac{1}{|d'-c|^2} \leq \frac{16}{R^2} \| b^2 f(b) \|_{l^1}.
\end{align*}
Hence, by Proposition \eqref{p:ineq},
$$ \| \mathcal{F}_{\{d'\} G'_2} \| \leq 16 \| b^2 f(b) \|_{l^1}
\frac{1}{R^2}. $$
Similarly,
$$ \| \mathcal{G}_{G'_2 \{d'\}} \| \leq 16 \| b^2 f(b) \|_{l^1}
\frac{1}{R^2}. $$
Therefore, combining all this, for $1 \le j \le 2$ we obtain
$$ \left| \frac{\partial^{n+m}}{\partial k_1^n \partial k_2^m}
\mathcal{R}_j(k) \right| \leq \eps^{-(n+m+1)} C_{\Lambda, A,n,m,f,g}
\frac{1}{|z| R^2}. $$

{\tt Step 2.} Recall from \eqref{RR4} the expression for
$\mathcal{R}_4$. Then, similarly as above, by applying Proposition
\ref{p:dSWZ} for $j=3$ we find that
$$ \left| \frac{\partial^{n+m}}{\partial k_1^n \partial k_2^m}
\mathcal{R}_4(k) \right| \leq \| \mathcal{F}_{\{d'\} G'_1} \| \left\|
\frac{\partial^{n+m}}{\partial k_1^n \partial k_2^m} \Delta_k^{-1} Z
\right \| \| \mathcal{G}_{G'_1 \{d'\}} \| \leq \| f \|_{l^1} \| g
\|_{l^1} C_{\Lambda,A,q,n,m} \frac{1}{|z|_R^3}. $$

{\tt Step 3.} To bound the derivatives of $\mathcal{R}_3$ (which is
given by \eqref{RR3}) we need a few more estimates. Recall from
\eqref{Tj} that $W^{(j)}_{43} = \pi_{G'_4} T_{G'G'}^{j+1} \pi_{G'_3}$.
First observe that
$$ \frac{\partial^{r+p}}{\partial k_1^r \partial k_2^p} \, \pi_{G'_1}
\Delta_k^{-1} T_{33}^{m} T_{34} W^{(j-m-1)}_{43} =
\frac{\partial^{r+p}}{\partial k_1^r \partial k_2^p} \, \Delta_k^{-1}
\pi_{G'_1} T_{33}^{m} T_{34} T_{G'G'}^{j-m} \pi_{G'_3} $$
is given by a sum of $(j+2)^{r+p}$ terms of the form
$$ \frac{\partial^{l_1 + n_1} \Delta_k^{-1}}{\partial k_1^{l_1} \partial
k_2^{n_1}} \; \pi_{G'_1} \frac{\partial^{l_2 + n_2} T_{33}}{\partial
k_1^{l_2} \partial k_2^{n_2}} \; \cdots \; \frac{\partial^{l_{m+2} +
n_{m+2}} T_{34}}{\partial k_1^{l_{m+2}} \partial k_2^{n_{m+2}}} \;
\frac{\partial^{l_{m+3} + n_{m+3}} T_{G'G'}}{\partial k_2^{l_{m+3}}
\partial k_2^{n_{m+3}}} \cdots \frac{\partial^{l_{j+2} + n_{j+2}}
T_{G'G'}}{\partial k_1^{l_{j+2}} \partial k_2^{n_{j+2}}} \pi_{G'_3}. $$
Moreover, for each term in the sum we have $\sum_{i=1}^{j+2} l_i = r$
and $\sum_{i=1}^{j+2} n_i = p$. Thus,
\begin{equation} \label{above3}
\left \| \frac{\partial^{r+p}}{\partial k_1^r \partial k_2^p} \,
\pi_{G'_1} \Delta_k^{-1} T_{33}^{m} T_{34} W^{(j-m-1)}_{43} \right \|
\leq (j+2)^{r+p} \sup_{\mathcal{I}'} \left \| \left( \prod_{i=1}^{j+2}
\frac{\partial^{l_i + n_i} T_{(i)}}{\partial k_1^{l_i} \partial
k_2^{n_i}} \right) \pi_{G'_3} \right \|,
\end{equation}
where the set $\mathcal{I}'$ is given by \eqref{I} with $j$ replaced by
$j+2$ and
\begin{equation} \label{T(i)}
T_{(i)} \coloneqq \begin{cases} \Delta_k^{-1} \pi_{G'_1} & \text{for}
\quad i=1, \\ T_{33} & \text{for} \quad 2 \le i \le m+1, \\ T_{34} &
\text{for} \quad i=m+2, \\ T_{G'G'} & \text{for} \quad m+3 \le i \le
j+2. \end{cases}
\end{equation}

{\tt Step 3a.} The first step in bounding \eqref{above3} is to estimate
$\left \| \frac{\partial^{r+p} \Delta_k^{-1}}{\partial k_1^r \partial
k_2^p} \pi_{G'_1} \right \|$. We follow the same argument that we have
used in the proof of Lemma \ref{l:der} to bound $\left\|
\frac{\partial^{n+m} H_k^{-1}}{\partial k_1^n \partial k_2^m} \right
\|$. In fact, in view of \eqref{rule} one can see that
\begin{equation} \label{sum1}
\frac{\partial^p \Delta_k^{-1}}{\partial k_2^p} =
\sum_{\substack{\text{finite sum} \\ \text{where \# of terms} \\
\text{depend on $p$}}} \left[ \prod_{j=1}^p \Delta_k^{-1}
\frac{\partial^{n_j} \Delta_k}{\partial k_2^{n_j}} \right]
\Delta_k^{-1},
\end{equation}
where $\sum_{j=1}^p n_j =p$. Hence, when we compute
$\frac{\partial^r}{\partial k_1^r} \frac{\partial^p \Delta_k^{-1}
}{\partial k_2^p}$, the derivative $\frac{\partial^r}{\partial k_1^r}$
acts either on $\Delta_k^{-1}$ or
$\frac{\partial^{n_j}\Delta_k}{\partial k_2^{n_j}}$. However, since
$\left( \frac{\partial \Delta_k}{\partial k_2} \right)_{b,c} =
2(k_2+c_2) \delta_{b,c}$, we have $\frac{\partial^r}{ \partial k_1^r}
\frac{\partial^{n_j}}{\partial k_2^{n_j}} \Delta_k = 0$ if $n_j \geq 1$
and $\frac{\partial^r}{ \partial k_1^r} \frac{\partial^{n_j}}{\partial
k_2^{n_j}} \Delta_k = \frac{\partial^r}{ \partial k_1^r} \Delta_k$ if
$n_j=0$. Similarly, using again \eqref{rule} one can see that
$\frac{\partial^r \Delta_k^{-1}}{\partial k_1^r}$ is given by a finite
sum as in \eqref{sum1}, with $p$ and $k_2$ replaced by $r$ and $k_1$,
respectively, and $\sum_{j=1}^r n_j =r$. Thus, combining all this we
conclude that
\begin{equation} \label{drpd}
\frac{\partial^{r+p} \Delta_k^{-1}}{\partial k_1^r \partial k_2^p} =
\sum_{\substack{\text{finite sum where} \\ \text{\# of terms depend} \\
\text{on $r$ and $p$}}} \left[ \prod_{j=1}^{r+p} \Delta_k^{-1}
\frac{\partial^{n_j} \Delta_k}{\partial k_{i_j}^{n_j}} \right]
\Delta_k^{-1},
\end{equation}
where $\sum_{j=1}^{r+p} n_j \delta_{2,i_j}=p$ and $\sum_{j=1}^{r+p} n_j
\delta_{1,i_j}=r$.
If we observe that
$$ \left( \frac{\partial^{n_j} \Delta_k}{\partial k_{i_j}^{n_j}}
\right)_{b,c} = \begin{cases} 2(k_{i_j}+c_{i_j}) \delta_{b,c} &
\text{if} \quad n_j=1, \\ 2 \delta_{b,c} & \text{if} \quad n_j=2, \\ 0 &
\text{if} \quad n_j \ge 3, \end{cases} $$
and extract the ``leading term'' from the summation in \eqref{drpd}, in
a sense that will be clear below, we can rewrite \eqref{drpd} in terms
of matrix elements as
\begin{align*}
\frac{\partial^{r+p}}{\partial k_1^r \partial k_2^p} \, \frac{1}{N_c(k)}
& = \frac{(-1)^{r+p}(r+p)!}{N_c(k)} \left[ \frac{2(k_1+c_1)}{N_c(k)}
\right]^r \left[\frac{2(k_2+c_2)}{N_c(k)} \right ]^p \\ & \qquad \qquad
+ \sum_{\substack{\text{finite sum where} \\ \text{\# of terms depend}
\\ \text{ on $r$ and $p$}}} \frac{(2(k_1+c_1))^{\alpha_j}
(2(k_2+c_2))^{\beta_j}}{N_c(k)^{r+p+1}}, 
\end{align*}
where $\alpha_j + \beta_j < r+p$ for every $j$ in the summation. Recall
from \eqref{1/N} and \eqref{tildeb} that, for all $c \in G' \setminus \{
\tilde{c} \}$,
\begin{equation} \label{k/N}
\frac{|k_i+c_i|}{|N_c(k)|} \leq \frac{2}{\Lambda} < \frac{1}{3 \eps} <
\frac{7}{2 \eps} \qquad \text{and} \qquad
\frac{|k_i+\tilde{c}_i|}{|N_{\tilde{c}}(k)|} \leq \frac{\Lambda +
3|v|}{\eps |v|} \leq \frac{7}{2 \eps}.
\end{equation}
Hence,
\begin{equation} \label{d1/N}
\begin{split}
\left| \frac{\partial^{r+p}}{\partial k_1^r \partial k_2^p} \,
\frac{1}{N_c(k)} \right | & \leq \frac{(r+p)!}{|N_c(k)|} \left(
\frac{7}{\eps} \right)^{r+p} + \sum_{\substack{\text{finite sum where}
\\ \text{\# of terms depend} \\ \text{ on $r$ and $p$}}} \left(
\frac{7}{\eps} \right)^{\alpha_j+\beta_j} \frac{1}{|N_c(k)|^2} \\ 
& \leq \frac{(r+p)!}{|N_c(k)|} \left( \frac{7}{\eps} \right)^{r+p} +
C_{\eps,r,p} \, \frac{1}{|N_c(k)|^2}.
\end{split}
\end{equation}
Thus, by Proposition \ref{p:ineq}, since $|N_c(k)| \geq \eps |v| \geq
\eps |z|/3$ for all $c \in G'$, we have
\begin{equation} \label{ddk}
\left \| \frac{\partial^{r+p} \Delta_k^{-1}}{\partial k_1^r \partial
k_2^p} \pi_{G'_1} \right \| \leq \frac{7^{r+p} (r+p)!}{\eps^{r+p+1}} \,
\frac{3}{|z|} + \frac{C_{\eps,r,p}}{|z|^2}.
\end{equation}
Now, let $\rho_1 = \rho_{1;\eps,r,p}$ be the constant
$$ \rho_{1;\eps,r,p} \coloneqq \max_{\substack{l_1 \le r \\ n_1 \le p}}
\; \frac{\eps^{l_1+n_1+1} C_{\eps,l_1,n_1}}{4 (l_1+n_1)! \,
7^{l_1+n_1}}, $$
where $C_{\eps,l_1,n_1}$ is the constant in \eqref{ddk}. Then, for $|z|
> \rho_1$ and for any $l_1 \le r$ and any $n_1 \le p$,
\begin{equation} \label{delta}
\begin{aligned}
\left \| \frac{\partial^{l_1+n_1} \Delta_k^{-1}}{\partial k_1^{l_1}
\partial k_2^{n_1}} \pi_{G'_1} \right \| & \leq \frac{7^{l_1+n_1}
(l_1+n_1)!}{\eps^{l_1+n_1+1}} \, \frac{3}{|z|} + \frac{7^{l_1+n_1}
(l_1+n_1)!}{\eps^{l_1+n_1+1}} \, \frac{4}{|z|} = (l_1+n_1)! \left(
\frac{7}{\eps} \right)^{l_1+n_1+1} \frac{1}{|z|}.
\end{aligned}
\end{equation}
This is the first inequality we need to bound \eqref{above3}. We next
estimate the other factors in that expression.

{\tt Step 3b.} Recall from \eqref{TXY0} that
$$ T_{b,c} = \frac{1}{N_c(k)}( 2(c+k) \cdot \hat{A}(b-c) - \hat{q}(b-c)
). $$
By direct calculation we have
\begin{align*}
\frac{\partial^{r+p} \, T_{b,c}}{\partial k_1^r \partial k_2^p} & =
\left( \frac{\partial^{r+p}}{\partial k_1^r \partial k_2^p} \,
\frac{1}{N_c(k)} \right) ( 2(c+k) \cdot \hat{A}(b-c) - \hat{q}(b-c)) \\
& \qquad + r \left( \frac{\partial^{r-1+p}}{\partial k_1^{r-1} \partial
k_2^p} \, \frac{1}{N_c(k)} \right) 2 \hat{A}_j(b-c) + p \left(
\frac{\partial^{r+p-1}}{\partial k_1^{r} \partial k_2^{p-1}} \,
\frac{1}{N_c(k)} \right) 2 \hat{A}_j(b-c).
\end{align*}
Hence, using \eqref{k/N} and \eqref{d1/N}, since $|N_c(k)| \geq \eps|v|
\geq \eps|z|/3$ for all $c \in G'$ and $|v| > 1$,
\begin{equation} \label{drpTbc}
\begin{split}
\left| \frac{\partial^{r+p} \, T_{b,c}}{\partial k_1^r \partial k_2^p}
\right| & \leq \left( (r+p)! \left( \frac{7}{\eps} \right)^{r+p} +
\frac{C_{\eps,r,p}}{\eps |v|} \right) \left( \frac{7}{\eps}
|\hat{A}(b-c)| + \frac{|\hat{q}(b-c)|}{\eps |v|} \right) +
\frac{C_{\eps,r,p}}{|v|} |\hat{A}(b-c)| \\
& \leq (r+p)! \left( \frac{7}{\eps} \right)^{r+p+1} |\hat{A}(b-c)| +
\frac{C_{\eps,r,p}}{|z|} (|\hat{A}(b-c)| + |\hat{q}(b-c)|). 
\end{split}
\end{equation}
Therefore, by Proposition \ref{p:ineq},
\begin{equation} \label{Theta}
\left \| \frac{\partial^{r+p} \, T_{G'G'}}{\partial k_1^r \partial
k_2^p} \right\| \leq \Theta_{r,p},
\end{equation}
where
\begin{equation} \label{Theta2}
\Theta_{r,p} \coloneqq (r+p)! \left( \frac{7}{\eps} \right)^{r+p+1} \|
\hat{A} \|_{l^1} + C_{\eps,A,q,r,p} \frac{1}{|z|}.
\end{equation}
This is the second estimate we need to bound \eqref{above3}. We next
derive one more inequality.

{\tt Step 3c.} Set
\[ Q^{r,p}_{b,c} \coloneqq (1+|b-c|^2) \frac{\partial^{r+p} \,
T_{b,c}}{\partial k_1^r \partial k_2^p}. \]
We first prove that, for any $B,C \subset G'$,
$$ \sup_{b \in B} \sum_{c \in C} |Q^{r,p}_{b,c}| \leq \Omega_{r,p}
\qquad \text{and} \qquad \sup_{c \in C} \sum_{b \in B} |Q^{r,p}_{b,c}|
\leq \Omega_{r,p}, $$
where
\begin{equation} \label{Omega}
\Omega_{r,p} \coloneqq (r+p)! \left( \frac{7}{\eps} \right)^{r+p+1} \|
(1+b^2) \hat{A}(b) \|_{l^1} + C_{\eps,A,q,r,p} \frac{1}{|z|}.
\end{equation}
In fact, in view of \eqref{drpTbc} we have
\begin{align*}
\sup_{b \in B} \sum_{c \in C} |Q^{r,p}_{b,c}| & = \sup_{b \in B} \sum_{c
\in C} (1+|b-c|^2) \left| \frac{\partial^{r+p} T_{b,c}}{\partial k_1^r
\partial k_2^p} \right| \\
& \quad \leq \sup_{b \in B} \sum_{c \in C} (1+|b-c|^2) \\ 
& \qquad \quad \times \left[ (r+p)! \left( \frac{7}{\eps}
\right)^{r+p+1} |\hat{A}(b-c)| + \frac{C_{\eps,r,p}}{|z|}
(|\hat{A}(b-c)| + |\hat{q}(b-c)|) \right] \\
& \quad \leq (r+p)! \left( \frac{7}{\eps} \right)^{r+p+1} \| (1+b^2)
\hat{A}(b) \|_{l^1} + C_{\eps,A,q,r,p} \frac{1}{|z|},
\end{align*}
and similarly we estimate $\sup_{c \in C} \sum_{b \in B}
|Q^{r,p}_{b,c}|$. Now observe that, as in \eqref{bminusc}, for any
integer $m \geq 0$ and for any $\xi_0,\xi_1,\dots,\xi_{m+2} \in \Gdual$,
let $b = \xi_0$ and $c = \xi_{m+2}$. Then,
$$ |b-c|^2 \leq 2(m+2) \sum_{i=1}^{m+2} |\xi_{i-1}-\xi_i|^2. $$
To simplify the notation write $\partial^{l_i,n_i} = \frac{\partial^{l_i
+ n_i}}{\partial k_1^{l_i} \partial k_2^{n_i}}$,
and recall from \eqref{T(i)} and \eqref{Omega} the definition of
$T_{(i)}$ and $\Omega_{r,p}$. Hence, similarly as in the proof of
Proposition \ref{p:decayT}, since $|b-c| \geq R/4$ for all $b \in G'_1$
and $c \in G'_4$,
\begin{align*}
& \sup_{b \in G'_1} \sum_{c \in G'_4} \left| \left( \prod_{i=2}^{m+2}
\partial^{l_i,n_i} T_{(i)} \right)_{b,c} \right| \leq \sup_{\substack{b
\in G'_1 \\ c \in G'_4}} \frac{1}{1 + |b-c|^2} \,\, \sup_{b \in G'_1}
\sum_{c \in G'_4} (1+|b-c|^2) \left| \left( \prod_{i=2}^{m+2}
\partial^{l_i,n_i} T_{(i)} \right)_{b,c} \right| \\
& \leq \frac{2(m+2)}{1 + \frac{1}{16} R^2} \,\, \sup_{b \in G'_1} \,
\sum_{\xi_1 \in G'_3} (1+|b-\xi_1|^2) \left| \partial^{l_2,n_2}
T_{b,\xi_1} \right| \\ 
& \qquad \times \sum_{\xi_2 \in G'_3} (1+|\xi_1-\xi_2|^2) \left|
\partial^{l_3,n_3} T_{\xi_1,\xi_2} \right| \cdots \sum_{c \in G'_4}
(1+|\xi_{m+1} -c|^2) \left| \partial^{l_{m+2},n_{m+2}}
T_{\xi_{m+1},c}\right| \\
& \leq \frac{2(m+2)}{1 + \frac{1}{16} R^2} \,\, \sup_{b \in G'_1} \,
\sum_{\xi_1 \in G'_3} (1+|b-\xi_1|^2) \left| \partial^{l_2,n_2}
T_{b,\xi_1} \right| \sup_{\xi_1 \in G'_3} \sum_{\xi_2 \in G'_3}
(1+|\xi_1-\xi_2|^2) \left| \partial^{l_3,n_3} T_{\xi_1,\xi_2} \right| \\
& \qquad \times \sup_{\xi_{m+1} \in G'_3} \sum_{c \in G'_4}
(1+|\xi_{m+1} -c|^2) \left| \partial^{l_{m+2},n_{m+2}} T_{\xi_{m+1},c}
\right| \\
& = \frac{2(m+2)}{1 + \frac{1}{16} R^2} \,\, \sup_{b \in G'_1} \,
\sum_{\xi_1 \in G'_3} |Q^{l_2,n_2}_{b,\xi_1}| \cdots \sup_{\xi_{m+1} \in
G'_3} \sum_{c \in G'_4} |Q^{l_{m+2},n_{m+2}}_{\xi_{m+1},c}| \leq
\frac{2(m+2)}{1 + \frac{1}{16} R^2} \prod_{i=2}^{m+2} \Omega_{l_i,n_i}
\end{align*}
and similarly
$$ \sup_{c \in G'_4} \sum_{b \in G'_1} \left| \left( \prod_{i=2}^{m+2}
\partial^{l_i,n_i} T_{(i)} \right)_{b,c} \right| \leq \frac{2(m+2)}{1 +
\frac{1}{16} R^2} \prod_{i=2}^{m+2} \Omega_{l_i,n_i}. $$
Therefore, by Proposition \ref{p:ineq},
$$ \left\| \pi_{G'_1} \prod_{i=2}^{m+2} \frac{\partial^{l_i+n_i}
T_{(i)}}{\partial k_1^{l_i} \partial k_2^{n_i}} \right \| \leq
\frac{2(m+2)}{1 + \frac{1}{16} R^2} \prod_{i=2}^{m+2} \Omega_{l_i,n_i}.
$$
We have all we need to bound \eqref{above3}.

{\tt Step 3d.} From \eqref{Theta} and \eqref{delta} it follows that
$$ \left\| \prod_{i=m+3}^{j+2} \frac{\partial^{l_i+n_i}
T_{(i)}}{\partial k_1^{l_i} \partial k_2^{n_i}} \right \| \leq
\prod_{i=m+3}^{j+2} \Theta_{l_i,n_i} $$
and
$$ \left\| \frac{\partial^{l_1+n_1} T_{(1)}}{\partial k_1^{l_1} \partial
k_2^{n_1}} \right \| \leq (r+p)! \left( \frac{7}{\eps} \right)^{r+p+1}
\frac{1}{|z|}. $$
Thus, recalling \eqref{above3} we get
\begin{align*}
& \left \| \frac{\partial^{r+p}}{\partial k_1^r \partial k_2^p} \,
\Delta_k^{-1} \pi_{G'_1} T_{33}^{m} T_{34} W^{(j-m-1)}_{43} \right \|
\leq (j+2)^{r+p} \sup_{\mathcal{I}'} \left \| \left( \prod_{i=1}^{j+2}
\frac{\partial^{l_i + n_i} T_{(i)}}{\partial k_1^{l_i} \partial
k_2^{n_i}} \right) \pi_{G'_3} \right \| \\ 
& \qquad \le (j+2)^{r+p} \, \sup_{\mathcal{I}'} \left\{ \frac{1}{|z|} \,
\frac{2(m+2)}{1 + \frac{1}{16} R^2} (r+p)! \left( \frac{7}{\eps}
\right)^{r+p+1} \left[ \prod_{i=2}^{m+2} \Omega_{l_i,n_i} \right]
\prod_{i=m+3}^{j+2} \Theta_{l_i,n_i} \right\} \\ 
& \qquad \le (j+2)^{r+p} (m+2) \, \frac{C}{|z|R^2} \,
\sup_{\mathcal{I}'} \left\{ (l_1+n_1)! \left( \frac{7}{\eps}
\right)^{l_1+n_1+1} \left[ \prod_{i=2}^{m+2} \Omega_{l_i,n_i} \right]
\prod_{i=m+3}^{j+2} \Theta_{l_i,n_i} \right\},
\end{align*}
where $C$ is an universal constant. Now, recall the definition of
$\Theta_{r,p}$ and $\Omega_{r,p}$ in \eqref{Theta2} and \eqref{Omega},
observe that $\| \hat{A} \|_{l^1} < \| (1+b^2) \hat{A} \|_{l^1}$, and
let $\rho_2 = \rho_{2;\eps,A,q,r,p}$ be a sufficiently large constant
such that, for $|z| > \rho_2$ and for any $l_i \leq r$ and any $n_i \leq
p$,
$$ \Theta_{l_i,n_i}, \; \Omega_{l_i,n_i} \leq 2 (l_i+n_i)! \left(
\frac{7}{\eps} \right)^{l_i+n_i+1} \| (1+b^2) \hat{A}(b) \|_{l^1}. $$
Then,
\begin{align*}
& \left \| \frac{\partial^{r+p}}{\partial k_1^r \partial k_2^p} \,
\Delta_k^{-1} \pi_{G'_1} T_{33}^{m} T_{34} W^{(j-m-1)}_{43} \right \| \\
& \le (j+2)^{r+p} (m+2) \frac{C}{|z| R^2} \, \sup_{\mathcal{I}'} \left\{
(l_1+n_1)! \left( \frac{7}{\eps} \right)^{l_1+n_1+1} \left[
\prod_{i=2}^{m+2} \Omega_{l_i,n_i} \right] \prod_{i=m+3}^{j+2}
\Theta_{l_i,n_i} \right\} \\
& \le (j+2)^{r+p} \frac{(m+2)C}{|z| R^2} (2 \| (1+b^2) \hat{A}(b)
\|_{l^1})^{j+1} \left( \frac{7}{\eps} \right)^{j+2} \sup_{\mathcal{I}'}
\left \{ \left( \frac{7}{\eps} \right)^{\sum_{i=1}^{j+2} (l_i+n_i)}
\prod_{i=1}^{j+2} \, (l_i+n_i)! \right \} \\
& \text{(since $\textstyle \sum_{i=1}^{j+2} l_i = r$, $\textstyle
\sum_{i=1}^{j+2} n_i = p$ and $\textstyle \prod_{i=1}^{j+2} \,
(l_i+n_i)! < (r+p)!$)} \\ 
& \leq C (r+p)! \left( \frac{7}{\eps} \right)^{r+p+1} (m+2) (j+2)^{r+p}
\left( \frac{14}{\eps} \| (1+b^2) \hat{A}(b) \|_{l^1} \right)^{j+1}
\frac{1}{|z| R^2} \\ 
& \leq \frac{C_{\eps,r,p}}{|z| R^2} (m+2) (j+2)^{r+p} \left( \frac{4}{9}
\right)^{j+1},
\end{align*}
since $\| (1+b^2) \hat{A}(b) \|_{l^1} < 2\eps/63$. This establishes a
bound for \eqref{above3}.

{\tt Step 4.} We now apply the last inequality for deriving an estimate
for the derivatives of $\mathcal{R}_3$ and complete the proof of the
lemma for $j=3$. Recall from \eqref{itera} that 
$$ X^{(j)}_{33} = \sum_{m=0}^{j-1} T_{33}^m T_{34} W^{(j-m-1)}_{43}. $$
Then,
\begin{align*}
& \left \| \frac{\partial^{r+p}}{\partial k_1^r \partial k_2^p}
\pi_{G'_1} \Delta_k^{-1} X^{(j)}_{33} \right \| \leq \sum_{m=0}^{j-1}
\left \| \frac{\partial^{r+p}}{\partial k_1^r \partial k_2^p}
\Delta_k^{-1} \pi_{G'_1} T_{33}^m T_{34} W^{(j-m-1)}_{43} \right \| \\
& \qquad \leq \sum_{m=0}^{j-1} \frac{C_{\eps,r,p}}{|z| R^2} (m+2)
(j+2)^{r+p} \left( \frac{4}{9} \right)^{j+1} \leq
\frac{C_{\eps,r,p}}{|z|R^2} (j+2)^{r+p} \left( \frac{4}{9} \right)^{j+1}
\sum_{m=0}^{j-1} (m+2) \\
& \qquad = \frac{C_{\eps,r,p}}{|z| R^2} (j+2)^{r+p} \left( \frac{4}{9}
\right)^{j+1} \frac{1}{2}(j^2 + 3j).
\end{align*}
Thus, since $G'_1 \subset G'_3$,
\begin{align*}
& \left \| \pi_{G'_1} \frac{\partial^{r+p}}{\partial k_1^r \partial
k_2^p} \left[   \Delta_k^{-1}\sum_{j=1}^\infty X^{(j)}_{33} \right]
\pi_{G'_1} \right \| \leq \sum_{j=1}^\infty \left \|
\frac{\partial^{r+p}}{\partial k_1^r \partial k_2^p} \pi_{G'_1}
\Delta_k^{-1} X^{(j)}_{33} \right \| \\
& \qquad \qquad \qquad \leq \frac{C_{\eps,r,p}}{|z| R^2}
\sum_{j=1}^\infty (j+2)^{r+p}  \left( \frac{4}{9} \right)^{j+1}
\frac{1}{2}(j^2 + 3j) \leq C C_{\eps,r,p} \frac{1}{|z| R^2},
\end{align*}
where $C$ is an universal constant. Therefore,
\begin{align*}
\left| \frac{\partial^{r+p}}{\partial k_1^r \partial k_2^p}
\mathcal{R}_3(k) \right| & = \left| \mathcal{F}_{ \{d'\} G'_1}
\frac{\partial^{r+p}}{\partial k_1^r \partial k_2^p} \left[
\Delta_k^{-1} \sum_{j=1}^\infty X^{(j)}_{33} \right] \mathcal{G}_{G'_1
\{d'\}} \right| \\
& \leq \| \mathcal{F}_{ \{d'\} G'_1} \| \left \| \pi_{G'_1}
\frac{\partial^{r+p}}{\partial k_1^r \partial k_2^p} \left[
\Delta_k^{-1}\sum_{j=1}^\infty X^{(j)}_{33} \right] \pi_{G'_1} \right \|
\| \mathcal{G}_{G'_1 \{d'\}} \| \\
& \leq C C_{\eps,r,p} \| f \|_{l^1} \| g \|_{l^1} \frac{1}{|z| R^2}.
\end{align*}
Finally, combining all the estimates we have
\begin{align*}
\left \| \frac{\partial^{n+m}}{\partial k_1^n \partial k_2^m}
\alpha^{(3)}_{\mu,d'}(k) \right \| & \leq \sum_{j=1}^4 \left \|
\frac{\partial^{n+m}}{\partial k_1^n \partial k_2^m} \mathcal{R}_j(k)
\right \| \\
& \leq 3 \, \frac{C}{|z|R^2} + \frac{C}{|z|_R^3} 
\leq \frac{4C}{|z| R^2},
\end{align*}
where $C = C_{\eps,\Lambda,A,q,f,g,m,n}$ is a constant. Set
$\rho_{\eps,A,q,m,n} \coloneqq \max \{ \rho_{1;\eps,m,n}, \;
\rho_{2;\eps,A,q,m,n}\}$. The proof of the lemma for $j=3$ is complete.
\end{proof}


\end{document}